\documentclass[12pt]{article}

\newcommand{\zerob} {{\bf 0}}

\newcommand{\nub} {{\boldsymbol{\nub}}}

\newcommand{\Gammamat} {{\boldsymbol{\it \Gamma}}}
\renewcommand{\Gammamat} {{\boldsymbol{\Gamma}}}

\renewcommand{\zerob}{\mathbf{0}}

\renewcommand{\d}{\mathrm{d}}

\newcommand{\epsilonb}{\boldsymbol{\epsilon}}

\usepackage{multirow}
\usepackage{setspace}
\usepackage{subcaption}

\newcommand{\BayesRisk}{r}
\newcommand{\priordensity}{\pi}
\newcommand{\priordensitypower}{\pi_m}
\newcommand{\EMestimate}{\hat{\vec{\theta}}}{}
{}
\newcommand{\GHscaleparam}{\phi}

\newcommand{\GHsqrtterm}{\sqrt{\{\omega\eta + (\vec{z} - \vec{\mu})' \vec{\Sigma}^{-1}(\vec{z} - \vec{\mu})\}(\omega/\eta + \vec{\gamma}'\vec{\Sigma}^{-1}\vec{\gamma})}}

\newcommand{\convergencecaption}[1]{EM NBE sequences (black dot-dash lines) and post-burn-in averaged sequences (red solid lines) for different initial values and different Monte Carlo sample sizes $m$, for the #1.}

\newcommand{\supptitle}{\LARGE Supplementary Material for ``\Paste{title}''}

\usepackage{standalone} 
\usepackage{clipboard}
\usepackage{xr}

\newcommand{\supplement}{%
	\setcounter{section}{0}%
	\renewcommand{\thesection}{S\arabic{section}}%
	\renewcommand{\theHsection}{supp.S\arabic{section}}%
	\setcounter{subsection}{0}%
	\setcounter{table}{0}%
	\renewcommand{\thetable}{S\arabic{table}}%
	\setcounter{figure}{0}%
	\renewcommand{\thefigure}{S\arabic{figure}}%
	\setcounter{equation}{0}%
	\renewcommand{\theequation}{S\arabic{equation}}%
	\setcounter{algorithm}{0}%
	\renewcommand{\thealgorithm}{S\arabic{algorithm}}%
}

\usepackage{etoolbox}
\newbool{arxiv}
\newbool{blind}
\newbool{doublespace}

\setbool{arxiv}{true} 
\setbool{blind}{false} 
\setbool{doublespace}{false}

\newcommand{\reff}[1]{\ifbool{arxiv}{\ref{#1}}{\ref*{#1}}}  
\newcommand{\eqreff}[1]{\ifbool{arxiv}{\eqref{#1}}{(\ref*{#1})}}  

\newcommand{\reffsupp}[1]{\reff{#1} of the Supplementary Material} 
\newcommand{\reffmain}[1]{\reff{#1} of the main text}

\newcommand{\eqreffmain}[1]{Equation~\eqreff{#1} of the main text} 

\usepackage{authblk}

\usepackage{scalerel} 
\usepackage[top=1in, bottom=1in, left=1in, right=1in]{geometry}
\usepackage{caption} 
\usepackage{graphicx}
\usepackage[authoryear,comma,sectionbib,semicolon]{natbib}
\usepackage{bibunits}
\usepackage{placeins} 
\usepackage{amsmath}	
\usepackage{amsfonts}	
\usepackage{commath}
\usepackage[shortlabels]{enumitem}
\setlist[enumerate,1]{label={(\roman*)}}
\usepackage{amsthm}
\usepackage{tikz}
\usetikzlibrary{positioning}
\usepackage{hyperref}
\usepackage{xcolor}
\definecolor{Red}{rgb}{0.5,0,0}
\definecolor{Blue}{rgb}{0,0,0.5}
\hypersetup{%
	colorlinks = {true},
	linktocpage = {true},
	plainpages = {false},
	linkcolor = {Blue},
	citecolor = {Blue},
	urlcolor = {Red},
	pdfstartview = {Fit},
	pdfpagemode = {UseOutlines},
	pdfview = {XYZ null null null}
}

\usepackage{algpseudocode,algorithm,algorithmicx}

\usepackage{tabularx}
\makeatletter
\newcommand{\multiline}[1]{%
	\begin{tabularx}{\dimexpr\linewidth-\ALG@thistlm}[t]{@{}X@{}}
		#1
	\end{tabularx}
}
\makeatother

\newcommand{\proglang}[1]{\texttt{#1}}
\newcommand{\pkg}[1]{{\fontseries{b}\selectfont #1}}

\newcommand{\red}[1]{\textcolor{red}{#1}}

\def\mbf#1{{
		\mathchoice
		{\hbox{\boldmath$\displaystyle{#1}$}}
		{\hbox{\boldmath$\textstyle{#1}$}}
		{\hbox{\boldmath$\scriptstyle{#1}$}}
		{\hbox{\boldmath$\scriptscriptstyle{#1}$}}
}}
\def\vec{\mbf}

\def\d{\textrm{d}} 
\DeclareMathOperator*{\argmax}{arg\,max} 
\DeclareMathOperator*{\argmin}{arg\,min} 

\newcommand{\Gau}{{\text{Gau}}}
\newcommand{\Unif}[2]{{\text{Unif}}(#1, #2)}

\newcommand{\E}{\mathbb{E}} 
\newcommand{\var}{{\rm{var}}} 
\newcommand{\V}{\mathbb{V}} 
\newcommand{\Var}{\mathbb{V}} 
\newcommand{\cov}[2]{{\rm cov\!}\left(#1,\, #2\right)} 

\newcommand{\tp}{{\!\scriptscriptstyle \top}}

\newcommand{\ifootnote}[1]{{\ifthenelse{\isodd{2}}{\footnote{#1}}{}}}
\newcommand{\comment}[1]{{\ifthenelse{\isodd{1}}{\footnote{\red{#1}}}{}}}

\makeatletter
\newenvironment{proof*}[1][\proofname]{\par
	\pushQED{\qed}%
	\normalfont \partopsep=\z@skip \topsep=\z@skip
	\trivlist
	\item[\hskip\labelsep
	\itshape
	#1\@addpunct{.}]\ignorespaces
}{%
	\popQED\endtrivlist\@endpefalse
}
\makeatother
\newclipboard{quotes}
\ifbool{doublespace}{\doublespacing}{\singlespacing}

\title{\LARGE \Copy{title}{Neural Parameter Estimation with Incomplete Data}} 

\ifbool{blind}{%
\Copy{authors}{%
        \author{}}%
}{%
\Copy{authors}{%
	\author[1]{Matthew Sainsbury-Dale}%
	\author[2,3]{Andrew Zammit-Mangion}%
	\author[3]{Noel Cressie}%
	\author[1]{Rapha\"el Huser}%
	\affil[1]{Statistics Program, Computer, Electrical and Mathematical Sciences and Engineering Division, King Abdullah University of Science and Technology (KAUST), Saudi Arabia}%
  \affil[2]{School of Mathematics and Statistics, University of New South Wales, Australia}%
	\affil[3]{School of Mathematics and Applied Statistics, University of Wollongong, Australia}%
}
}

\date{}

\begin{document}

\newtheorem{theorem}{Theorem}
\newtheorem{proposition}{Proposition}
\newtheorem{conjecture}{Conjecture}
\newtheorem{corollary}{Corollary}

\begin{singlespace}
\maketitle
\begin{abstract}

\noindent Advances in artificial intelligence (AI) and deep learning have led to neural networks being used to generate lightning-speed answers to complex science questions, paintings in the style of Monet, or stories like those of Twain. Leveraging their computational speed and flexibility, neural networks are also being used to facilitate fast, likelihood-free statistical inference. However, it is not straightforward to use neural networks with data that for various reasons are incomplete, which precludes their use in many applications. A recently proposed approach to remedy this issue uses an appropriately padded data vector and a vector that encodes the missingness pattern as input to a neural network. While computationally efficient, this ``masking'' approach is not robust to the missingness mechanism and can result in statistically inefficient inferences. Here, we propose an alternative approach that is based on the Monte Carlo expectation-maximization (EM) algorithm. Our EM approach is likelihood-free, substantially faster than the conventional EM algorithm as it does not require numerical optimization at each iteration, and more statistically efficient than the masking approach. This research addresses a prototypical problem that asks how improvements could be made in AI by introducing Bayesian statistical thinking. We compare the two approaches to missingness using simulated incomplete data from a variety of spatial models. The utility of the methodology is shown on Arctic sea-ice data, analyzed using a novel hidden Potts model with an intractable likelihood. 
  \\ 

\end{abstract}
\end{singlespace}

\begin{bibunit}[apalike] 

\section{Introduction}\label{sec:introduction}

Artificial intelligence (AI) and deep learning have spurred significant advancements in recent years, revolutionizing fields ranging from image recognition to natural language processing, and enabling transformative technologies like ChatGPT \citep{openai_2025_chatgpt}. These breakthroughs are profoundly influencing how we work, communicate, and create, leaving an indelible mark on society. In parallel, the last decade has seen growing interest in the adoption of neural networks for simulation-based inference, which is often used in statistical or physical models for which the likelihood function is unavailable or computationally intractable \citep{Diggle_1984_implicit_generative_models, Tavare_1997_ABC, Cranmer_2020_simulation-based_inference}. Neural networks are being used to approximate the likelihood function \citep[e.g.,][]{Papamakarios_2019}, the likelihood-to-evidence ratio \citep[e.g.,][]{Hermans_2020, Thomas_2022_ratio_estimation, Walchessen_2023_neural_likelihood_surfaces}, and the posterior distribution \citep[e.g.,][]{Greenberg_2019, Goncalves_2020, Radev_2022_BayesFlow}; see \citet{Zammit_2024_ARSIA} for a recent review. 
  In this work, we consider neural Bayes estimators \citep[NBEs;][]{Sainsbury-Dale_2022_neural_Bayes_estimators}, which are neural networks that map data to a point summary of the posterior distribution. These estimators are amortized, in the sense that, after an initial set-up cost, inference from observed data can be made in a fraction of the time required by conventional approaches.  
  NBEs have been used to make fast inference with models for population genetics \citep{Flagel_2018}, financial options \citep{Hernandez_2017, Horvath_2021, Coloma_Kleiber_2025}, cognitive processes \citep{Pan_2025}, point processes \citep{Lambe_2026}, spatial processes \citep{Gerber_Nychka_2021_NN_param_estimation, 
  Banesh_2021_neural_estimator_GP, 
  Lenzi_2021_NN_param_estimation, Richards_2023_censoring, Sainsbury-Dale_2022_neural_Bayes_estimators, Sainsbury-Dale_2023_GNNs, Tsyrulnikov_2024, 
  villazon_2025, Wang_2024_GSUN_NBEs}, and spatio-temporal processes \citep{delloro_gaetan_2025_spatiotemporal_NBEs}. 

NBEs represent a promising approach to inference, yet significant challenges remain that hinder their widespread adoption. One of the key challenges is their application to the often-encountered ``incomplete data'' setting, where the structure of the available data renders the use of conventional neural networks problematic. Consider, for instance, remote-sensing data collected over a regular grid. If all pixels are observed (i.e., no data are missing), one can readily construct an NBE of a geophysical parameter using a standard convolutional neural network (CNN). However, this parsimonious, efficient architecture cannot be used directly if data are missing (e.g., due to cloud cover). Similarly, in medical and health applications, incomplete data often arise in electronic health records, where patients may have missing clinical or demographic information, or in clinical trials, where participants may drop out or miss scheduled evaluations. In these scenarios, the application of NBEs requires novel methods to handle incomplete data in a statistically principled manner. 

A recently proposed method for handling missing data involves ``masking'' \citep{Wang_2022_neural_missing_data, gloeckler2024allinone}. In this approach, the neural network takes as input the data completed with missing elements replaced by zero (or some other fixed constant) and a binary vector encoding the missingness pattern. Although computationally efficient, the masking approach has some drawbacks. For instance, incorporating the missingness pattern as an additional input results in a more challenging learning task. Further, the approach necessitates a stochastic model for the missingness mechanism and, as we will show, misspecification of this model can lead to biased and suboptimal inference.

Statistical efficiency is as important as computational efficiency, and in this research we consider both in a prototype problem that indicates where improvements could be made in AI by introducing statistical inferential tools. Specifically, we propose to facilitate the use of NBEs in the presence of missing data by leveraging methods developed for incomplete-data problems, where the data are challenging to analyze directly but their analysis becomes tractable through appropriate data augmentation. Several algorithms have been developed to address such problems, among which the most prominent are the expectation-maximization (EM) algorithm for maximum-likelihood or maximum-a-posteriori (MAP) estimation \citep{Dempster_1977_EM_algorithm} and the data-augmentation algorithm for posterior sampling \citep{Tanner_Wong_1987}.  These approaches exploit likelihood functions or posterior distributions, respectively, that are intractable for the observed incomplete data but tractable under data augmentation. There is a strong parallel between the classical incomplete-data problem and that faced in neural inference: there are situations where neural-network architectures are complex or unavailable for observed incomplete data, but they are simpler, more parsimonious, and easier to train, under data augmentation. Note that in machine learning, ``data augmentation'' typically refers to methods that artificially increase the amount of data used to train neural networks; instead, we use it exclusively here to refer to the process of augmenting data with latent random variables. 

Building on this connection between classical and neural inference, we introduce an implementation of NBEs for incomplete data using data augmentation. Specifically, we develop a type of Monte Carlo EM (MCEM) algorithm \citep{Wei_Tanner_1990_Monte_Carlo_EM} where, after conditional simulation of the missing data, the remaining components of the usual E- and M-steps are obtained almost instantaneously with an NBE trained to approximate the MAP estimator. MAP estimators are prominent in areas as diverse as image analysis \citep{Hardie_2004} and plant breeding \citep{Montesinos-Lopez_2020}. 
 Our EM approach to neural Bayes estimation is likelihood-free, in the sense that it does not require evaluation of the likelihood function, and it does not require numerical optimization at each iteration, making it much faster than the classical EM algorithm. Critically, the NBE is trained on, and applied to, complete data only, which we show alleviates the drawbacks of the masking approach discussed above. Although our EM approach applies to a narrower class of models and is computationally slower than the masking approach, it provides stronger statistical guarantees; this trade-off clearly emerges in our simulation studies. Both methods have been incorporated into the user-friendly open-source software package \ifbool{blind}{\texttt{<redacted for anonymity>}}{\pkg{NeuralEstimators} \citep{NeuralEstimators}}, which is available in \proglang{Julia} and \proglang{R}.  

The remainder of this article is organized as follows. In Section~\ref{sec:methodology}, we review NBEs, discuss and give new insights into the masking approach of \cite{Wang_2022_neural_missing_data}, and introduce our EM approach. In Section~\ref{sec:simulationstudies}, we conduct simulation studies to investigate the strengths and weaknesses of these two approaches for dealing with missing data. In Section~\ref{sec:application}, we apply our methodology to an analysis of Arctic sea-ice data. Finally, in Section~\ref{sec:conclusion} we give conclusions and outline avenues for future research. Supplementary material is also available that contains additional theoretical details, simulations, and figures. Code that reproduces all results in the manuscript is available from \ifbool{blind}{\texttt{<redacted for anonymity>}}{\url{https://github.com/msainsburydale/NeuralIncompleteData}}. 

\section{Methodology}\label{sec:methodology}

In Section~\ref{sec:neuralBayesEstimators}, we review NBEs. (For a more comprehensive introduction, see \citeauthor{Sainsbury-Dale_2022_neural_Bayes_estimators}, \citeyear{Sainsbury-Dale_2022_neural_Bayes_estimators}, and \citeauthor{Zammit_2024_ARSIA}, \citeyear{Zammit_2024_ARSIA}.) In Section~\ref{sec:neuralmasking}, we present the masking approach for missing data and, in Section~\ref{sec:EMalgorithm}, we describe our novel approach based on MCEM. 

\subsection{Neural Bayes estimators}\label{sec:neuralBayesEstimators}

The goal of parametric point estimation is to estimate a $d$-dimensional parameter $\vec{\theta} \in \Theta$ from data $\vec{Z} \in \mathcal{Z}$ using an estimator, $\hat{\vec{\theta}} : \mathcal{Z}\to\Theta$. For ease of exposition, we let $\Theta \subseteq \mathbb{R}^d$ and $\mathcal{Z} \subseteq \mathbb{R}^n$, although the approaches we describe generalize to other spaces. A ubiquitous decision-theoretic approach to the construction of estimators is based on average-risk optimality (e.g., \citeauthor{Lehmann_Casella_1998_Point_Estimation}, \citeyear{Lehmann_Casella_1998_Point_Estimation}, Ch.~4; \citeauthor{Robert_2007_The_Bayesian_Choice}, \citeyear{Robert_2007_The_Bayesian_Choice}, Ch.~4). Consider a 
 loss function $L: \Theta \times \Theta \to [0, \infty)$ and, for ease of exposition, assume that the prior measure for $\vec{\theta}$ admits a density $\priordensity(\cdot)$ with respect to Lebesgue measure on $\mathbb{R}^d$. Then the Bayes risk of the estimator $\hat{\vec{\theta}}(\cdot)$ is  
\begin{equation}\label{eqn:BayesRisk}
 \BayesRisk(\hat{\vec{\theta}}(\cdot)) 
 \equiv \int_\Theta \int_{\mathcal{Z}}  L(\vec{\theta}, \hat{\vec{\theta}}(\vec{z}))p_{\vec{Z} \mid \vec{\theta}}(\vec{z} \mid \vec{\theta}) \priordensity(\vec{\theta}) \d \vec{z} \d \vec{\theta},
 \end{equation}  
 where here and throughout, for generic random quantities $A$ and $B$, we use $p_{A\mid B}(\cdot \mid \cdot)$ to denote the conditional probability density or mass function of $A$ given $B$, and $p_{A}(\cdot)$ for the corresponding marginal. When clear from context, we omit subscripts for brevity. A minimizer of \eqref{eqn:BayesRisk} is said to be a Bayes estimator with respect to $L(\cdot,\cdot)$ and $\priordensity(\cdot)$. 

Bayes estimators are functionals of the posterior distribution (e.g., the posterior mean under quadratic loss) and are often unavailable in closed form. However, since estimators are mappings from the sample space $\mathcal{Z}$ to the parameter space $\Theta$, Bayes estimators could, in principle, be approximated well by a sufficiently flexible function. Recently, motivated by universal-function-approximation theorems \citep[e.g.,][]{Hornik_1989_FNN_universal_approximation_theorem, Zhou_2020_universal_approximation_CNNs} and the speed at which they can be evaluated, neural networks have been used to approximate Bayes estimators \citep[see, e.g.,][Sec.~3.1]{Zammit_2024_ARSIA}. Let $\vec{f} : \mathcal{Z}\to\Theta$ denote a neural network parameterized by $\vec{\gamma}$, that is, 
\begin{equation}\label{eqn:NN}
\vec{f}(\vec{Z}; \vec{\gamma}) 
= (\vec{f}_{J} \circ \vec{f}_{J - 1} \circ \dots \circ \vec{f}_1)(\vec{Z}; \vec{\gamma}), \quad \vec{Z} \in \mathcal{Z}, 
\end{equation} 
 where $\vec{f}_j(\cdot \,; \vec{\gamma}_j)$, $j = 1, \dots, J$, are nonlinear functions parameterized by $\vec{\gamma}_j$, and `$\circ$' denotes function composition. Throughout, we use $\vec{\gamma} = (\vec{\gamma}_1^\tp, \dots, \vec{\gamma}_{J}^\tp)^\tp \in \Gamma$ to denote generic neural-network parameters, although the number of parameters varies depending on context. Then, a Bayes estimator may be approximated by substituting 
\begin{equation}\label{eqn:optimization_task}
\vec{\gamma}^*
\equiv 
\underset{\vec{\gamma} \in \Gamma}{\mathrm{arg\,min}} \;
\frac{1}{K} \sum_{k=1}^K L\{\vec{\theta}^{(k)}, \vec{f}(\vec{Z}^{(k)}; \vec{\gamma})\}
\end{equation} 
into~\eqref{eqn:NN} where,
independently for each $k$, $\vec{\theta}^{(k)} \sim \priordensity(\vec{\theta})$  and $\vec{Z}^{(k)} \sim p(\vec{Z} \mid  \vec{\theta}^{(k)})$. 
The process of performing the optimization task \eqref{eqn:optimization_task} on neural-network parameters $\vec{\gamma}$ given in \eqref{eqn:NN} is referred to as ``training the network'', and this can be done efficiently using back-propagation and stochastic gradient descent \citep{Goodfellow_2016_Deep_learning}. 

The trained neural network $\vec{f}(\cdot; \vec{\gamma}^*)$ minimizes the approximate Bayes risk, and therefore it is called an NBE. Once trained, an NBE can be applied repeatedly to data sets that are realizations from the statistical model used for training, at a fraction of the computational cost of conventional inferential methods. It is therefore ideal to use an NBE in settings where inference needs to be made repeatedly; in this case, the initial training cost is said to be \textit{amortized} over time. 

Although \eqref{eqn:optimization_task} is written for a fixed training set of size $K$, in practice one often uses ``on-the-fly'' simulation, continually generating fresh parameter--data pairs $\vec{\theta}^{(k)}$ and $\vec{Z}^{(k)}$ during training. In that regime, the number $K$ of such pairs need not be specified at all. However, with a fixed training set, $K$ becomes an important hyperparameter that typically needs to be increased with the dimension $d$ of $\vec{\theta}$ and the tail-heaviness of the statistical model \citep{Rodder_2025}. By contrast, when using parsimonious architectures that exploit the structure of the data (see below), the dimension $n$ of $\vec{Z}$ does not usually drive the choice of $K$, since parameter sharing and invariances in neural networks tend to reduce the effective dimensionality of the learning problem. 

When constructing an NBE, a central consideration in designing the neural-network architecture (i.e., the functional form of \eqref{eqn:NN}) is the underlying structure of the data. For example, if the data are gridded, a CNN is typically used, while for unstructured data, one typically adopts a classical multilayer perceptron (MLP). When the data are exchangeable, frameworks such as DeepSets \citep{Zaheer_2017_Deep_Sets} adapt these standard architectures in a manner that parsimoniously exploits exchangeability. However, beyond these general guidelines, the specific design of the network (e.g., depth, width, activation functions) and the choice of training-related hyperparameters (e.g., learning rate) typically require experimentation, tuning and, in some cases, automated search methods \citep[see, e.g.,][]{Elsken_2019}. While \citet{Rodder_2025} take important first steps in developing theory for the case of a single-layer perceptron applied to replicated data, extending such results to the richer classes of architectures commonly used in practice remains a challenge. Fortunately, once a suitable architecture is found, it often generalizes well across models, as demonstrated by our simulation experiments in Section~\ref{sec:simulationstudies}.

As discussed in Section~\ref{sec:introduction}, the standard architectures outlined above do not naturally cater for missing data, and this has limited the applicability of neural Bayes estimation. Next, we discuss two approaches that address this challenge.  
 
\subsection{The masking approach for missing data}\label{sec:neuralmasking}

Data are often incomplete, and hence inference on $\vec{\theta}$ is made using only a subset of $\vec{Z}$. For a given $\vec{Z} \equiv (Z_1, \dots, Z_n)^\tp$, we denote the subvectors of observed and missing elements as $\vec{Z}_1$ and $\vec{Z}_2$, respectively. We use $\mathcal{I}_1 \equiv \{i: Z_i \text{ is observed}\}$ to denote the ordered set of indices corresponding to the observed component, so that $\vec{Z}_1 \equiv (Z_i : i \in \mathcal{I}_1)^\tp$. 

The masking approach we present in this section closely follows that of \cite{Wang_2022_neural_missing_data}, who applied it in the context of approximate posterior inference. Their approach consists of first constructing a masked version of $\vec{Z}$, denoted by $\vec{U} \in 
 \mathcal{U} \subseteq \mathbb{R}^n$, 
 with components
\begin{align}
\begin{split}
  \vec{U}_1 &\equiv (U_i : i \in \mathcal{I}_1)^\tp  = \vec{Z}_1,  \\ 
  \vec{U}_2 &\equiv (U_i : i \in \mathcal{I}_2)^\tp  = c\vec{1},  
\end{split}\label{eqn:paddeddata}
\end{align}
where 
$\mathcal{I}_2 \equiv \{1, \dots, n\} \setminus \mathcal{I}_1$, 
$c \in \mathbb{R}$ is fixed (we set $c=0$ throughout), and $\vec{1}$ denotes a vector of 1s of appropriate dimension. Now, define a vector of indicator variables, $\vec{W} \in \mathcal{W} = \{0, 1\}^n$, as follows:
\begin{align}
\begin{split}
\vec{W} \equiv (\mathbb{I}(i\in \mathcal{I}_1) : i = 1, \dots, n)^\tp,
\end{split}\label{eqn:binarymask}
\end{align}
where $\mathbb{I}(\cdot)$ denotes the indicator function. 
While $\vec{Z}_1$ and $\mathcal{I}_1$ have a dimension that might vary across different observed data sets, the quantities $\vec{U}$ and $\vec{W}$ are each of fixed dimension $n$, which enables the use of parsimonious, efficient neural-network architectures. 

An NBE based on the masking approach is constructed by first defining 
\begin{equation}\label{eqn:NNmasking}
\vec{g}(\vec{U}, \vec{W}; \vec{\gamma}), 
\text{ for }\vec{U} \in \mathcal{U}, \; \vec{W} \in \mathcal{W}, 
\end{equation} 
where $\vec{g}(\cdot, \cdot \,; \vec{\gamma})$ is a neural network parameterized by $\vec{\gamma} \in \Gamma$. Then, substitute 
 \begin{equation}\label{eqn:optimization_task_encoding}
\vec{\gamma}^*
\equiv 
\underset{\vec{\gamma} \in \Gamma}{\mathrm{arg\,min}} \;
\sum_{k=1}^K L\big\{\vec{\theta}^{(k)}, \vec{g}(\vec{U}^{(k)}, \vec{W}^{(k)}; \vec{\gamma})\big\}
\end{equation}  
into \eqref{eqn:NNmasking}, where $\vec{\theta}^{(k)} \sim \priordensity(\vec{\theta})$ and, independently for each $k$, $\vec{U}^{(k)}$ and $\vec{W}^{(k)}$ are constructed by first sampling indices 
 $\mathcal{I}_1^{(k)}$ from a model $p(\mathcal{I}_1 \mid \vec{\theta}^{(k)})$ for the missingness mechanism, simulating complete data $\vec{Z}^{(k)} \sim p(\vec{Z} \mid \vec{\theta}^{(k)})$, subsetting $\vec{Z}_1^{(k)} \equiv (Z_i^{(k)} : i \in \mathcal{I}_1^{(k)})^\tp$,
    and substituting these quantities into~\eqref{eqn:paddeddata} and \eqref{eqn:binarymask}. 
  We note that for some models it may be possible to directly simulate incomplete data.
 Further, as discussed in Section~\ref{sec:neuralBayesEstimators}, rather than fixing $K$, one may simulate parameter--data pairs and missingness patterns on-the-fly. Once trained, the neural network can be used repeatedly to estimate parameters from new incomplete data sets. Algorithm~\ref{alg:one-hot} summarizes the approach, and Figure~\ref{fig:masking} illustrates its Estimation stage. 

\begin{algorithm}[!t]
\caption{The masking approach to neural Bayes estimation with missing data.}\label{alg:one-hot}
\begin{algorithmic}[1]
\smallskip
\Statex \hspace{-2em} \textbf{Training stage} (slow, to be done only once offline)
\Require{Prior $\priordensity(\vec{\theta})$, number of training samples $K$, probability models $p(\mathcal{I}_1 \mid \vec{\theta})$ and $p(\vec{Z} \mid \vec{\theta})$,  $c \in \mathbb{R}$ for use in~\eqref{eqn:paddeddata}, loss function $L: \Theta \times \Theta \to [0, \infty)$, neural network $\vec{g}: \mathcal{U} \times \mathcal{W} \to \Theta$ parameterized by $\vec{\gamma}$.}
\smallskip
    \For{$k = 1, \dots, K$}
    \State  Sample parameters $\vec{\theta}^{(k)} \sim \priordensity(\vec{\theta})$.
    \State  Simulate data $\vec{Z}^{(k)} \sim p(\vec{Z} \mid \vec{\theta}^{(k)})$. 
    \State  Sample indices $\mathcal{I}_1^{(k)} \sim p(\mathcal{I}_1 \mid \vec{\theta}^{(k)})$. 
 	  \State Subset $\vec{Z}_1^{(k)} \equiv (Z_i^{(k)} : i \in \mathcal{I}_1^{(k)})^\tp$. 
     \State Compute $\vec{U}^{(k)}$ using \eqref{eqn:paddeddata}. 
     \State Compute $\vec{W}^{(k)}$ using \eqref{eqn:binarymask}. 
    \EndFor
     \State Solve $\vec{\gamma}^* \equiv \mathrm{argmin}_\vec{\gamma} \;
\sum_{k=1}^K L\big\{\vec{\theta}^{(k)}, \vec{g}(\vec{U}^{(k)}, \vec{W}^{(k)}; \vec{\gamma})\big\}$ to obtain the masking NBE, $\vec{g}(\cdot, \cdot; \vec{\gamma}^*)$.
\bigskip
    \Statex  \hspace{-2em} \textbf{Estimation stage} (fast, repeatable for arbitrarily many observed data sets) 
    \setcounter{ALG@line}{0}
\Require{Observed data $\vec{Z}_1$ and $\mathcal{I}_1$, $c \in \mathbb{R}$ for use in \eqref{eqn:paddeddata}.} 
\State Compute $\vec{U}$ using \eqref{eqn:paddeddata}.
\State Compute $\vec{W}$ using \eqref{eqn:binarymask}.
\State Return $\vec{g}(\vec{U}, \vec{W}; \vec{\gamma}^*)$. 
\end{algorithmic}
\end{algorithm}

\begin{figure}[t!]
	\centering
	\includegraphics[width = 0.6\linewidth]{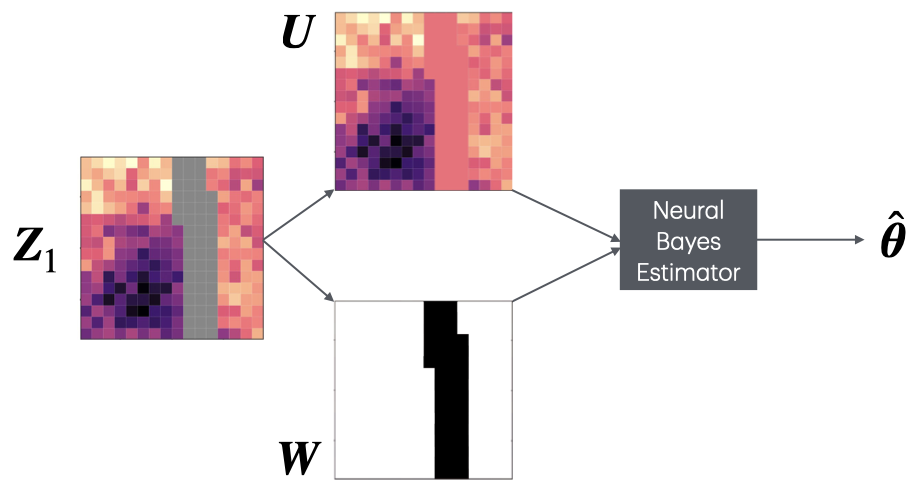}  
	\caption{The Estimation stage of Algorithm~\ref{alg:one-hot}. Observed data $\vec{Z}_1$, and the associated indices $\mathcal{I}_1$ (here, implicit) that identify which elements of $\vec{Z}$ are observed, are used to construct $\vec{U}$, a masked version of the complete data $\vec{Z}$ with missing entries replaced by a constant $c \in \mathbb{R}$, and $\vec{W}$, a vector of indicator variables that encode the missingness pattern. The encoded data $\vec{U}$ and $\vec{W}$ are then input to an NBE to obtain point estimates $\hat{\vec{\theta}}$ of a model parameter $\vec{\theta}$.}\label{fig:masking} 
\end{figure}

Algorithm~\ref{alg:one-hot} has a number of strengths. It only requires simulation from $p(\mathcal{I}_1 \mid \vec{\theta})$ and $p(\vec{Z} \mid \vec{\theta})$, and no information on $\vec{\theta}$ is lost by using $\vec{U}$ and $\vec{W}$ instead of $\vec{Z}_1$ and $\mathcal{I}_1$, since there is a one-to-one mapping between these quantities (see Theorem~\reff{thm:masking_likelihood} in Section~\reffsupp{app:likelihoodequivalence}). However, the approach has two drawbacks. First, the neural network must learn a mapping from $\mathcal{U} \times \mathcal{W}$ to $\Theta$. This learning task can be more challenging than learning a mapping from $\mathcal{Z}$ to $\Theta$, since an element in $\mathcal{U} \times \mathcal{W}$ is semi-discrete and of dimension $2n$. 
 Second, the quality of the fit is subject to a judicious choice of the missingness model. To see this, rewrite the Bayes risk~\eqref{eqn:BayesRisk} in terms of an estimator $\hat{\vec{\theta}}(\cdot, \cdot)$ treating $\mbf{U}$ and $\mbf{W}$ as the data:
\begin{equation}\label{eqn:risk_encoding}
 \BayesRisk(\hat{\vec{\theta}}(\cdot, \cdot)) \equiv \sum_{\vec{w} \in \mathcal{W}}\int_\Theta  \int_{\mathcal{U}} L(\vec{\theta}, \hat{\vec{\theta}}(\vec{u}, \vec{w}))p_{\vec{U}, \vec{W} \mid \vec{\theta}}(\vec{u}, \vec{w} \mid \vec{\theta}) \priordensity(\vec{\theta}) \d \vec{u}\d \vec{\theta}.
 \end{equation} 
It is a straightforward corollary of Theorem~\reff{thm:masking_likelihood} in Section~\reffsupp{app:likelihoodequivalence} and Bayesian sufficiency \citep[e.g., ][Ch.~2]{Cox_1974} that a Bayes estimator that minimizes \eqref{eqn:risk_encoding} also minimizes the Bayes risk defined in terms of $\vec{Z}_1$ and $\mathcal{I}_1$. Further, under the assumption that $\mathcal{I}_1$ is independent of $\vec{\theta}$, which is often reasonable in practice, a Bayes estimator that minimizes \eqref{eqn:risk_encoding} is invariant to distributions on $\mathcal{I}_1$ that are positive on the power set of $\{1, \dots, n\}$ (see Theorem~\reff{thm:invariance} in Section~\reffsupp{app:likelihoodequivalence}). However, in practice, the empirical analogue of~\eqref{eqn:risk_encoding}, given by the sum on the right-hand-side of \eqref{eqn:optimization_task_encoding}, is subject to Monte Carlo error that depends on the chosen distribution for $\mathcal{I}_1$. Therefore, the choice of distribution has practical implications on the approximation of the Bayes estimator obtained by substituting \eqref{eqn:optimization_task_encoding} into \eqref{eqn:NNmasking}, particularly in high-dimensional settings where the total number of possible missingness patterns ($2^n$) is large and the missingness mechanism is difficult to specify. In Section~\ref{sec:simulationstudies}, we show through simulation that selecting a distribution for $\mathcal{I}_1$ that assigns low probability to the observed missingness pattern can lead to statistically inefficient and biased estimators.   
  
To address these limitations, in this paper we propose an alternative statistical approach to neural Bayes estimation with incomplete data. Our approach does not require the missingness pattern to be an input to the neural network, or the specification of a missingness mechanism. We do this by embedding neural networks in a classical data-augmentation approach to solving incomplete-data problems, namely an MCEM algorithm. 

\subsection{The EM approach for missing data}\label{sec:EMalgorithm} 

In Section~\ref{sec:EMalgorithm:background}, we provide an overview of the classical EM algorithm and its Monte Carlo version. In Section~\ref{sec:EMalgorithm:general}, we outline the general structure of our proposed EM approach to neural Bayes estimation with incomplete data. In Section~\ref{sec:neuralMAPestimation}, we detail the construction of a neural approximation to the MAP estimator, which is needed for our EM approach. 

\subsubsection{The classical EM algorithm and its Monte Carlo version}\label{sec:EMalgorithm:background}

Recall that we use $\vec{Z}_1$ and $\vec{Z}_2$ to denote the subvectors of $\vec{Z}$ that are treated as observed and missing, respectively, and that we use \mbox{$\vec{Z}$} to denote the complete data. We also have available an ordered set of indices, $\mathcal{I}_1$, associated with $\vec{Z}_1$. However, since one does not need to construct a mask from these indices in our approach described below, we omit the explicit notation of these indices in this subsection. 

In the classical statistics literature, many algorithms have been developed based on the ``data augmentation principle'', which is applied when inference based on $\vec{Z}$ is easier than inference based only on $\vec{Z}_1$ 
\citep[see, e.g.,][]{Tanner_Wong_1987, 
vanDyk_2001}. 
 A popular approach to point estimation that follows from this principle is the EM algorithm \citep{Dempster_1977_EM_algorithm, Wu_1983_EM_algorithm, McLachlan_2008_EM_algorithm}, an iterative algorithm for estimating $\vec{\theta}$ with the $l$th iteration given by 
 \begin{equation}\label{eqn:EM}
\EMestimate^{(l)} = M(\EMestimate^{(l-1)}) \equiv 
\argmax_{\vec{\theta} \in \Theta} \, \Big\{ \log\priordensity(\vec{\theta}) + \E_{\vec{Z}_2 \mid \vec{Z}_1, \EMestimate^{(l-1)}}\ell(\vec{\theta}; \vec{Z}_1, \vec{Z}_2)\Big\};\quad  l = 1, 2, \dots, 
 \end{equation}
 where $\priordensity(\vec{\theta})$ denotes a prior density (sometimes referred to as a penalty function) and $\ell(\vec{\theta}; \vec{Z}_1, \vec{Z}_2) \equiv  \log p(\vec{Z} \mid \vec{\theta})$ denotes the complete-data log-likelihood.
 The EM algorithm satisfies the ascent property, which means that the incomplete-data posterior density $p(\vec{\theta} \mid \vec{Z}_1)$ is non-decreasing at each iteration. Thus, under mild conditions \citep[][Ch.~3]{Boyles_1983, Wu_1983_EM_algorithm, McLachlan_2008_EM_algorithm}, it yields a local maximizer of 
 $p(\vec{\theta} \mid \vec{Z}_1)$.
  For the special case of $\priordensity(\vec{\theta}) \propto 1$, 
  the EM algorithm increases the incomplete-data likelihood $p(\vec{Z}_1 \mid \vec{\theta})$ at each iteration and, therefore, yields a local maximizer of 
  $p(\vec{Z}_1 \mid \vec{\theta})$. 
 
When the conditional expectation in~\eqref{eqn:EM} is intractable but conditional simulation is feasible, one often adopts a Monte Carlo version of the EM algorithm, MCEM \citep{Wei_Tanner_1990_Monte_Carlo_EM, Ruth_2024}, 
which has as the $l$th iteration,  
\begin{equation}\label{eqn:MCEM}
     \EMestimate^{(l)} = M_m(\vec{\theta}^{(l-1)}) \equiv 
     \argmax_{\vec{\theta} \in \Theta} \bigg\{\log \priordensity(\vec{\theta}) + \frac{1}{m}\sum_{j = 1}^m \ell(\vec{\theta};  \vec{Z}_1,  \vec{Z}_2^{(l,j)}) \bigg\},
\end{equation}
where $\{\vec{Z}_2^{(l,j)}\}_{j=1}^m$ are simulated from the conditional distribution of $\vec{Z}_2 \mid \vec{Z}_1, \EMestimate^{(l-1)}$.
The frequentist MCEM algorithm is recovered as a special case of~\eqref{eqn:MCEM} by taking $\priordensity(\vec{\theta}) \propto 1$.

A practical challenge in implementing the MCEM algorithm is choosing the Monte Carlo sample size $m$. One way to tackle this challenge is through averaging \citep[][Sec.~11.1.2.2]{Fort_Moulines_2003, Cappe_2005}, where convergence assessment and final estimates are based on the averaged subsequence, 
\begin{equation}\label{eqn:MCEMaverage}
\bar{\vec{\theta}}^{(l)}_b \equiv \frac{1}{l - b}\sum_{t = b + 1}^{l} \EMestimate^{(t)},
\end{equation}
where $b < l$ denotes a burn-in chosen such that $\{\EMestimate^{(t)} : t = b + 1, \dots, l\}$ is approximately stationary. This averaging greatly reduces sensitivity to the choice of $m$, as Monte Carlo variability is effectively ``averaged out'' over multiple iterations. It is also formally motivated by the result of \citet{Chan_1995}, that for large $m$, an MCEM sequence starting in a small neighborhood of a local maximizer of $p(\vec{\theta} \mid \vec{Z}_1)$ can be approximated by a stationary first-order autoregressive (AR(1)) process centered at that maximizer. When employing averaging, the main consideration in choosing $m$ is ensuring that the iterates become stable, in the sense that they eventually remain centered around a single point.
 Since the Monte Carlo variability is driven by the conditional variance $\V_{\vec{Z}_2 \mid \vec{Z}_1, \vec{\theta}} \{\ell(\vec{\theta}; \vec{Z}_1, \vec{Z}_2)\}$, which typically increases with both the tail-heaviness of the model and the proportion of missing data, $m$ should be increased with both of these quantities whenever stability is a concern.

A drawback of the (MC)EM algorithm is that it can be slow since it is ``doubly iterative'' in the typical case where each maximization step requires numerical optimization. Although a substantial amount of work has been devoted to trying to speed up the algorithm, primarily by accelerating its rate of convergence \citep[e.g.,][]{Louis_1982, Meng_1993, Liu_1994, Jamshidian_1997, Liu_1998_PX-EM, Neal_1998, Varadhan_2008, Lewandowski_2010_PX-EM_review}, computational speed remains a practical limitation in many settings. Further, although the MCEM algorithm bypasses the conditional expectation in~\eqref{eqn:EM}, it still requires evaluation of the complete-data log-likelihood function, which is not always possible. Next, we describe how the MCEM algorithm can incorporate NBEs with incomplete data in a manner that is not subject to these limitations. 

\subsubsection{NBEs and approximate MCEM algorithms}\label{sec:EMalgorithm:general}
 
Our EM approach to using NBEs with incomplete data is predicated on the fact that~\eqref{eqn:MCEM} is equivalent to,
\begin{equation}\label{eqn:EM:MAP}
  	\EMestimate^{(l)} = M_m(\EMestimate^{(l-1)}) = 
    \argmax_{\vec{\theta} \in \Theta} \bigg\{\log\priordensitypower(\vec{\theta}) + \sum_{j = 1}^m \ell(\vec{\theta};  \vec{Z}_1,  \vec{Z}_2^{(l,j)}) \bigg\},
\end{equation}
where the probability density $\priordensitypower(\vec{\theta}) \propto \{\priordensity(\vec{\theta})\}^m$. The reformulation of~\eqref{eqn:MCEM} as~\eqref{eqn:EM:MAP} shows that the conventional MCEM update $M_m(\vec{\theta})$ is a MAP estimate under a modified prior distribution, which can be approximated by an NBE (denoted by $\vec{h}(\cdot)$ in \eqref{eqn:DeepSets} and fitted according to \eqref{eqn:optimization_task:adjustedprior_bayes} in Section~\ref{sec:neuralMAPestimation}). This leads to our proposed EM approach (Algorithm~\ref{alg:neuralEM}), a fast version of the MCEM algorithm that does not require the evaluation of any likelihood functions. Figure~\ref{fig:neuralEM} illustrates the Estimation stage of the algorithm. After conditional simulation, the update~\eqref{eqn:EM:MAP} is obtained in a fraction of a second. Importantly, since the incomplete data are completed by conditional simulation, the NBE is applied to complete-data vectors only. Therefore, our proposed algorithm does not require a model for the missingness mechanism.

\begin{algorithm}[!t]
  \caption{The EM approach to neural Bayes estimation with missing data.}\label{alg:neuralEM}
  \begin{algorithmic}[1]
\smallskip
    \Statex \hspace{-2em} \textbf{Training stage} (slow, to be done only once offline)
\Require{Prior $\priordensity(\vec{\theta})$, number of training samples $K$, number of Monte Carlo samples $m$ used in the Estimation stage, probability model $p(\vec{Z} \mid \vec{\theta})$, neural network $\vec{h}: \mathcal{Z}^m \to \Theta$ parameterized by $\vec{\gamma}$.} 
	\smallskip
	 \For{$k = 1, \dots, K$}
    \State Sample parameters $\vec{\theta}^{(k)} \sim \priordensity(\vec{\theta})$. 
    \State Simulate data $\vec{Z}^{(k,j)}  \sim p(\vec{Z} \mid \vec{\theta}^{(k)})$ for $j = 1, \dots, m$. 
    \EndFor
    \State 
    Solve $\vec{\gamma}^* \equiv \mathrm{argmin}_\vec{\gamma} \; \sum_{k=1}^K \{\priordensity(\vec{\theta}^{(k)})\}^{m-1}L\big\{\vec{\theta}^{(k)}, \vec{h}(\{\vec{Z}^{(k,j)}\}_{j=1}^m; \vec{\gamma})\big\}$ (see~\eqref{eqn:optimization_task:adjustedprior_bayes}) with $L(\cdot, \cdot)$ a continuous and almost-everywhere differentiable approximation of the 0--1 loss function (e.g., \eqref{eqn:surrogateloss}), to obtain the NBE, $\vec{h}(\cdot \,; \vec{\gamma}^*)$, that approximates the MAP estimator. 
\medskip
    \Statex \hspace{-2em} \textbf{Estimation stage} (fast, repeatable for arbitrarily many observed data sets) 
    \setcounter{ALG@line}{0}
\Require{Observed data $\vec{Z}_1$, number of Monte Carlo samples $m$, initial estimates $\EMestimate^{(0)}$, convergence criterion, burn-in $b$, maximum number of iterations, 
algorithm to simulate from the conditional distribution $p(\vec{Z}_2 \mid \vec{Z}_1, \vec{\theta})$.} 
\smallskip
    \State  Set $l = 0$.
    \Repeat
    \State Set $l = l + 1$.
    \State  \multiline{Simulate missing data $\vec{Z}_2^{(l,j)} \sim p(\vec{Z}_2 \mid \vec{Z}_1, \EMestimate^{(l-1)})$, $j = 1, \dots, m$, 
    resulting in $m$ conditionally independent replicates of the completed data, $\{\vec{Z}^{(l,j)}\}_{j=1}^m$.}
    \State Update $\EMestimate^{(l)} = \vec{h}(\{\vec{Z}^{(l,j)}\}_{j=1}^m; \vec{\gamma}^*)$. 
    \If{$l > b$}
      \State Compute the post--burn-in mean $\bar{\vec{\theta}}^{(l)}_b  = \frac{1}{\,l - b\,} \sum_{t = b+1}^{l} \EMestimate^{(t)}$.
      \State Check convergence of $\bar{\vec{\theta}}^{(l)}_b$ according to the specified criterion.  
    \EndIf
    \Until{converged or maximum number of iterations reached.}
    \State Return $\bar{\vec{\theta}}^{(l)}_b$. 
  \end{algorithmic}
\end{algorithm}

\begin{figure}[t!]
	\centering
	\includegraphics[width = \linewidth]{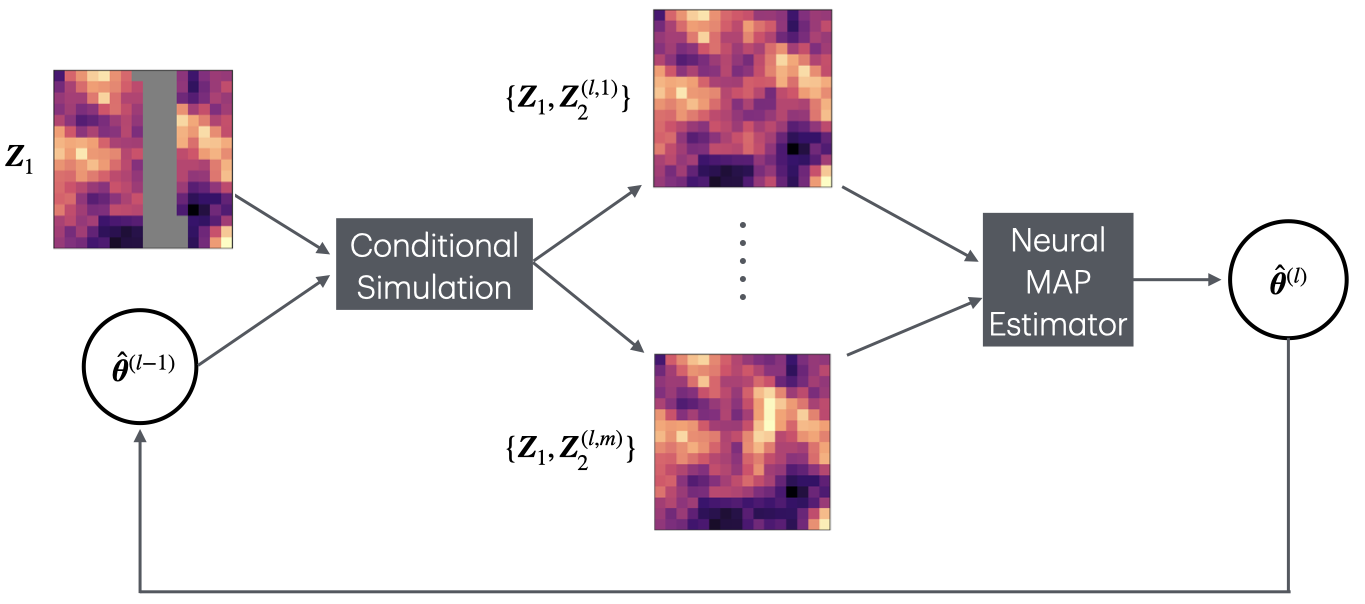}  
	\caption{The Estimation stage of Algorithm~\ref{alg:neuralEM}. Incomplete data $\vec{Z}_1$ with missing entries are completed by conditional simulation using the previous parameter estimate $\EMestimate^{(l-1)}$ of model parameter $\vec{\theta}$. The $m$ conditionally-independent replicates are then input to an NBE trained to approximate the MAP estimator. 
  The parameter estimate $\EMestimate^{(l)}$ is then used for conditional simulation in the next iteration of the algorithm.}\label{fig:neuralEM} 
\end{figure}

Algorithm~\ref{alg:neuralEM} is ideally suited when simulations can be carried out from the distributions of $\vec{Z} \mid \vec{\theta}$ and $\vec{Z}_2 \mid \vec{Z}_1, \vec{\theta}$, and when the incomplete-data and complete-data likelihoods are unavailable or computationally intractable. Statistical models that have these properties include (hidden) Markov random fields (e.g., \citeauthor{Besag_1974}, \citeyear{Besag_1974}; \citeauthor{Rue_Held_2005_GMRF:TaA}, \citeyear{Rue_Held_2005_GMRF:TaA}), such as the (hidden) \citet{Potts_1952} model considered in Section~\ref{sec:Potts}, the autologistic model \citep{Besag_1972}, and other auto-models proposed by \citet{Besag_1974}. These models have intractable likelihoods due to the computational complexity of the required normalizing constants. However, marginal and conditional simulation is feasible through Gibbs sampling. 

 Algorithm~\ref{alg:neuralEM} is equivalent to the conventional MCEM algorithm, up to the error introduced by using an NBE to approximate the conventional update $M_m(\cdot)$. Indeed, it is one member of a broader family of approximate MCEM algorithms that may be expressed in the following general form:
\begin{equation}\label{eqn:approxMCEMupdate}
  \EMestimate^{(l)} = M_m(\EMestimate^{(l-1)}) + \vec{\delta}_{m}(\EMestimate^{(l-1)}), \quad l = 1, 2, \dots,
\end{equation}
where $\vec{\delta}_{m}(\cdot)$ denotes (random) approximation error. In Section~\reffsupp{sec:supp:convergenceofapproximateMCEM}, we derive the asymptotic behavior of the sequence~\eqref{eqn:approxMCEMupdate}. For instance, for large $m$ and with mean-zero approximation error, the sequence~\eqref{eqn:approxMCEMupdate} behaves approximately as an inhomogoneous nonlinear AR(1) process centered on the EM update~\eqref{eqn:EM}. Further, if $\vec{\theta}^*$ is an isolated (local) maximizer of the incomplete data posterior density $p(\vec{\theta} \mid \vec{Z}_1)$, then a sequence~\eqref{eqn:approxMCEMupdate} starting in the vicinity of $\vec{\theta}^*$ behaves approximately as a stationary AR(1) process with mean $\vec{\theta}^*$. These results provide a justification for Algorithm~\ref{alg:neuralEM} and motivate its use with the averaging procedure in \eqref{eqn:MCEMaverage}.
 
 Finally, most of the considerations regarding the choice of Monte Carlo sample size $m$ and the convergence criterion carry over directly from the conventional MCEM algorithm to Algorithm~\ref{alg:neuralEM}. In particular, we employ the averaging prodedure~\eqref{eqn:MCEMaverage} to improve robustness to the choice of $m$; and we adopt the convergence criterion of~\cite{Booth_1999_Monte_Carlo_EM} that terminates the algorithm when the relative change in the averaged estimates falls below a specified tolerance 
 for a number of 
 successive iterations. 
 Details on our specific choices for $m$, burn-in $b$, and convergence threshold are provided in Section~\ref{sec:generalsetting}.

\subsubsection{Neural MAP estimation}\label{sec:neuralMAPestimation}

Algorithm~\ref{alg:neuralEM} hinges on the capacity to approximate the MAP estimator using an NBE. Here, we outline how this can be 
 done, 
and address additional practical considerations. 

\paragraph{Approximating the MAP estimator.} For continuous parameter spaces, the MAP estimator is, under mild conditions \citep{Bassett_2019_MAP_and_Bayes_estimators}, the limit as $\epsilon \to 0$ of the Bayes estimators associated with the 0--1 loss function 
\begin{equation}\label{eqn:0-1loss}
L(\vec{\theta}, \hat{\vec{\theta}}) = \mathbb{I}(\| \vec{\theta} - \hat{\vec{\theta}} \| > \epsilon), 
\end{equation}
where $\|\cdot\|$ denotes any norm in $\mathbb{R}^d$ \citep[][pg.~166]{Robert_2007_The_Bayesian_Choice}, which in our work we choose to be the Euclidean norm. 
  However, \eqref{eqn:0-1loss} is not amenable to gradient-based methods for solving~\eqref{eqn:optimization_task}, and therefore it cannot be used to construct an NBE that approximates the MAP estimator. 
 This challenge may be circumvented by noting that, under suitable regularity conditions, for large $m$ and uniform prior $\pi(\vec{\theta})$, the objective in~\eqref{eqn:EM:MAP} tends to the logarithm of an unnormalized Gaussian density (Bernstein--von Mises theorem; see, e.g., \citeauthor{vanDerVaart_1998}, \citeyear{vanDerVaart_1998}, pg.~140--141). 
 Since the mean of a Gaussian random variable is also its mode, one could therefore choose a quadratic loss function instead of \eqref{eqn:0-1loss} under a uniform prior. For robustness reasons, we prefer to adopt a mathematically convenient surrogate for \eqref{eqn:0-1loss} with derivatives that are continuous almost everywhere; for example, 
    \begin{equation}\label{eqn:surrogateloss}
 L(\vec{\theta}, \hat{\vec{\theta}}; \kappa) 
 \equiv \rm{tanh}(\|\hat{\vec{\theta}} - \vec{\theta}\|/\kappa),
   \quad \kappa > 0,
 \end{equation} 
where $\rm{tanh}(\cdot)$ denotes the hyperbolic tangent function, yields the 0--1 loss function in the limit as $\kappa \to 0$ (see 
Section~\reffsupp{sec:approximations01loss}, 
where we also discuss alternative loss functions). Therefore, instead of~\eqref{eqn:0-1loss} for some $\epsilon$ close to 0, we use~\eqref{eqn:surrogateloss} for some $\kappa$ close to 0. 

\paragraph{Neural-network architecture for handling $m$ conditionally independent replicates.}

In Algorithm~\ref{alg:neuralEM}, the update \eqref{eqn:EM:MAP} involves neural Bayes estimation based on a set of $m$ completed data vectors. 
 A suitable neural-network architecture for these data that ensures scalability with respect to $m$ is DeepSets \citep{Zaheer_2017_Deep_Sets}: 
\begin{equation}\label{eqn:DeepSets} 
\vec{h}(\{\vec{Z}^{(j)}\}_{j=1}^m\,; \vec{\gamma}) = \vec{\phi}\Big\{\frac{1}{m} \sum_{j=1}^m\vec{\psi}(\vec{Z}^{(j)}; \vec{\gamma}_{\psi}); \, \vec{\gamma}_{\phi}\Big\},
\end{equation}
where $\vec{Z}^{(j)} \in \mathcal{Z}$ for $j = 1, \dots, m$, 
$\vec{\psi}(\cdot; \vec{\gamma}_{\psi})$ is a neural network whose architecture depends on the structure of the data, 
 \mbox{$\vec{\phi}(\cdot; \vec{\gamma}_{\phi})$} is an MLP, and the neural-network parameters are $\vec{\gamma} = (\vec{\gamma}_{\psi}^\tp, \vec{\gamma}_{\phi}^\tp)^\tp$. The representation \eqref{eqn:DeepSets} has several motivations. 
First, when the data are exchangeable, the MAP estimator is invariant to permutations of the data, and estimators constructed from~\eqref{eqn:DeepSets} are guaranteed to exhibit this property. Second, under certain conditions, \eqref{eqn:DeepSets} is a universal approximator for continuously differentiable permutation-invariant functions; 
 therefore, any MAP estimator that is a continuously differentiable function of the data can be approximated arbitrarily well by an estimator of the form \eqref{eqn:DeepSets}. Third, \eqref{eqn:DeepSets} may be used with any value of $m$. \citet{Sainsbury-Dale_2022_neural_Bayes_estimators} give further details on the motivation and use of \eqref{eqn:DeepSets} in the general setting of neural Bayes estimation. 

 \paragraph{Accounting for the concentrated probability density $\priordensitypower(\vec{\theta}) \propto \{\priordensity(\vec{\theta})\}^m$.}
The update~\eqref{eqn:EM:MAP} involves $\priordensitypower(\vec{\theta}) \propto \{\priordensity(\vec{\theta})\}^m$, the prior raised to the power of $m$. \Copy{PriorRemainsTheSame}{Note that the prior remains $\priordensity(\vec{\theta})$, since~\eqref{eqn:MCEM} and~\eqref{eqn:EM:MAP} are equivalent.} There are at least two ways in which this distribution can be accounted for. In the first way, one may sample directly from $\priordensitypower(\vec{\theta})$ when constructing the set $\{\vec{\theta}^{(k)}\}_{k=1}^K$, and then train the network as usual. In the second way, one may sample from $\priordensity(\vec{\theta})$ and adjust the objective function through an importance sampling scheme in which samples drawn from $\priordensity(\vec{\theta})$ are reweighted to approximate expectations under $\priordensitypower(\vec{\theta})$. When constructing an NBE under the loss function~\eqref{eqn:surrogateloss}, the optimization task becomes
\begin{equation}\label{eqn:optimization_task:adjustedprior_bayes}
\vec{\gamma}^*
\equiv 
\underset{\vec{\gamma} \in \Gamma}{\mathrm{arg\,min}} \;
\frac{1}{K} \sum_{k=1}^K \{\priordensity(\vec{\theta}^{(k)})\}^{m-1} L\big\{\vec{\theta}^{(k)}, \vec{h}(\{\vec{Z}^{(k,j)}\}_{j=1}^m; \vec{\gamma}); \kappa\big\},
\end{equation} 
where, independently for each $k$, $\vec{\theta}^{(k)} \sim \priordensity(\vec{\theta})$ and $\{\vec{Z}^{(k,j)}\}_{j=1}^m$ are simulated from $p(\vec{Z} \mid \vec{\theta}^{(k)})$.
The choice of how to account for $\priordensitypower(\vec{\theta})$ depends on whether it is easier to sample from $\priordensitypower(\vec{\theta})$ or to evaluate $\priordensity(\vec{\theta})$. Since it is usually straightforward to evaluate $\priordensity(\vec{\theta})$, we present Algorithm~\ref{alg:neuralEM} in terms of~\eqref{eqn:optimization_task:adjustedprior_bayes}. Finally, when $\priordensity(\vec{\theta}) \propto 1$, the posterior is identical under $\priordensitypower(\vec{\theta})$ and $\priordensity(\vec{\theta})$, and no change to the workflow is required.

\section{Simulation studies}\label{sec:simulationstudies}

We now conduct simulation studies to investigate the strengths and weaknesses of the masking approach (Algorithm~\ref{alg:one-hot}) compared to our EM approach (Algorithm~\ref{alg:neuralEM}), both developed for estimation in the presence of missing data. We refer to an NBE employing the masking approach as a ``Masking NBE'', and an NBE employing our EM approach as an ``EM NBE''. In Section~\ref{sec:generalsetting}, we outline the general setting. In Section~\ref{sec:GP}, we consider a spatial Gaussian-process model and estimate its parameters. In Section~\ref{sec:Potts}, we examine the hidden \citet{Potts_1952} model.

\subsection{General setting}\label{sec:generalsetting}

We conduct our experiments using functionality we have added to the package \ifbool{blind}{\texttt{<redacted for anonymity>}}{\pkg{NeuralEstimators} \citep{NeuralEstimators}}, which is available in \proglang{Julia} and \proglang{R}. We use a workstation with an AMD EPYC 7402 3.00GHz CPU with 128 GB of CPU RAM, and a Nvidia Quadro RTX 6000 GPU with 24 GB of GPU RAM. All subsequent results can be generated using reproducible code at \ifbool{blind}{\texttt{<redacted for anonymity>}}{\url{https://github.com/msainsburydale/NeuralIncompleteData}}. 

To elucidate the differences between the masking approach (Algorithm~\ref{alg:one-hot}) and our proposed EM approach (Algorithm~\ref{alg:neuralEM}), which are greater with data that are high-dimensional, our simulation studies consider spatial models where the data are observed incompletely over a regular grid of size $n = 64^2 = 4096$, and we therefore use a CNN-based architecture for \eqref{eqn:NNmasking} and \eqref{eqn:DeepSets} detailed in Section~\reffsupp{sec:ensemble}. There, we also illustrate the benefits of using an ensemble 
\citep{Hansen_Salamon_1990} of neural networks in the context of neural Bayes estimation; 
 throughout our experiments we use an ensemble of five NBEs for both the masking approach and the EM approach. 

We train our NBEs under the pre-limiting 0--1 loss function \eqref{eqn:surrogateloss}, with $\kappa = 0.1$ and $K=25000$ in both~\eqref{eqn:optimization_task_encoding} (Masking) and~\eqref{eqn:optimization_task:adjustedprior_bayes} (EM). To avoid vanishing gradients early in training, we first pretrain each NBE using the mean-absolute-error loss function before switching to~\eqref{eqn:surrogateloss}. During training, we use the default values of \ifbool{blind}{\texttt{<redacted for anonymity>}}{\pkg{NeuralEstimators}}; specifically, we utilize the Adam optimizer \citep{Kingma_2014_ADAM} with an initial learning rate of 0.0005 and a cosine-annealing learning-rate schedule \citep{Loshchilov_2017_SGDR}. We cease training when the objective function in~\eqref{eqn:optimization_task_encoding} or~\eqref{eqn:optimization_task:adjustedprior_bayes} has not decreased in five consecutive epochs, where an epoch is defined to be one complete pass through the entire training data set when doing stochastic gradient descent to decrease the objective functions. When training the Masking NBEs, we use a missing-completely-at-random (MCAR) model for the missingness mechanism, with the percentage of missing data varying uniformly between 10\% and 50\% across data sets. 
Figure~\reffsupp{fig:risk_profile} shows 
 the value of the objective function in~\eqref{eqn:optimization_task_encoding} or~\eqref{eqn:optimization_task:adjustedprior_bayes} evaluated at the end of each epoch.

Post training, the estimators' statistical efficiencies are compared on 
 unseen test data under both MCAR missingness and a model for the missingness where data are missing in a contiguous block (MICB). 
  This is done to assess the Masking NBEs under the correct and an incorrect specification of the missingness mechanism. 
  We compute empirical root-mean-squared errors (RMSEs) based on simulated data using a new set of 1000 parameter vectors sampled from the prior. For the parameter vectors $\{\vec{\theta}^{(1)}, \dots, \vec{\theta}^{(1000)}\}$, 
\begin{equation}\label{eqn:RMSE}
\text{RMSE}(\hat{\vec{\theta}}(\cdot)) = \bigg\{\frac{1}{1000} \sum_{j=1}^{1000} \| \hat{\vec{\theta}}^{(j)} - \vec{\theta}^{(j)}\|^2\bigg\}^{1/2},
\end{equation}
where, for $j = 1, \dots, 1000$, $\hat{\vec{\theta}}^{(j)}$ is the corresponding estimate using the estimator $\hat{\vec{\theta}}(\cdot)$ from incomplete data $\vec{Z}_1^{(j)} \sim p(\vec{Z}_1 \mid \vec{\theta}^{(j)})$, and recall that $\|\cdot\|$ denotes the Euclidean norm. 
 We use RMSE$_{\text{MCAR}}$ and RMSE$_{\text{MICB}}$ to denote the RMSE of an estimator based on incomplete data simulated under the MCAR and MICB mechanisms, respectively. 
 For the EM NBEs, we set $m=30$ in~\eqref{eqn:EM:MAP}; we use the mean of the prior distribution for the initial estimates, $\EMestimate^{(0)}$; we use a fixed burn-in of $b=5$ in~\eqref{eqn:MCEMaverage}; and we stop the algorithm after 50 iterations or if the maximum elementwise relative change in $\bar{\vec{\theta}}^{(l)}_{b}$, as defined in \eqref{eqn:MCEMaverage}, is less than $0.001$ for three consecutive values of $l$. 

\subsection{Gaussian process model}\label{sec:GP}

In this simulation study, we consider a spatial Gaussian process model, where $\vec{Z} \equiv (Z_{1}, \dots, Z_{n})^\tp$ are data at locations \mbox{$\{\vec{s}_{1}, \dots, \vec{s}_{n}\}$} in a spatial domain $\mathcal{D} \subseteq \mathbb{R}^2$. The data are modeled as spatially-correlated mean-zero Gaussian random variables with Mat\'{e}rn covariance function given by
 \begin{equation}\label{eqn:Matern_covariance_function}
\cov{Z_i}{Z_j} = 
 \sigma^2 \frac{2^{1 - \nu}}{\Gamma(\nu)} \left(\frac{\|\vec{s}_i - \vec{s}_j\|}{\rho}\right)^\nu K_\nu\!\left(\frac{\|\vec{s}_i - \vec{s}_j\|}{\rho}\right) + \tau^2\mathbb{I}(i = j); \quad i, j = 1, \dots, n,
 \end{equation} 
where $\sigma^2$ is a variance parameter; $\Gamma(\cdot)$ is the gamma function; $K_\nu(\cdot)$ is the modified Bessel function of the second kind of order $\nu$; $\rho > 0$ and $\nu > 0$ are range and smoothness parameters, respectively; and $\tau^2$ is a fine-scale variance parameter. 

In this example, we take the spatial domain to be $\mathcal{D} \equiv [0, 1] \times [0, 1]$, and we simulate complete data on a regular square grid of size $n = 64^2 = 4096$. For computational tractability, we use the package \pkg{GpGp} to simulate training data via the \citet{Vecchia_1988} approximation. We implement the Vecchia approximation using a ``maxmin'' ordering of the locations, with a maximum of 30 neighbors assigned to each location; see \citet{Guinness_2018} for further details. The parameters to be estimated are $\vec{\theta} \equiv (\tau, \rho)^\tp$, and we fix $\nu = 1$ and $\sigma^2 = 1$. We assume that $\tau$ and $\rho$ are independent \textit{a priori}, and we use the priors \mbox{$\tau \sim \Unif{0.01}{1}$} and \mbox{$\rho \sim \Unif{0.03}{0.35}$}.  

We compare the estimators using the time taken for their training and for estimation for a single data set post-training  (computational efficiency); and their empirical RMSEs (statistical efficiency). 
 Since the likelihood function is available for this model, we compare the two competing NBEs to the MAP estimator that numerically maximizes the incomplete-data posterior density. 
Our results, summarized in  Table~\ref{tab:GP} and Figure~\ref{fig:GP:missingness}, highlight several important properties of the estimators. 

 \begin{table}[t!]
\centering
\caption{The 
training time, 
estimation time for a single test data set, and empirical RMSE under two missingness models for 
three 
estimators of the parameters of the Gaussian process model (Section~\ref{sec:GP}). 
}\label{tab:GP}
\begin{tabular}{lcccc}
  \hline
 Estimator         & Training time (mins) & Estimation time (s) & RMSE$_{\text{MCAR}}$ & RMSE$_{\text{MICB}}$  \\ 
  \hline
 MAP               &  --                    & 1.12                & \textbf{0.015}    & \textbf{0.014}           \\ 
 EM NBE            & \textbf{21.6}                   & 0.39                & \textbf{0.015}    & 0.015           \\ 
 Masking NBE       & 25.3                   & \textbf{0.01}       & 0.016             & 0.202                    \\ 
 \hline
\end{tabular}
\end{table} 

\begin{figure}[t!]
\centering
\includegraphics[width = \linewidth]{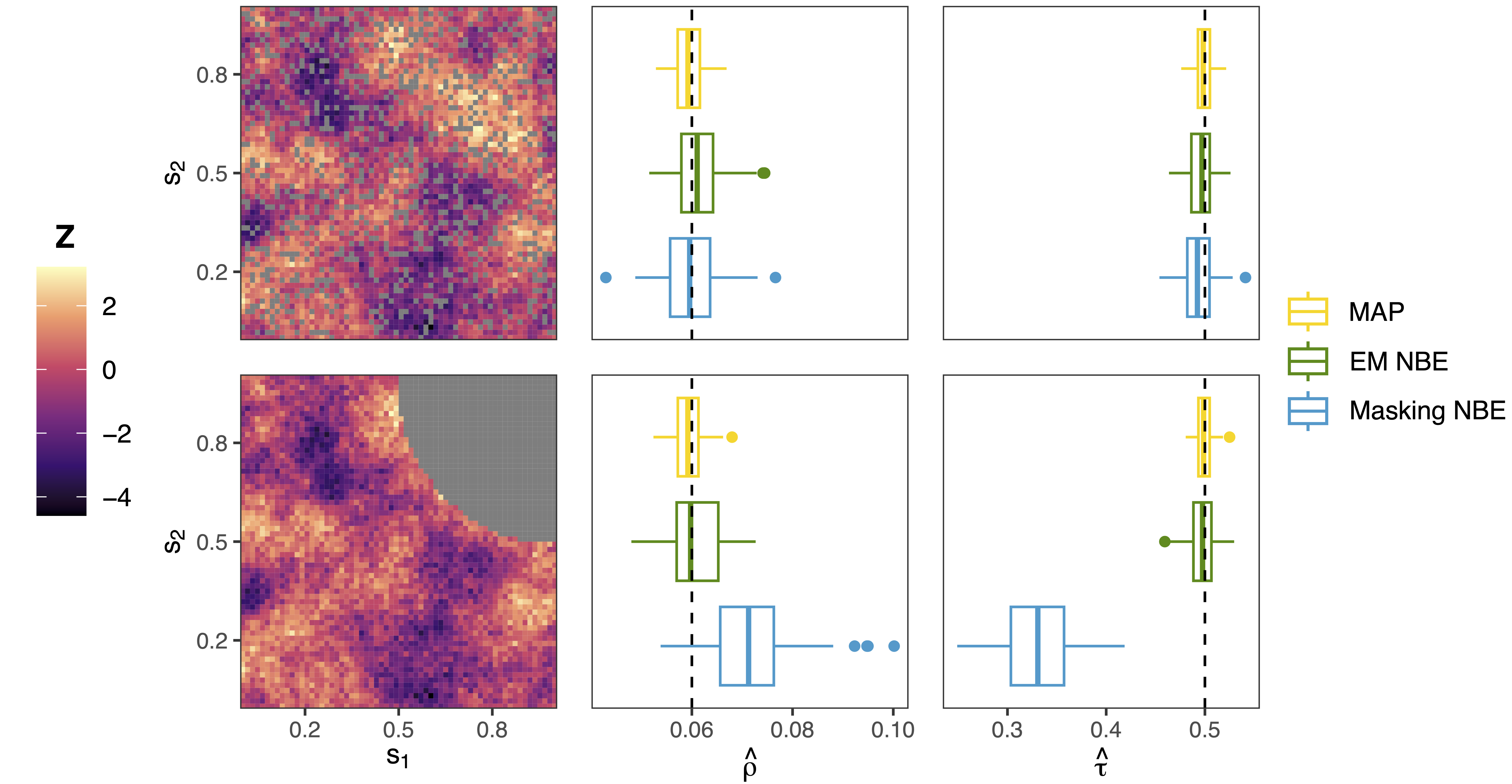}  
\caption{Spatial data (first column) where the missingness is of type MCAR (first row) or MICB (second row) with missingness shown in gray, and corresponding empirical  distributions (second and third columns) for three estimators of the parameters of the Gaussian process model (Section~\ref{sec:GP}). True parameter values are shown as a dashed vertical line.}\label{fig:GP:missingness} 
\end{figure}

First, the EM NBE is agnostic to the missingness pattern and it performs well under both MCAR and MICB data, as illustrated by the similarity of its sampling distribution (Figure~\ref{fig:GP:missingness}) and RMSE (Table~\ref{tab:GP}) to those of the MAP estimator. 
 By contrast, the Masking NBE has substantial bias and inflated RMSE on MICB test data (recall that the Masking NBE was trained with MCAR data). The almost-as-good RMSE performance of the EM NBE relative to the MAP estimator indicates that Algorithm~\ref{alg:neuralEM} in the Estimation stage converges to a suitable point estimate across most, if not all, data sets; convergence for a single data set is illustrated in Figure~\reffsupp{fig:convergence:GP}.

Second, the NBEs represent a trade-off between computational efficiency and statistical efficiency. 
 The speed of the Masking NBE is due to it not requiring likelihood computation or conditional simulation, and because it is not an iterative algorithm. For the spatial Gaussian process model, the MAP estimate is obtained straightforwardly by numerically maximizing the unnormalized posterior density, which is available in closed form. Hence, for this model, the EM NBE provides only a moderate reduction in estimation time compared to the MAP estimator (Table~\ref{tab:GP}). 


These results provide empirical evidence that NBEs can come close to the statistical efficiency of a gold-standard likelihood-based estimator. 
 It will be seen in the next section that NBEs are even more beneficial when the incomplete data likelihood function is unavailable in closed form or is computationally intractable. 

\subsection{Hidden Potts model}\label{sec:Potts}

We now consider a hidden Markov random field (e.g., \citeauthor{Besag_1974}, \citeyear{Besag_1974};  \citeauthor{Cressie_1993_stats_for_spatial_data}, \citeyear{Cressie_1993_stats_for_spatial_data}, Ch.~6; \citeauthor{Rue_Held_2005_GMRF:TaA}, \citeyear{Rue_Held_2005_GMRF:TaA}) based on the so-called \cite{Potts_1952} model.
  Consider a regular grid of pixels indexed by $i = 1, \dots, n$, where each pixel takes on a label $Y_i$ from a finite set of discrete states $\mathcal{Q} \equiv \{1, \dots, Q\}$. Then, the Potts model is specified through the conditional distributions,  
\begin{equation}\label{eqn:Potts}
 {\rm{Pr}}(Y_i = y \mid \vec{Y}_{\backslash i}, \beta) \propto \exp\Big\{\beta \sum_{j \in \mathcal{N}_i} \mathbb{I}(Y_j = y)\Big\}, \; \text{for } y \in \mathcal{Q}, 
\end{equation}
where $\vec{Y}_{\backslash i}$ denotes all pixel labels excluding the $i$th pixel; $\beta > 0$ is a parameter controlling the strength of spatial dependence; and $\mathcal{N}_i$ contains the indices of the ``neighbors'' of pixel $i$. Here, we consider a Potts model modified to account for edge effects. Specifically, we let the neighbors $\mathcal{N}_i$ be the adjacent pixels of pixel $i$, with four neighbors for interior pixels, three for edge pixels, and two for corner pixels. The Potts model exhibits a phase transition at the critical parameter value $\beta_c = \log(1 + \sqrt{Q})$ \citep{Moores_2021_bayesImageS}, transitioning from disorder when $\beta < \beta_c$ (where most neighboring pixels do not have the same label) to order when $\beta > \beta_c$ (most neighboring pixels have the same label). 
  Its likelihood function, obtained from \eqref{eqn:Potts}, involves a normalizing constant that is computationally intractable for moderate-to-large $n$. 
  In what follows in this section, we focus on a hierarchical extension of the Potts model, the hidden Potts model. 

The hidden Potts model extends the Potts model to settings where the labels are not directly observed. The observed data $\vec{Z} \equiv (Z_1, \dots, Z_n)^\tp$ are assumed 
 conditionally independent given the latent labels $\vec{Y} \equiv (Y_1, \dots, Y_n)^\tp$, with a label-specific observation distribution:  
\begin{equation}\label{eqn:HiddenPottsObservation}
p(Z_i \mid Y_i = y, \vec{\lambda}_y); \quad i =  1, \dots, n,
\end{equation}
where each label $y \in \mathcal{Q}$ is associated with observation-distribution parameters $\vec{\lambda}_y$ that we collect in $\vec{\lambda} \equiv (\vec{\lambda}_1^\tp, \dots, \vec{\lambda}_Q^\tp)^\tp$. The parameters of the hidden Potts model are $\vec{\theta} \equiv (\beta, \vec{\lambda}^\tp)^\tp$. 
 As with the standard Potts model, likelihood-based inference in the hidden Potts model is computationally prohibitive except for grids with a small number of grid cells.

The general hidden Potts model given by \eqref{eqn:Potts} and \eqref{eqn:HiddenPottsObservation} is often made concrete by specifying Gaussian data distributions in~\eqref{eqn:HiddenPottsObservation}. 
 We focus on this common case in the current simulation study, while in Section~\ref{sec:application}, we consider a more complex hidden Potts model in which the observations follow either point-mass or Beta distributions. Specifically, 
 in the current simulation study, 
 we assume a hidden Potts model with $Q=3$ states and observation distributions
\begin{equation}\label{eqn:HiddenPottsObservationGaussian}
p(Z_i \mid Y_i = y, \vec{\lambda}_y) = \text{Gau}(\mu_y, \sigma_y^2); \quad i =  1, \dots, n,
\end{equation}
where each label $y \in \mathcal{Q}$ is associated with its own mean $\mu_y \in \mathbb{R}$ and variance $\sigma^2_y > 0$. 
With three hidden states, the phase transition occurs at $\beta_c = 1.005$. 
We adopt the priors $\beta \sim \Unif{0}{1.5}$, $\mu_1 \sim \text{\Gau}(-1, 0.3^2)$, $\mu_2 \sim \text{\Gau}(0, 0.3^2)$, $\mu_3 \sim \text{\Gau}(1, 0.3^2)$, and $\sigma_y \sim \Unif{0}{1/3}$ for $y \in \{1, 2, 3\}$, and we impose the identifiability constraint $\mu_1 < \mu_2 < \mu_3$. 
 We estimate the seven unknown parameters $\vec{\theta} \equiv (\beta, \mu_1, \mu_2, \mu_3, \sigma_1, \sigma_2, \sigma_3)^\tp$ based on incomplete data on a square grid of size $n = 64^2 = 4096$. 
   To simulate realizations $\vec{Y}$, 
 we use the \citet{Swendsen_Wang_1987} algorithm implemented in 
  \pkg{bayesImageS} 
  \citep{Moores_2021_bayesImageS}. 
 During the Estimation stage of Algorithm~\ref{alg:neuralEM}, we simulate conditionally on a partially observed, noisy field using 
 Gibbs sampling (see Section~\reffsupp{sec:MCMCHiddenPotts}). 

  Table~\ref{tab:Potts} and Figure~\ref{fig:Potts} report our results 
  (in Figure~\ref{fig:Potts} we focus on the parameter $\beta$). As with the Gaussian process model, the EM NBE is agnostic to the missingness pattern, and performs well under both missingness models considered in this experiment. By contrast, the Masking NBE is biased for most values of $\beta$ when the missingness model is misspecified. 
   The performance of the EM NBE indicates that Algorithm~\ref{alg:neuralEM} in the Estimation stage converges to a suitable point estimate across most, if not all, data sets; convergence for a single data set is illustrated in Figure~\reffsupp{fig:convergence:Potts}. 
 
\begin{table}[t!]
\centering
\caption{
The 
training time, 
 estimation time for a single test data set, and empirical RMSE for 
 two
 estimators of the parameters of the hidden Potts model (Section~\ref{sec:Potts}). 
}\label{tab:Potts}
\begin{tabular}{lccccc}
 \hline
 Estimator  & 
 Training time (mins) & 
 Estimation time (s) & RMSE$_{\text{MCAR}}$ & RMSE$_{\text{MICB}}$  \\ 
 \hline
 EM NBE        &  \textbf{20.4}   & 0.23              &  \textbf{0.037}  & \textbf{0.042} \\  
 Masking NBE   &  28.1            & \textbf{0.01}     &  \textbf{0.037}  & 0.166 \\ 
 \hline
\end{tabular}
\end{table}

\begin{figure}[t!]
\vspace{1em}
\centering 
\includegraphics[width = \linewidth]{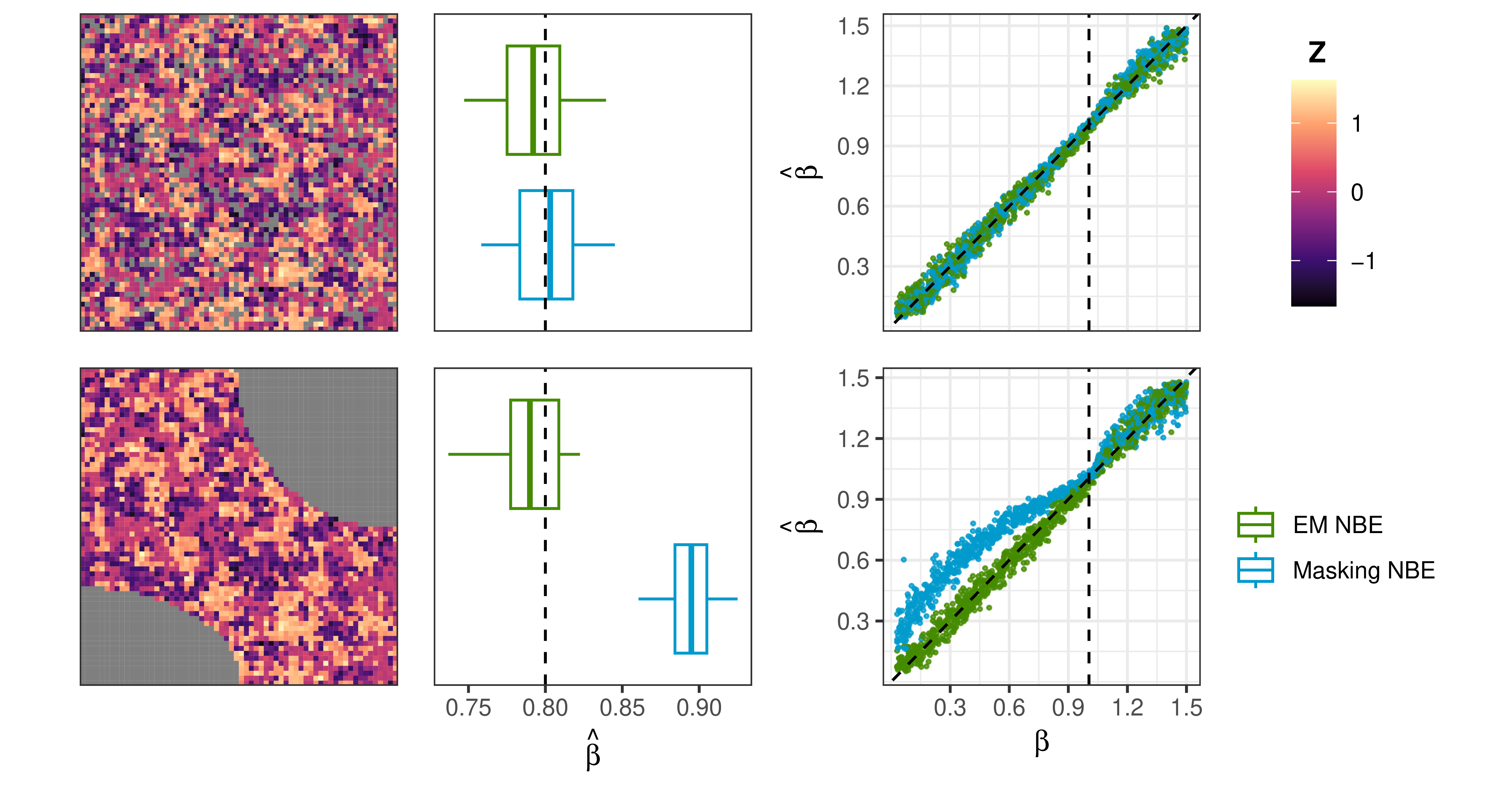}  
\caption{Spatial data (first column) simulated from the hidden Potts model of Section~\ref{sec:Potts}, where the missingness is of type MCAR (first row) or MICB (second row); empirical distributions (second column) for two estimators of the parameter $\beta$, with $\beta = 0.8$ fixed (dashed vertical line); and estimates versus true values (third column) for many different values of $\beta$, with the critical parameter value $\beta_c = 1.005$ demarcated by a dashed vertical line. 
}\label{fig:Potts} 
\end{figure}

Finally, this example highlights the advantages of NBEs (and related neural inferential methods) over another popular simulation-based method of inference, approximate Bayesian computation \citep[ABC;][]{Lintusaari_2017_ABC_review, Sisson_2018_ABC_handbook}. 
 ABC typically relies on user-defined summary statistics, which are difficult to construct for latent variable models such as the hidden Potts model. In contrast, NBEs do not require user-defined summary statistics and can be applied directly to high-dimensional data. To provide a comparison in a setting where informative summary statistics are available, in Section~\reffsupp{sec:additionalsimulationstudies} we present an additional simulation study using a spatial version of the generalized hyperbolic distribution. Even in this case where informative summary statistics are available, we find that ABC leads to less statistically efficient estimates than the EM NBE. 

\section{Application}\label{sec:application}


Here, we consider a remote-sensing application with Arctic sea-ice data. These data are high-dimensional with a missingness mechanism that is difficult to model, but these challenges can be overcome using the EM NBE (Algorithm~\ref{alg:neuralEM}). 

Arctic sea ice plays an important role in regulating our climate: it acts as a reflective surface that reduces the amount of solar energy absorbed by Earth. Melting sea ice exposes darker ocean water, thereby further accelerating the melting process due to an albedo-ice feedback effect. 
Changes in Arctic sea-ice area and thickness also affect atmospheric circulation and ocean currents, which can influence weather patterns worldwide \citep{Cvijanovic_2017}. 
Further, Arctic sea ice provides vital habitats for species such as polar bears and seals, and its loss can disrupt fragile ecosystems, thereby affecting biodiversity, food webs, and fisheries \citep{Meier_2014}. 
Understanding the temporal evolution of Arctic sea ice is therefore crucial for informing policies aimed at mitigating the impacts of climate change, managing resources sustainably, and protecting vulnerable ecosystems \citep{UN_SDGs}. 

In this application, we consider spatial data of Arctic sea-ice proportion, that is, the proportion of sea ice in grid cells at given spatial locations (also commonly referred to as sea-ice concentration), produced by the National Oceanic and Atmospheric Administration (NOAA) as part of their National Snow and Ice Data Center’s (NSIDC) Climate Data Record \citep[CRD;][]{arctic_sea_ice}. The data are derived from passive microwave remote sensing retrievals from the Nimbus 7 satellite and the F8, F11, F13, and F17 satellites of the Defence Meteorological Satellite Program, projected onto 25km $\times$ 25km grid cells within a region of the Northern Hemisphere spanning longitudes 180°E to 180°W and latitudes at or above 60$^\circ$N \citep{Zhang_2020_spatio-temporal_Arctic_sea_ice, arctic_sea_ice}. Arctic-sea-ice cover typically reaches its annual minimum in 
September \citep{Parkinson_2014}, and we therefore base our analysis on the ice cover on the first day of September in each year. 
 Our 
  data set comprises 45 spatial images (one image for each year between 1979 and 2023), with each image containing $199\times 219 = 43581$ grid cells, and we analyze each year separately. 
 
Figure~\ref{fig:sea_ice:missing} shows that the data are incomplete and that the missingness patterns are relatively complicated. Here, missingness occurs for several reasons, including cloud cover and unpredictable issues with the remote-sensing instrument \citep{arctic_sea_ice}. The data are also subject to a more consistent form of missingness around the North Pole: this area, called the Arctic Pole Hole, changes in size over time as it is a function of both the remote-sensing instrument and the prevailing atmospheric conditions \citep{arctic_sea_ice}. 

\begin{figure}[t!]
\centering
\includegraphics[width = \linewidth]{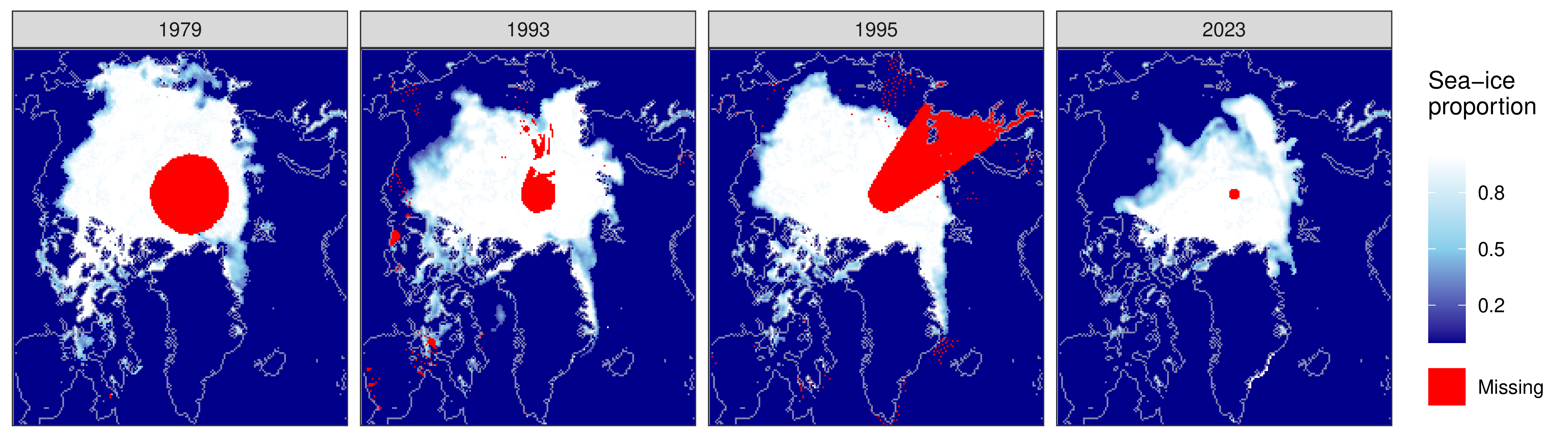}
\caption{Arctic sea-ice data from the first day of September for the years 1979, 1993, 1995, and 2023. Faint gray lines denote coastlines, with Greenland appearing at the bottom. The data are subject to both random sources of missingness (e.g., cloud cover) and more systematic sources of missingness due to remote-sensing limitations (e.g., the Arctic Pole Hole).}\label{fig:sea_ice:missing} 
\end{figure}

Previous studies of Artic sea-ice \citep[e.g.,][]{Parkinson_2014, Zhang_2020_spatio-temporal_Arctic_sea_ice} often apply a fixed threshold (e.g., 15\%) to classify each grid cell as either ``ice'' or ``not ice''. This binarization simplifies modeling but it has important limitations: the threshold is arbitrary, it discards valuable information from the underlying data (which gives proportions of sea-ice cover), and it can bias estimates of sea-ice area (the total ice-covered region). Instead, we model the sea-ice proportions directly using a hidden Potts model (see \eqref{eqn:Potts}) with $Q=4$ states, spatial-dependence parameter $\beta > 0$, and the following observation distributions:
\begin{alignat*}{2}
p(Z_i \mid Y_i = 1) &= \delta_0(Z_i),              \qquad &&p(Z_i \mid Y_i = 2) = \text{Beta}(Z_i; a_1, b_1), \\
p(Z_i \mid Y_i = 3) &= \text{Beta}(Z_i; a_2, b_2), \qquad &&p(Z_i \mid Y_i = 4) = \delta_1(Z_i),
\end{alignat*}
where $\delta_x(\cdot)$ denotes the Dirac delta function that models point mass at $x \in \mathbb{R}$, $\text{Beta}(z; a, b)$, $z \in (0, 1)$, denotes the Beta density function with shape parameters $a > 0$ and $b > 0$, and for identifiability we impose ordering of the means: $a_1(a_1 + b_1)^{-1} < a_2(a_2 + b_2)^{-1}$. This hidden Potts model captures the multimodal structure evident in the empirical histogram of sea-ice proportions (see Figure~\reffsupp{fig:sea_ice:histograms}). The inclusion of two Beta components allows flexible modeling of the continuous observations in $(0,1)$, which exhibit strong spatial dependence (see, e.g., Figure~\reffsupp{fig:sea_ice:predictions_and_Y}). 
 Conditional simulation, which is required by our EM NBE (Algorithm~\ref{alg:neuralEM}) and for prediction at grid cells with missing pixels, can be done using MCMC, 
 as described in Section~\reffsupp{sec:MCMCHiddenPotts}. 

In our hidden Potts model, the unknown parameters are $\vec{\theta} \equiv (\beta, a_1, a_2, b_1, b_2)^\tp$. 
 We use a relatively uninformative prior for the spatial-dependence parameter $\beta > 0$, namely, $\beta \sim \Unif{0}{1.5}$. For the Beta shape parameters, we adopt the informative priors 
$a_1 \sim \Unif{2}{5}$, 
$b_1 \sim \Unif{2}{5}$, 
$a_2 \sim \Unif{2}{5}$, 
and $b_2 \sim \Unif{0}{1}$. 
Given the large data size, the computationally intractable model, and the complicated missingness mechanisms, our EM NBE (Algorithm~\ref{alg:neuralEM}) is well suited for inference in this application. We trained the NBE using the same settings given in Sections~\ref{sec:generalsetting} and~\ref{sec:Potts}, with a total training time (including data simulation) of 32 minutes. We then applied the EM NBE to each of the 45 images. The total time for estimating the 45 parameter vectors $\{\vec{\theta}_t : t = 1, \dots, 45\}$, for each of the 45 years, was 44 seconds. 

Figure~\ref{fig:sea_ice:results}, left panel, shows estimates $\{\hat{\beta}_t : t = 1, \dots, 45\}$, as well as 95\% pointwise confidence intervals (Figure~\reffsupp{fig:sea_ice:all_parameters} presents the corresponding plots for all parameters). To obtain these intervals, we used a separate parametric bootstrap for each year. Specifically, for each year, we simulated 100 data sets from the fitted model, removed data from the same grid cells that were missing in the observed data set, and then applied the EM NBE again to each of the 100 simulated (incomplete) data sets. All estimates of $\beta$ are larger than the critical value $\beta_c = \log(1 + \sqrt{Q}) = 1.099$, confirming a strong tendency for neighboring grid cells to share the same label. 

\begin{figure}[t!]
\centering
\includegraphics[width = \linewidth]{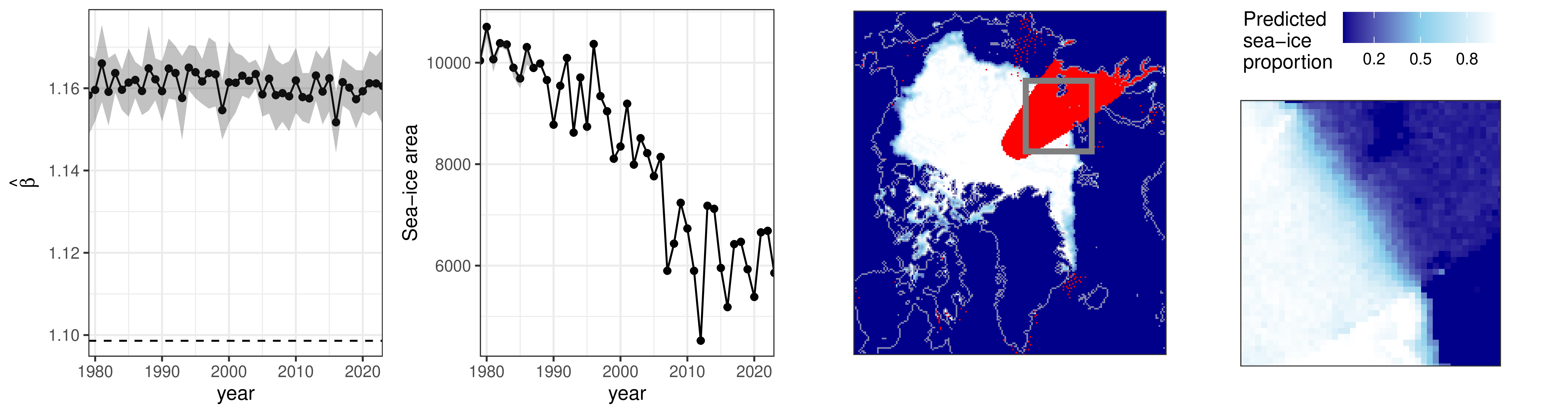}  
\caption{Analysis of Arctic sea-ice data. (Left) EM NBE estimates of the spatial dependence parameter $\beta$ in \eqref{eqn:Potts} versus year. Shaded areas display 95\% pointwise confidence intervals obtained using parametric bootstrap sampling, while the dashed horizontal line indicates the critical value $\beta_c = \log(1 + \sqrt{Q}) = 1.099$. (Center-left) Predictions of sea-ice area versus year. Note that 95\% prediction intervals are plotted, but these are not visible due to the low uncertainty in the predictions. (Center-right) Arctic sea-ice data for September 1, 1995. Faint gray lines denote coastlines, with Greenland appearing at the bottom. (Right) Predicted sea-ice proportion for all grid cells within the gray box of the center-right panel.}\label{fig:sea_ice:results} 
\end{figure}

 
Figure~\ref{fig:sea_ice:results}, center-left panel, shows predictions of sea-ice area as a function of year. To make these predictions, we imputed missing values by simulating 100 times from the fitted model (of the corresponding year) conditionally on the observed data, and we treated the mean of these conditional simulations as the prediction. In line with the general scientific consensus, our analysis indicates that sea-ice area is decreasing dramatically over time, nearly halving over the study period. The center-right and right panels of Figure~\ref{fig:sea_ice:results} show the incomplete data from September 1, 1995, and the resulting predictions of sea-ice proportion in a box containing a subset of the incomplete data (see Figures~\reff{fig:sea_ice:sims}~and~\reffsupp{fig:sea_ice:predictions_and_Y} for individual conditional simulations and predictions of both the data and the hidden labels). Due to the strong spatial dependence, the predicted proportions in the interior of the ice sheet and in regions far from ice are close to 1 and 0, respectively; however, there is greater variability within the region corresponding to the sea-ice boundary. Post-training, inference (including bootstrap) for all 45 years of data took a total of 1.7 hours clock time. This can be compared to \citeauthor{Zhang_2020_spatio-temporal_Arctic_sea_ice}'s (\citeyear{Zhang_2020_spatio-temporal_Arctic_sea_ice}) fully Bayesian inference who, with a 
 binary model and with a smaller data set, needed more than a day of clock time on a comparable high-performance computer.

Finally, we note evidence of spatial nonstationarity in these Arctic sea-ice data: see Figure~\reffsupp{fig:sea_ice_nonstationarity}, which summarizes our analysis over a subdomain of interest, namely, the Canadian Arctic Archipelago. Future work in this application might hence consider spatially-varying versions of $a_1$, $a_2$, $b_1$, and $b_2$. The advantage of the proposed EM NBE is that more complicated (nonstationary) models will not alter the workflow, and that more complicated models can be fitted and assessed with relative ease. 

\section{Conclusion}\label{sec:conclusion}

Incomplete data are ubiquitous in applications of AI, arising from various sources such as cloud cover, equipment malfunctions, and data corruption. In this article, we focus on the problem of missing data in the context of neural Bayes estimation, which uses neural networks to map data to point estimates of parameters. We first discuss and give new insights on the masking approach of \cite{Wang_2022_neural_missing_data} (Algorithm~\ref{alg:one-hot}), where inference is performed on an extended data set containing the observed data and auxiliary variables that encode the missingness pattern. The masking approach has advantages: It can be used to quickly generate Bayes estimates under general loss functions; it only requires marginal (i.e., unconditional) simulation from the data model; and it is theoretically well motivated (see Theorems~\reff{thm:masking_likelihood} and~\reff{thm:invariance} in Section~\reffsupp{app:likelihoodequivalence}). 
However, its major disadvantage is that it requires a stochastic model for the missingness mechanism, which can lead to bias and statistical inefficiency when the model is wrong and the data are high-dimensional. We therefore propose an alternative approach that is based on the MCEM algorithm, where the E- and M-steps are approximated using conditional simulation and an NBE that returns the (approximate) MAP estimate from the conditionally-simulated complete data. Our EM approach (Algorithm~\ref{alg:neuralEM}) is likelihood-free, in the sense that it does not require evaluation or knowledge of the likelihood function; it is fast, since it does not require numerical optimization at each iteration; and, in contrast to the masking approach, it is agnostic to the missingness pattern of the observed data. Moreover, our research can be viewed as a prototype problem that indicates how improvements could be made in AI by introducing statistical inferential tools. 

Our EM approach to neural Bayes estimation with incomplete data 
 relies 
on conditional simulation, which can be a computational bottleneck for certain models. Future research will explore the use of approximate conditional simulation \citep[e.g.,][]{Wu_2023, Simkus_2023_VGI, Simkus_Gutmann_2025} to extend the applicability of our EM approach to models for which conditional simulation is intractable or computationally prohibitive. While we focus on point estimation, our insights on the masking approach extend to methods that use it to approximate the full posterior distribution of the parameters 
  and more general inference frameworks \citep[e.g.,][]{gloeckler2024allinone}. In these contexts, the data-augmentation algorithm of \cite{Tanner_Wong_1987} offers advantages analogous to those of our EM approach in the point-estimation setting.  Finally, neural Bayes estimation requires specifying a potentially large set of hyperparameters related to the neural network's architecture and its training. Although automated search methods 
   \citep[see, e.g.,][]{Elsken_2019}  
  can be useful, 
  and \citet{Rodder_2025} provide important initial theoretical results, further developments are needed to establish practical guidelines for hyperparameter selection.

\ifbool{blind}{}{%
\section*{Acknowledgements}
\ifbool{blind}{\texttt{<redacted for anonymity>}}{The authors thank Matthew Moores for discussion on, and code relating to, the Potts model. 
}
\section*{Funding}
\ifbool{blind}{\texttt{<redacted for anonymity>}}{%
MSD’s research was supported by an Australian Government Research Training Program Scholarship. MSD, AZM, and RH were supported by the KAUST Opportunity Fund Program ORFS-2023-OFP-5550.2. AZM’s research was supported by Australian Research Council (ARC) Discovery Early Career Research Award, DE180100203.  AZM and NC were supported by ARC Discovery Project (DP), DP190100180. RH was also partially supported by  KAUST Office of Sponsored Research (OSR) under Award No. OSR-CRG2020-4394. This material is based upon work supported by the Air Force Office of Scientific Research under award number FA2386-23-1-4100 (AZM and NC).}}

\ifbool{doublespace}{%
\begin{spacing}{1.5}%
{\small \putbib[bibliography]}%
\end{spacing}
}{%
{\small \putbib[bibliography]}%
}
\end{bibunit}

\ifbool{arxiv}{

\ifbool{arxiv}{%
    \begin{center}
    \supptitle
    \vspace{1em}
    \end{center}}{%
    \maketitle
    \setcounter{page}{1}
}

\supplement

\begin{bibunit}[apalike] 

\noindent In Section~\ref{app:likelihoodequivalence}, we provide a theoretical rationale for the masking approach of \citet{Wang_2022_neural_missing_data}. In Section~\ref{sec:supp:convergenceofapproximateMCEM}, we discuss the asymptotic behavior of the Monte Carlo EM (MCEM) algorithm when the conventional update is approximate. In Section~\ref{sec:approximations01loss}, we consider several continuous approximations of the 0--1 loss function that may be used when training an NBE to approximate the MAP estimator. 
In Section~\ref{sec:additionalsimulationstudies}, we conduct an additional simulation experiment using a spatial version of the generalized hyperbolic (GH) distribution. 
In Section~\ref{sec:ensemble}, we propose and illustrate the benefits of using an ensemble of deep neural networks in the context of neural Bayes estimation. 
 In Section~\ref{sec:MCMCHiddenPotts}, we discuss Markov chain Monte Carlo (MCMC) sampling in hidden Potts models. Finally, in Section~\ref{sec:additionalfigures}, we provide additional figures supporting our methodology. 
 
\section{Rationale for the masking approach}\label{app:likelihoodequivalence}

Here, we use a sufficiency argument to show that the masking approach of \citet{Wang_2022_neural_missing_data} does not lead to any loss of information on $\vec{\theta}$. 

\begin{theorem}\label{thm:masking_likelihood} 
Let the complete data $\vec{Z} \in \mathbb{R}^n$ be distributed according to a family of probability distributions indexed by $\vec{\theta}$. Partition $\vec{Z}$ into the components $\vec{Z}_1$ and $\vec{Z}_2$ containing the observed and missing elements, respectively, where it is assumed that there is at least one observation. Define the ordered set $\mathcal{I}_1 \equiv \{i: Z_i \text{ is observed}\}$ such that $\vec{Z}_1 = (Z_i : i \in \mathcal{I}_1)^\tp$, and let $\vec{U}$ and $\vec{W}$ be defined as in Equations~\eqreff{eqn:paddeddata}~and~\eqreff{eqn:binarymask} of the main text. Then,
$$\vec{T}(\vec{Z}_1, \mathcal{I}_1) \equiv (\vec{U}, \vec{W}),$$ 
is a sufficient statistic for $\vec{\theta}$. 
\end{theorem}

\begin{proof}
Given that there is at least one observation, $\mathcal{I}_1$ is non-empty and $\vec{Z}_1$ is well defined. Now, $\vec{Z}_1$ and $\mathcal{I}_1$ represent all of the available information from which to make inference on $\vec{\theta}$, that is, together they are sufficient for $\vec{\theta}$, and hence we need only show that $\vec{T}(\cdot, \cdot)$ defines a one-to-one mapping from the space of $(\vec{Z}_1, \mathcal{I}_1)$ to that of $(\vec{U}, \vec{W})$, since any one-to-one transformation of a sufficient statistic is itself sufficient \citep[][pg.~280]{Casella_Berger_2001_Statistical_Inference}. 
 First, note that the construction of $\vec{U}$ using \eqreffmain{eqn:paddeddata} can be equivalently written as $\vec{U} = \vec{Z} \odot \vec{W} + c(\vec{1} - \vec{W})$, where $\odot$ denotes elementwise multiplication, and $\vec{1}$ denotes the vector of 1s of appropriate dimension. Now, for any $\vec{Z}_1^a$ and $\mathcal{I}_1^a$, any $\vec{Z}_1^b$ and $\mathcal{I}_1^b$, and any $c \in \mathbb{R}$,
\begin{align*}
 \vec{T}(\vec{Z}_1^a, \mathcal{I}_1^a) = \vec{T}(\vec{Z}_1^b, \mathcal{I}_1^b) 
 &\implies (\vec{U}^a, \vec{W}^a) = (\vec{U}^b, \vec{W}^b)\\
 &\implies \vec{U}^a = \vec{U}^b \text{ and } \vec{W}^a = \vec{W}^b \neq \vec{0}\\
 &\implies \vec{Z}^a \odot \vec{W}^a + c(\vec{1} - \vec{W}^a) = \vec{Z}^b \odot \vec{W}^a + c(\vec{1} - \vec{W}^a) \text{ and } \mathcal{I}_1^a = \mathcal{I}_1^b \\
 &\implies \vec{Z}^a \odot \vec{W}^a = \vec{Z}^b \odot \vec{W}^a \text{ and } \mathcal{I}_1^a = \mathcal{I}_1^b\\
 &\implies \vec{Z}_1^a = \vec{Z}_1^b \text{ and } \mathcal{I}_1^a = \mathcal{I}_1^b,
\end{align*}
since $\vec{W}^a \neq \vec{0}$. Therefore, $\vec{T}(\cdot, \cdot)$ defines a one-to-one mapping from the space of $(\vec{Z}_1, \mathcal{I}_1)$ to that of $(\vec{U}, \vec{W})$ and, hence, $\vec{T}(\vec{Z}_1, \mathcal{I}_1)$ is sufficient for $\vec{\theta}$. 
\end{proof}

\noindent Theorem~\ref{thm:masking_likelihood} is important in establishing that the masking approach makes full use of the information contained in the data. However, using $\vec{U}$ and $\vec{W}$ in place of $\vec{Z}_1$ and $\mathcal{I}_1$ in a neural Bayes estimation setting requires assigning a distribution to $\mathcal{I}_1$ (or equivalently, to $\vec{W}$). Using an argument that closely follows that of \citet[][Theorem~1]{Richards_2023_censoring} and \citet[][Theorem~1]{Sainsbury-Dale_2023_GNNs}, we now show that, under mild conditions, the Bayes estimator is invariant to this distribution. 

\begin{theorem}\label{thm:invariance}
Let the complete data $\vec{Z} \in \mathbb{R}^n$ be distributed according to a family of probability distributions indexed by $\vec{\theta}$. Partition $\vec{Z}$ into the components $\vec{Z}_1$ and $\vec{Z}_2$ containing the observed and missing elements, respectively. 
 Let $\vec{U} \in \mathcal{U} \subseteq \mathbb{R}^n$ and $\vec{W} \in \mathcal{W} = \{0, 1\}^n$ be defined as in Equations~\eqreff{eqn:paddeddata}~and~\eqreff{eqn:binarymask} of the main text. Let $L:\Theta \times \Theta \to [0, \infty)$ denote a loss function, and assume that the Bayes estimate $\hat{\vec{\theta}}^\star$ under this loss function exists, is unique, and has finite posterior expected loss
 for all 
 $\vec{U} \in \mathcal{U}$ and $\vec{W} \in \mathcal{W}$. Then 
 the Bayes estimator $\hat{\vec{\theta}}^\star(\vec{U}, \vec{W})$ that returns $\hat{\vec{\theta}}^\star$ for any $\vec{U} \in \mathcal{U}$ and $\vec{W} \in \mathcal{W}$ is invariant to the distribution $p_{\vec{W}}(\cdot)$ of $\vec{W}$, provided  
 \begin{enumerate}[(i)]
    \item $p_{\vec{W}}(\vec{w}) > 0$ for all $\vec{w} \in \mathcal{W}$, 
    \item $\vec{W}$ and $\vec{\theta}$ are independent.
\end{enumerate}
\end{theorem}

\begin{proof}
For all fixed $\vec{U} \in \mathcal{U}$ and $\vec{W} \in \mathcal{W}$, the Bayes estimate $\hat{\vec{\theta}}^\star$ minimizes the posterior expected loss; that is, 
\begin{equation}
\hat{\vec{\theta}}^\star = \argmin_{\hat{\vec{\theta}}} \int_\Theta L(\vec{\theta}, \hat{\vec{\theta}})p_{\vec{\theta} \mid \vec{U}, \vec{W}}(\vec{\theta} \mid \vec{U}, \vec{W})\d\vec{\theta}.\label{eq:Bayesestimate}
\end{equation}
 Consider now the Bayes estimator  $\hat{\vec{\theta}}^\star(\vec{U}, \vec{W})$ that returns the Bayes estimate for any fixed $\vec{U} \in \mathcal{U}$ and $\vec{W} \in \mathcal{W}$ \citep[see][for a proof of the existence of a Borel measurable Bayes estimator under mild conditions]{Brown_1973}. Since the posterior expected loss is, by assumption, bounded and nonnegative for all $\vec{U} \in \mathcal{U}$ and $\vec{W} \in \mathcal{W}$, we have that 
\begin{equation}\label{eq:Bayesmin}
  \hat{\vec{\theta}}^\star(\cdot,\cdot) = 
 \argmin_{\hat{\vec{\theta}}(\cdot,\cdot)}
  \sum_{\vec{w} \in \mathcal{W}} \tilde{p}_{\vec{W}}(\vec{w})
  \int_\mathcal{U}
  \int_\Theta 
  L(\vec{\theta}, \hat{\vec{\theta}}(\vec{u},\vec{w}))p_{\vec{\theta} \mid \vec{U}, \vec{W}}(\vec{\theta} \mid \vec{u}, \vec{w})
  \d\vec{\theta} 
  \d F(\vec{u} \mid \vec{w}),
\end{equation}
for any strictly positive conditional (on $\vec{W}$) probability measure $F(\vec{u} \mid \vec{w})$ and any strictly positive probability mass function $\tilde{p}_{\vec{W}}(\cdot)$ on $\mathcal{W}$. Choosing $\d F(\vec{u} \mid \vec{w}) = p_{\vec{U} \mid \vec{W}}(\vec{u} \mid \vec{w})\d\vec{u}$ for the conditional measure in \eqref{eq:Bayesmin}, we have that
$$
  \hat{\vec{\theta}}^\star(\cdot,\cdot) = \argmin_{{\vec{\theta}}(\cdot,\cdot)}
  \sum_{\vec{w} \in \mathcal{W}} \tilde{p}_{\vec{W}}(\vec{w})
  \int_\mathcal{U}
  \int_\Theta 
  L(\vec{\theta}, \hat{\vec{\theta}}(\vec{u},\vec{w}))p_{\vec{\theta} \mid \vec{U}, \vec{W}}(\vec{\theta} \mid \vec{u}, \vec{w})\d\vec{\theta}  p_{\vec{U} \mid \vec{W}}(\vec{u} \mid \vec{w})  \d\vec{u}. 
$$
Applying Bayes rule to $p_{\vec{\theta} \mid \vec{U}, \vec{W}}(\vec{\theta} \mid \vec{u}, \vec{w})$ yields
\begin{equation}\label{eq:Bayesmin2}
  \hat{\vec{\theta}}^\star(\cdot,\cdot) = \argmin_{{\vec{\theta}}(\cdot,\cdot)}
  \sum_{\vec{w} \in \mathcal{W}} 
  \int_\mathcal{U}
  \int_\Theta 
  L(\vec{\theta}, \hat{\vec{\theta}}(\vec{u},\vec{w}))
  \frac{p_{\vec{U} \mid \vec{W}, \vec{\theta}}(\vec{u} \mid \vec{w}, \vec{\theta})p_{\vec{W} \mid \vec{\theta}}(\vec{w} \mid \vec{\theta})p_{\vec{\theta}}(\vec{\theta})}{p_{\vec{W}}(\vec{w})}
  \tilde{p}_{\vec{W}}(\vec{w}) \d\vec{\theta}  \d\vec{u}. 
\end{equation}
From \eqref{eq:Bayesmin2} we see that if $p_{\vec{W} \mid \vec{\theta}}(\vec{w} \mid \vec{\theta}) = p_{\vec{W}}(\vec{w})$, then 
\begin{equation*}
  \hat{\vec{\theta}}^\star(\cdot,\cdot) = \argmin_{{\vec{\theta}}(\cdot,\cdot)}
  \sum_{\vec{w} \in \mathcal{W}} 
  \int_\mathcal{U}
  \int_\Theta 
  L(\vec{\theta}, \hat{\vec{\theta}}(\vec{u},\vec{w}))
  p_{\vec{U} \mid \vec{W}, \vec{\theta}}(\vec{u} \mid \vec{w}, \vec{\theta})p_{\vec{\theta}}(\vec{\theta})
  \tilde{p}_{\vec{W}}(\vec{w}) \d\vec{\theta}  \d\vec{u},
\end{equation*}
for any positive $\tilde{p}_{\vec{W}}(\cdot)$ on $\mathcal{W}$, thus completing the proof. 
\end{proof}

\noindent Theorem~\ref{thm:invariance} establishes that, apart from the requirement of strict positivity, the choice of distribution for $\mathcal{I}_1$ (equivalently, $\vec{W}$) is immaterial for the Bayes estimator. However, in practice, NBEs are trained on a finite number of realizations from the statistical model, and the empirical Bayes risk is therefore affected by Monte Carlo error that does depend on this distribution, particularly in high-dimensional settings where the number of possible missingness patterns ($2^n$) is very large. This pitfall is demonstrated in Section~\reffmain{sec:simulationstudies}, where we show that choosing a distribution for $\mathcal{I}_1$ that assigns low probability to the observed missingness pattern can yield a 
 statistically inefficient
 estimator.  
 For further discussion, see Section~\reffmain{sec:neuralmasking}.

\section{Asymptotic behavior of approximate MCEM sequences}\label{sec:supp:convergenceofapproximateMCEM}

\citet{Chan_1995} showed formally that, for large Monte Carlo sample size $m$, a sequence generated by the MCEM algorithm behaves approximately as a first-order autoregressive (AR(1)) process centered on the conventional EM update. They also showed that, in the vicinity of an isolated local maximizer $\vec{\theta}^*$ of the incomplete-data posterior density, a sequence generated by the MCEM algorithm behaves approximately as a stationary AR(1) process with mean $\vec{\theta}^*$. Here, we extend these results to the more general setting where each MCEM update is only approximate, for example, when each update is performed using an NBE, like in our proposed Algorithm~\reffmain{alg:neuralEM}. 

We first recall the EM and MCEM algorithms that were reviewed in Section~\reffmain{sec:EMalgorithm:background}, and the family of approximate MCEM algorithms introduced in Section~\reffmain{sec:EMalgorithm:general}. There we use $\vec{Z}_1$ and $\vec{Z}_2$ to denote the subvectors of $\vec{Z}$ that are treated as observed and missing, respectively, and we use \mbox{$\vec{Z}$} to denote the complete data. In what follows, we take a Bayesian perspective, based on maximization of the posterior density; the frequentist version of the subsequent algorithms are recovered by taking the prior, which can be viewed as a penalty function, to be $\priordensity(\vec{\theta}) \propto 1$. The expected complete-data log-posterior, central to EM-based algorithms, is given by
$$
  Q(\vec{\theta}' \mid \vec{\theta}) = \log \priordensity(\vec{\theta}') + \E_{\vec{Z}_2 \mid \vec{Z}_1, \vec{\theta}}\log p(\vec{\theta}' \mid \vec{Z}_1, \vec{Z}_2)
$$
where the `prime' notation $\vec{\theta}'$ is used to denote an alternative parameter vector (and not the transpose of $\vec{\theta}$). 
 The EM algorithm \citep{Dempster_1977_EM_algorithm, Wu_1983_EM_algorithm, McLachlan_2008_EM_algorithm} iteratively finds $\vec{\theta}^{(l)}$ such that
\begin{equation}\label{eqnsupp:EM}
  \vec{\theta}^{(l)}  = M(\vec{\theta}^{(l-1)}) \equiv \argmax_{\vec{\theta}' \in \Theta} Q(\vec{\theta}' \mid \vec{\theta}^{(l-1)});  \quad l = 1, 2, \dots,
\end{equation}
where $M(\cdot)$ is deterministic and $\Theta$ denotes the parameter space. In the MCEM algorithm \citep{Wei_Tanner_1990_Monte_Carlo_EM}, one replaces $Q(\vec{\theta}' \mid \vec{\theta})$ with a Monte Carlo approximation,
\begin{equation}\label{eqn:MCEM_Q}
 Q_{m}(\vec{\theta}' \mid \vec{\theta})  = \log \priordensity(\vec{\theta}') + \frac{1}{m} \sum_{j=1}^m \log p(\vec{\theta}' \mid \vec{Z}_1, \vec{Z}_2^{(j)}),
\end{equation}
where $m$ is the Monte Carlo sample size and $\vec{Z}_2^{(j)} \sim p(\vec{Z}_2 \mid \vec{Z}_1, \vec{\theta})$, independently for $j = 1,\ldots,m$. 
Hence, the MCEM algorithm iteratively finds $\vec{\theta}_{m}^{(l)}$ such that
\begin{equation*}\label{eqn:MCEM_update}
  \vec{\theta}_{m}^{(l)} = M_m(\vec{\theta}_{m}^{(l-1)}) \equiv \argmax_{\vec{\theta}' \in \Theta} Q_{m}(\vec{\theta}' \mid \vec{\theta}_{m}^{(l-1)}); \quad l = 1, 2, \dots.
\end{equation*}
In this work, we consider an approximate version of the MCEM algorithm that we write as:
\begin{equation}\label{eqn_supp:approxMCEMupdate}
  \tilde{\vec{\theta}}_m^{(l)} = \tilde{M}_m(\tilde{\vec{\theta}}_m^{(l-1)}) \equiv M_m(\tilde{\vec{\theta}}_m^{(l-1)}) + \vec{\delta}_{m}(\tilde{\vec{\theta}}_m^{(l-1)}); \quad l = 1, 2, \dots,
\end{equation}
where $\vec{\delta}_{m}(\cdot)$ denotes (random) approximation error that may or may not have mean zero. In our setting, the additional approximation error $\vec{\delta}_{m}(\cdot)$ is due to our use of an NBE to approximate the random map $M_m(\cdot)$. 
  
We now extend the formal analysis of \citet{Chan_1995} to algorithms of the form~\eqref{eqn_supp:approxMCEMupdate},  
  under the assumptions that, for $\vec{\theta} \in \Theta$, $\E\{\vec{\delta}_{m}(\vec{\theta})\} \to \vec{\mu}(\vec{\theta})$ and \mbox{$\V\{\vec{\delta}_{m}(\vec{\theta})\} \to \vec{\Sigma}(\vec{\theta})$}, as $m \to \infty$. 
The starting point is the expansion of the gradient of~\eqref{eqn:MCEM_Q} 
around $M(\vec{\theta})$: 
\begin{equation}\label{eq:approx_Q}
\nabla_{\vec{\theta}'}Q_{m}(\vec{\theta}' \mid \vec{\theta}) \approx  \nabla_{\tilde{\vec{\theta}}} Q_{m}(\tilde{\vec{\theta}} \mid \vec{\theta})\big|_{\tilde{\vec{\theta}} = M(\vec{\theta})} + \nabla^2_{\tilde{\vec{\theta}}}Q_{m}(\tilde{\vec{\theta}} \mid \vec{\theta})\big|_{\tilde{\vec{\theta}} = M(\vec{\theta})} \left\{ \vec{\theta}' - M(\vec{\theta}) \right\},
\end{equation}
where we neglect higher-order terms. Evaluating~\eqref{eq:approx_Q} at $\vec{\theta}' = \tilde{M}_m(\vec{\theta})$, we obtain: 
\begin{equation}\label{eqn:approx_MCEM_M_expansion}
  \tilde{M}_m(\vec{\theta}) \approx  M(\vec{\theta}) - \left\{\nabla^2_{\tilde{\vec{\theta}}}Q_{m}(\tilde{\vec{\theta}} \mid \vec{\theta})\big|_{\tilde{\vec{\theta}} = M(\vec{\theta})}\right\}^{-1} \left\{\nabla_{\tilde{\vec{\theta}}} Q_{m}(\tilde{\vec{\theta}} \mid \vec{\theta})\big|_{\tilde{\vec{\theta}} = M(\vec{\theta})} - \vec{\xi}(\vec{\theta})\right\},
\end{equation}
where 
\begin{equation}\label{eqn:approx_MCEM_gradient}
  \vec{\xi}(\vec{\theta}) \equiv  \nabla_{\tilde{\vec{\theta}}}Q_{m}(\tilde{\vec{\theta}} \mid \vec{\theta})\big|_{\tilde{\vec{\theta}} = \tilde{M}_m(\vec{\theta})}.
\end{equation}
Evaluating~\eqref{eqn:approx_MCEM_M_expansion} at $\vec{\theta} = \tilde{\vec{\theta}}_{m}^{(l-1)}$ and substituting the resulting expression for $\tilde{M}_m(\tilde{\vec{\theta}}_{m}^{(l-1)})$ into~\eqref{eqn_supp:approxMCEMupdate}, we obtain the following first-order autoregressive model:
\begin{equation}\label{eq:AR1_MCEM}
\tilde{\vec{\theta}}_{m}^{(l)} \approx M(\tilde{\vec{\theta}}_{m}^{(l-1)}) + \epsilonb_{m}(\tilde{\vec{\theta}}_{m}^{(l-1)}), 
\end{equation}
where
\begin{equation}\label{eq:innovations}
\epsilonb_{m}(\vec{\theta}) \equiv -\left\{\nabla^2_{\tilde{\vec{\theta}}}Q_{m}(\tilde{\vec{\theta}} \mid \vec{\theta})\big|_{\tilde{\vec{\theta}} = M(\vec{\theta})}\right\}^{-1} \left\{\nabla_{\tilde{\vec{\theta}}} Q_{m}(\tilde{\vec{\theta}} \mid \vec{\theta})\big|_{\tilde{\vec{\theta}} = M(\vec{\theta})} - \vec{\xi}(\vec{\theta}) \right\}.
\end{equation}

To derive the asymptotic properties of the innovations defined in~\eqref{eq:innovations}, we first derive expressions for the asymptotic expectation and asymptotic variance of $\vec{\xi}(\vec{\theta})$, defined in~\eqref{eqn:approx_MCEM_gradient}. 
 To proceed, we consider a second expansion of the gradient of~\eqref{eqn:MCEM_Q}, this time around $M_m(\vec{\theta})$. Neglecting higher-order terms, 
\begin{align}
	\nabla_{\vec{\theta}'}Q_{m}(\vec{\theta}' \mid \vec{\theta}) &\approx  \nabla_{\tilde{\vec{\theta}}} Q_{m}(\tilde{\vec{\theta}} \mid \vec{\theta})\big|_{\tilde{\vec{\theta}} = M_m(\vec{\theta})} + \nabla^2_{\tilde{\vec{\theta}}}Q_{m}(\tilde{\vec{\theta}} \mid \vec{\theta})\big|_{\tilde{\vec{\theta}} = M_m(\vec{\theta})} \left\{ \vec{\theta}' - M_m(\vec{\theta}) \right\}\nonumber\\ 
	& =  \nabla^2_{\tilde{\vec{\theta}}}Q_{m}(\tilde{\vec{\theta}} \mid \vec{\theta})\big|_{\tilde{\vec{\theta}} = M_m(\vec{\theta})} \left\{ \vec{\theta}' - M_m(\vec{\theta}) \right\},\label{eq:approx_Q2}
\end{align}
since  $\nabla_{\tilde{\vec{\theta}}} Q_{m}(\tilde{\vec{\theta}} \mid \vec{\theta})\big|_{\tilde{\vec{\theta}} = M_m(\vec{\theta})} = \zerob$ by definition. Evaluating~\eqref{eq:approx_Q2} at $\vec{\theta}' = \tilde{M}_m(\vec{\theta})$ and substituting the resulting expression into~\eqref{eqn:approx_MCEM_gradient}, we have
\begin{align*}
\vec{\xi}(\vec{\theta}) 
&\approx \nabla^2_{\tilde{\vec{\theta}}}Q_{m}(\tilde{\vec{\theta}} \mid \vec{\theta})\big|_{\tilde{\vec{\theta}} = M_m(\vec{\theta})} \left\{\tilde{M}_m(\vec{\theta}) - M_m(\vec{\theta}) \right\}\\
& = \nabla^2_{\tilde{\vec{\theta}}}Q_{m}(\tilde{\vec{\theta}} \mid \vec{\theta})\big|_{\tilde{\vec{\theta}} = M_m(\vec{\theta})} \vec{\delta}_{m}(\vec{\theta}),
\end{align*}
 where $\vec{\delta}_{m}(\vec{\theta}) \equiv \tilde{M}_m(\vec{\theta}) - M_m(\vec{\theta})$ is the approximation error that appears in~\eqref{eqn_supp:approxMCEMupdate}. 
 Now, as $m \rightarrow \infty$, assuming the law of large numbers holds, and suitable regularity conditions on $\log p(\vec{\theta} \mid \vec{Z}_1, \vec{Z}_2)$, we have that $M_m(\vec{\theta}) \to M(\vec{\theta})$, and 
\begin{equation}\label{eq:Vmat}
-\nabla^2_{\tilde{\vec{\theta}}}Q_{m}(\tilde{\vec{\theta}} \mid \vec{\theta})\big|_{\tilde{\vec{\theta}} = M_m(\vec{\theta})} \longrightarrow -\E_{\vec{Z}_2 \mid \vec{Z}_1, \vec{\theta}}\left[\nabla^2_{\tilde{\vec{\theta}}} \log p(\tilde{\vec{\theta}} \mid \vec{Z}_1, \vec{Z}_2)\big|_{\tilde{\vec{\theta}} = M(\vec{\theta})}\right] \equiv \vec{V}(\vec{\theta}).
\end{equation}
Therefore, for large $m$, we have
\begin{align}
\E\{\vec{\xi}(\vec{\theta})\} &\approx \E\{\vec{V}(\vec{\theta})\vec{\delta}_{m}(\vec{\theta})\} \approx -\vec{V}(\vec{\vec{\theta}})\vec{\mu}(\vec{\theta}), \label{eq:Egammab}\\
\Var\{\vec{\xi}(\vec{\theta})\}  
&\approx  \V\{\vec{V}(\vec{\theta})\vec{\delta}_{m}(\vec{\theta})\} \approx \vec{V}(\vec{\theta})\vec{\Sigma}(\vec{\theta})\vec{V}(\vec{\theta}).\label{eq:Vargammab}
\end{align}

We now derive the asymptotic expectation and asymptotic variance of the innovations $\epsilonb_{m}(\vec{\theta})$ defined in~\eqref{eq:innovations}. 
 Suppose it holds that
 \begin{align*}
\E_{\{\vec{Z}_2^{(j)} \mid \vec{Z}_1, \vec{\theta}\}_{j=1}^m}\left[\nabla_{\tilde{\vec{\theta}}} Q_{m}(\tilde{\vec{\theta}} \mid \vec{\theta})\big|_{\tilde{\vec{\theta}} =  M(\vec{\theta})}\right] &= \nabla_{\tilde{\vec{\theta}}} \E_{\{\vec{Z}_2^{(j)} \mid \vec{Z}_1, \vec{\theta}\}_{j=1}^m}\left[Q_{m}(\tilde{\vec{\theta}} \mid \vec{\theta})\big|_{\tilde{\vec{\theta}} =  M(\vec{\theta})}\right].
\end{align*}
Then,
\begin{equation}\label{eq:zero_exp}
\E_{\{\vec{Z}_2^{(j)} \mid \vec{Z}_1, \vec{\theta}\}_{j=1}^m}\left[\nabla_{\tilde{\vec{\theta}}} Q_{m}(\tilde{\vec{\theta}} \mid \vec{\theta})\big|_{\tilde{\vec{\theta}} =  M(\vec{\theta})}\right] = \nabla_{\tilde{\vec{\theta}}} Q(\tilde{\vec{\theta}} \mid \vec{\theta})\big|_{\tilde{\vec{\theta}} =  M(\vec{\theta})} = \zerob, 
\end{equation}
since the gradient at the true maximizer is zero by definition. Therefore, for large $m$, $\epsilonb_{m}(\vec{\theta}) \approx -\vec{V}^{-1}(\vec{\theta})\vec{\xi}(\vec{\theta})$, so that from~\eqref{eq:Egammab} we have $\E\{\epsilonb_{m}(\vec{\theta})\} \approx \vec{\mu}(\vec{\theta})$. Further, for an independent Monte Carlo sample, which we assume throughout, we have that
\begin{align*}
  \Var_{\{\vec{Z}_2^{(j)} \mid \vec{Z}_1, \vec{\theta}\}_{j=1}^m}\left[\nabla_{\tilde{\vec{\theta}}} Q_{m}(\tilde{\vec{\theta}} \mid \vec{\theta})\big|_{\tilde{\vec{\theta}} =  M(\vec{\theta})}\right] &=  \Var_{\{\vec{Z}_2^{(j)} \mid \vec{Z}_1, \vec{\theta}\}_{j=1}^m}\left[\frac{1}{m}\sum_{j=1}^m \nabla_{\tilde{\vec{\theta}}}\log p(\tilde{\vec{\theta}} \mid \vec{Z}_1, \vec{Z}_2^{(j)})\big|_{\tilde{\vec{\theta}} = M(\vec{\theta})}\right] \\
    &= \frac{1}{m} \Gammamat(\vec{\theta}),
\end{align*}
where 
$$
\Gammamat(\vec{\theta}) \equiv \frac{1}{m}\sum_{j=1}^m \Var_{\{\vec{Z}_2^{(j)} \mid \vec{Z}_1, \vec{\theta}\}_{j=1}^m}\left[\nabla_{\tilde{\vec{\theta}}}\log p(\tilde{\vec{\theta}} \mid \vec{Z}_1, \vec{Z}_2^{(j)})\big|_{\tilde{\vec{\theta}} = M(\vec{\theta})}\right].
$$
Now, as $m \to \infty$, $\Gammamat(\vec{\theta}) \to \vec{V}(\vec{\theta})$, given by~\eqref{eq:Vmat}. Then, using~\eqref{eq:Vargammab}, for large $m$ we have:
\begin{align*}
\Var\{\epsilonb_{m}(\vec{\theta})\}
&\approx \vec{V}(\vec{\theta})^{-1}\Var\{\vec{\xi}(\vec{\theta})\}\vec{V}(\vec{\theta})^{-1} 
+ 
\frac{1}{m}\vec{V}(\vec{\theta})^{-1}\Gammamat(\vec{\theta})\vec{V}(\vec{\theta})^{-1}\\
&\approx \vec{\Sigma}(\vec{\theta}) 
+ 
\frac{1}{m}\vec{V}(\vec{\theta})^{-1}. 
\end{align*}

Having obtained the first-order autoregressive formulation~\eqref{eq:AR1_MCEM} and a formal derivation of its asymptotic properties, we now consider the behavior near an isolated local maximizer $\vec{\theta}^*$ of the incomplete-data posterior density $p(\vec{\theta} \mid \vec{Z}_1)$. Since $\vec{\theta}^*$ is an isolated local maximizer of $p(\vec{\theta} \mid \vec{Z}_1)$, it is also an asymptotically stable fixed point of $M(\cdot)$ \citep[][Lemma~1]{Chan_1995}. In particular, $M(\vec{\theta}^*) = \vec{\theta}^*$, and hence a first-order Taylor expansion of $M(\vec{\theta})$ around $\vec{\theta}^*$ yields:
\begin{align}
M(\vec{\theta}) 
&\approx M(\vec{\theta}^*) + \vec{A}(\vec{\theta}^*) \left( \vec{\theta} - \vec{\theta}^* \right)\nonumber\\
&=\vec{\theta}^*  + \vec{A}(\vec{\theta}^*) \left( \vec{\theta} - \vec{\theta}^* \right),\label{eq:approx_M}
\end{align}
where $\vec{A}(\vec{\theta}^*) \equiv \nabla_{\vec{\theta}}M(\vec{\theta})\big|_{\vec{\theta} = \vec{\theta}^*}$. Substituting the approximation~\eqref{eq:approx_M} into \eqref{eq:AR1_MCEM} yields
\begin{equation}\label{eq:AR1_MCEM_approx}
\tilde{\vec{\theta}}_{m}^{(l)}  \approx  \vec{\theta}^*  +  \vec{A}(\vec{\theta}^*)\Big\{ \tilde{\vec{\theta}}_{m}^{(l-1)} - \vec{\theta}^* \Big\}+ \epsilonb_{m}(\tilde{\vec{\theta}}_{m}^{(l-1)}),
\end{equation}
and therefore $\{\tilde{\vec{\theta}}_{m}^{(l)}\}$ is approximately an AR(1) process with autoregressive parameter $\vec{A}(\vec{\theta}^*)$. Further, substituting the approximation~\eqref{eq:approx_M} into~\eqref{eqnsupp:EM} yields the linearized system (i.e., linear approximation to~\eqref{eqnsupp:EM}),
 \begin{equation}\label{eq:AR1_EM}
\vec{\theta}^{(l)} \approx \vec{\theta}^* + \vec{A}(\vec{\theta}^*)\Big\{\vec{\theta}^{(l-1)} - \vec{\theta}^* \Big\}.
\end{equation}
Under the assumption that the local maximizer $\vec{\theta}^*$ is also an asymptotically stable fixed point of~\eqref{eq:AR1_EM}, the spectral radius of $\vec{A}(\vec{\theta}^*)$ is strictly less than one \citep[][Thm.~4.13]{Elaydi_2005}. 
  Therefore, the AR(1) process \eqref{eq:AR1_MCEM_approx} that approximates $\{\tilde{\vec{\theta}}_{m}^{(l)}\}$ is stationary, and the innovations can be treated as independent and identically distributed with mean and variance evaluated at $\vec{\theta}^*$. This gives the following stationary AR(1) representation: 
\begin{equation}\label{eq:AR1_MCEM_approx2}
  \tilde{\vec{\theta}}_{m}^{(l)}  \approx  \vec{\theta}^*  + \vec{A}(\vec{\theta}^*) \Big\{  \tilde{\vec{\theta}}_{m}^{(l-1)} - \vec{\theta}^* \Big\} + \epsilonb_{m}(\vec{\theta}^*),
\end{equation}
where $\E\{\vec{\epsilon}_{m}(\vec{\theta}^{*})\} \approx \vec{\mu}(\vec{\theta}^*)$ and $\V\{\vec{\epsilon}_{m}(\vec{\theta}^{*})\} \approx \vec{\Sigma}(\vec{\theta}^*) + \frac{1}{m}\vec{V}(\vec{\theta}^{*})^{-1}$. The process~\eqref{eq:AR1_MCEM_approx2} has mean 
\begin{equation}\label{eqn_supp:stationary_mean}
\E\{\tilde{\vec{\theta}}_{m}^{(l)}\} \approx \vec{\theta}^* + \{\vec{I} - \vec{A}(\vec{\theta}^*)\}^{-1}\vec{\mu}(\vec{\theta}^*), 
\end{equation}
where $\vec{I}$ denotes the identity matrix; and variance $\V\{\tilde{\vec{\theta}}_{m}^{(l)}\} \approx \vec{C}(\vec{\theta}^*)$, where $\vec{C}(\vec{\theta}^*)$ is the solution to 
 $\vec{C}(\vec{\theta}^*) = \vec{A}(\vec{\theta}^*) \vec{C}(\vec{\theta}^*) \vec{A}(\vec{\theta}^*)^\tp + \vec{\Sigma}(\vec{\theta}^*) + \frac{1}{m}\vec{V}(\vec{\theta}^{*})^{-1}$. 

Importantly, in the ideal special case of mean-zero approximation error, $\vec{\mu}(\vec{\theta}) = \vec{0}$, and the mean 
 \eqref{eqn_supp:stationary_mean} 
reduces to $\vec{\theta}^*$, allowing one to recover $\vec{\theta}^*$ through the averaging procedure described in Section~\reffmain{sec:EMalgorithm:background}. 

\section{Continuous approximations of the 0--1 loss function}\label{sec:approximations01loss} 

In Section~\reffmain{sec:neuralMAPestimation}, we describe how an NBE can be constructed to approximate the MAP estimator. The approach hinges on the use of a continuous approximation of the 0--1 loss function. In this section, we consider several candidate loss functions. 

In the main text, we adopt the loss function, 
\begin{equation}\label{eqn:supp:tanhloss}
 L_{\rm{TANH}}(\vec{\theta}, \hat{\vec{\theta}}; \kappa) 
 = \rm{tanh}(\|\hat{\vec{\theta}} - \vec{\theta}\|/\kappa),
 \quad \kappa > 0.
\end{equation}
where $\|\cdot\|$ denotes the Euclidean norm (although any norm in $\mathbb{R}^d$ could be used) and $d$ denotes the dimension of $\vec{\theta}$. For $d=1$, Figure~\ref{fig:supp:lossfunctions}, panel A, shows \eqref{eqn:supp:tanhloss} and its gradient for $\kappa \in \{1, 0.5, 0.1, 0.05\}$. 
 For fixed $\kappa$, the gradient of \eqref{eqn:supp:tanhloss} is bounded as $\|\hat{\vec{\theta}} - \vec{\theta}\| \to 0$, and it yields the 0--1 loss function in the limit as $\kappa \to 0$, which we now prove.  
\begin{proof} From the definition of the hyperbolic tangent function, it follows that
\begin{align*}
 L_{\rm{TANH}}(\vec{\theta}, \hat{\vec{\theta}}; \kappa) 
 &= \rm{tanh}(\|\hat{\vec{\theta}} - \vec{\theta}\|/\kappa) \\
  &= \frac{\exp(\|\hat{\vec{\theta}} - \vec{\theta}\|/\kappa) - \exp(-\|\hat{\vec{\theta}} - \vec{\theta}\|/\kappa)}{\exp(\|\hat{\vec{\theta}} - \vec{\theta}\|/\kappa) + \exp(-\|\hat{\vec{\theta}} - \vec{\theta}\|/\kappa)},\\ 
 &= \frac{1 - \exp(-2\|\hat{\vec{\theta}} - \vec{\theta}\|/\kappa)}{1 + \exp(-2\|\hat{\vec{\theta}} - \vec{\theta}\|/\kappa)}. 
\end{align*}
First, suppose that $\hat{\vec{\theta}} = \vec{\theta}$, so that $\|\hat{\vec{\theta}} - \vec{\theta}\| = 0$. In this case, we have that
$$
 L_{\rm{TANH}}(\vec{\theta}, \hat{\vec{\theta}}; \kappa)\rvert_{\hat{\vec{\theta}} = \vec{\theta}} 
 = \frac{1 - \exp(0)}{1 + \exp(0)}
 = 0. 
$$
Second, suppose that $\hat{\vec{\theta}} \neq \vec{\theta}$, so that $\|\hat{\vec{\theta}} - \vec{\theta}\| > 0$. In this case, we have that
\begin{align*}
\lim_{\kappa \to 0} 
 L_{\rm{TANH}}(\vec{\theta}, \hat{\vec{\theta}}; \kappa) 
 &= \lim_{\kappa \to 0} \frac{1 - \exp(-2\|\hat{\vec{\theta}} - \vec{\theta}\|/\kappa)}{1 + \exp(-2\|\hat{\vec{\theta}} - \vec{\theta}\|/\kappa)} = 1, 
\end{align*}
since $\lim_{x\to\infty} e^{-ax} = 0$ for all $a > 0$.  Therefore, we obtain 
\begin{equation*}
 \lim_{\kappa \to 0} L_{\rm{TANH}}(\vec{\theta}, \hat{\vec{\theta}}; \kappa) = 
\begin{cases}
0 & \text{if }\hat{\vec{\theta}} = \vec{\theta},\\
1 & {\text{otherwise}},
\end{cases} 
\end{equation*}
which is the 0--1 loss function. 
\end{proof}

 \begin{figure}[!t]
	\centering
	\includegraphics[width = \linewidth]{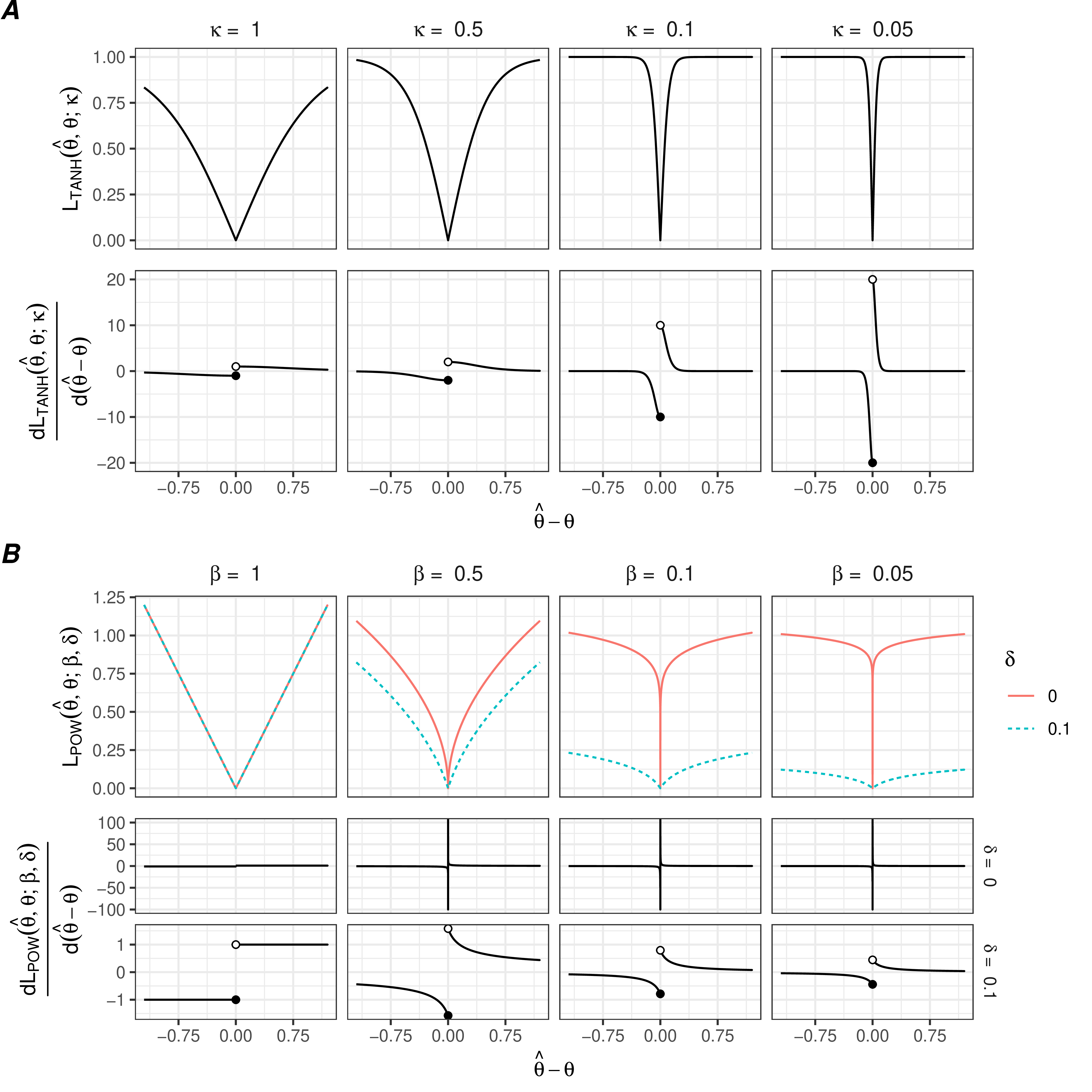}  
	\caption{(A) The loss function \eqref{eqn:supp:tanhloss} and its gradient for $\kappa \in \{1, 0.5, 0.1, 0.05\}$. (B) The loss function \eqref{eqn:supp:powerloss_modified} and its gradient for $\beta \in \{1, 0.5, 0.1, 0.05\}$ and $\delta \in \{0, 0.1\}$. For $\delta = 0$ and $\beta<1$, the gradient diverges at the origin; for $\delta = 0.1$, the gradient at the origin is finite and decreases as $\beta$ decreases.}\label{fig:supp:lossfunctions} 
\end{figure}

Other continuous approximations are available. For example, the loss function 
\begin{equation}\label{eqn:supp:powerloss}
 L(\vec{\theta}, \hat{\vec{\theta}}; \beta) 
 = \big(\|\hat{\vec{\theta}} - \vec{\theta}\|)^\beta, 
 \quad \beta > 0,
 \end{equation} 
generalizes the loss function given in \citet[][Eqn.~6]{Cressie_2022_loss_functions} to the multiparameter setting, and it also yields the 0--1 loss in the limit as $\beta \to 0$. A possible surrogate for the 0--1 loss function is therefore given by \eqref{eqn:supp:powerloss} with $\beta$ close to zero. However, with $\beta < 1$, the gradient of \eqref{eqn:supp:powerloss} diverges as $\|\hat{\vec{\theta}} - \vec{\theta}\| \to 0$, which can cause numerical instability during training. There are several ways to alleviate this issue. For instance, one may add a small positive constant $\delta$ to $\|\hat{\vec{\theta}} - \vec{\theta}\|$ in \eqref{eqn:supp:powerloss}, yielding the loss function, 
\begin{equation}\label{eqn:supp:powerloss_modified}
 L_{\rm{POW}}(\vec{\theta}, \hat{\vec{\theta}}; \beta, \delta) 
 = \big(\|\hat{\vec{\theta}} - \vec{\theta}\| + \delta)^\beta - \delta^\beta, 
 \quad \beta>0, \delta > 0,
 \end{equation} 
 Figure~\ref{fig:supp:lossfunctions}, panel B, shows \eqref{eqn:supp:powerloss_modified} and its gradient for $\beta \in \{1, 0.5, 0.1, 0.05\}$ and $\delta \in \{0, 0.1\}$. 

Another possible surrogate loss function is given by 
\begin{equation}\label{eqn:supp:corrloss}
 L_{\rm{CORR}}(\vec{\theta}, \hat{\vec{\theta}}; \rho, \alpha) = 1 - \{1 + (\|\hat{\vec{\theta}} - \vec{\theta}\|/\rho)^\alpha\}^{-1},  \quad \rho > 0,\, \alpha \geq 1,
\end{equation} 
 which yields the 0--1 loss in the limit as $\rho \to 0$, where the right-hand side of \eqref{eqn:supp:corrloss} is a variogram constructed from the Cauchy correlation function \citep{Gneiting_2004}. Similar to \eqref{eqn:supp:powerloss_modified} and \eqref{eqn:supp:tanhloss}, the gradient of \eqref{eqn:supp:corrloss} is bounded. In general, similarly constructed variogram models provide a broad family of functions that may serve as continuous approximations of the 0--1 loss function, with various degrees of differentiability at the origin that may be controlled by the practitioner. These classes are mentioned here for completeness; in this paper we use the loss function \eqref{eqn:supp:tanhloss}, which we find to work well in practice.

Finally, we make practical remarks relating to the training of NBEs under continuous approximations of the 0--1 loss function. 
First, for most conventional norms, $\|\hat{\vec{\theta}} - \vec{\theta}\|$ tends to increase as the dimension $d$ of $\vec{\theta}$ increases. This represents a potential pitfall since, for most approximations of the 0--1 loss function, like \eqref{eqn:supp:tanhloss}, the gradient vanishes as $\|\hat{\vec{\theta}} - \vec{\theta}\|$ becomes large. 
 A simple remedy for this is to replace the Euclidean norm with a scaled norm, for instance, $\|\cdot\|_p/d^{1/p}$, where $\|\cdot\|_p$ denotes the $L_p$ norm, to reduce the growth of distances with increasing dimension. 
 Second, to prevent the NBE from becoming ``stuck'' in regions of the parameter space where gradients are extremely small at the start of training, it can be beneficial to pretrain the network using conventional, well behaved loss functions, such as the mean-absolute-error or the mean-squared-error loss functions. Importantly, once the NBE has been pretrained, it can be trained under an approximation of the 0--1 loss function at almost no additional computational cost. This is achieved by fixing the initial portion of the network (often referred to as the ``summary network'', when the architecture is viewed in a modular fashion), which typically contains relatively expensive operations, such as convolution. Then updating occurs only in the latter portion of the network (often referred to as the ``inference network''), which typically consists of a simple, fast-to-evaluate multilayer perceptron. In our experiments, the pretraining stage using the mean-absolute-error loss function, ranges from several minutes to a few hours, whereas training under the approximate 0--1 loss function requires only a matter of seconds.

\section{Additional simulation study}\label{sec:additionalsimulationstudies}

In Section~\reffmain{sec:simulationstudies}, we conducted simulation studies involving spatial versions of the Gaussian process model and the hidden Potts model. Here, we conduct an additional simulation experiment, using a spatial version of a popular normal mean-variance mixture known as the generalized hyperbolic (GH) distribution. The likelihood is available for this case, permitting comparison 
 with a gold-standard likelihood-based estimator, namely the MAP estimator obtained by directly maximizing the posterior density. Here, we also consider a conventional simulation-based approach, approximate Bayesian computation \citep[ABC;][]{Tavare_1997_ABC, Pritchard_1999_ABC, Lintusaari_2017_ABC_review, Sisson_2018_ABC_handbook}. Before proceeding with the analysis, we first review the GH distribution, including its density function and the formulas required for conditional simulation.

The $n$-dimensional random vector $\vec{Z}$ is called a NMVM if it can be represented as, 
\begin{equation}\label{eqn:NMVM}
	\vec{Z} = \vec{\mu} + M\vec{\alpha} + \sqrt{M} \vec{V}, 
\end{equation}
where $\vec{\mu} \in \mathbb{R}^n$, $\vec{\alpha} \in \mathbb{R}^n$, $M$ is a positive latent random variable that is independent of the small-scale variation $\vec{V} \sim \Gau(\vec{0}, \vec{\Sigma})$, and where $\vec{\Sigma}$ is a covariance matrix. This family of distributions is closed under conditioning: for any partitioning of $\vec{Z}$ into components $\vec{Z}_1$ and $\vec{Z}_2$, and with $\vec{\mu}$, $\vec{\alpha}$, and $\vec{\Sigma}$ partitioned accordingly, $\vec{Z}_2 \mid \vec{Z}_1$ is also a NMVM, with parameters 
\begin{align*}
\vec{\mu}_{2\mid 1} &= \vec{\mu}_2 + \vec{\Sigma}_{21}\vec{\Sigma}_{11}^{-1}(\vec{Z}_1 - \vec{\mu}_1),\\
\vec{\alpha}_{2\mid 1} &= \vec{\alpha}_2 - \vec{\Sigma}_{21}\vec{\Sigma}_{11}^{-1}\vec{\alpha}_1,\\
\vec{\Sigma}_{22\mid 1} &= \vec{\Sigma}_{22} - \vec{\Sigma}_{21} \vec{\Sigma}_{11}^{-1} \vec{\Sigma}_{12},
\end{align*}
and with latent variable $M_{2|1}$ that is distributed according to $M \mid \vec{Z}_1$  \citep[][Theorem~1]{Jamalizadeh_Balakrishnan_2019_conditional_distribution_normal_mean_variance_mixture}. 

The generalized hyperbolic (GH) distribution is obtained when $M$ in \eqref{eqn:NMVM} follows a generalized inverse Gaussian (GIG) distribution. The GH distribution is prominent in financial modeling due to its flexible marginal distributions and infinite divisibility \citep{Barndorff-Nielsen_1997, Prause_1999}. 
We consider the parameterization of the GIG density employed by \citet{Browne_2015_GH_mixtures}, 
\begin{equation}\label{eqn:GIG:param2}
f_{\rm{GIG}}(m; \omega, \GHscaleparam, \lambda) 
= 
\frac{(m/\GHscaleparam)^{\lambda - 1}}{2\GHscaleparam K_{\lambda}(\omega)} \exp{\bigg\{-\frac{\omega}{2}\bigg(\frac{\GHscaleparam}{m} + \frac{m}{\GHscaleparam}\bigg)\bigg\}},  \quad m > 0, 
\end{equation}
with concentration parameter $\omega > 0$, shape parameter $\lambda \in \mathbb{R}$, scale parameter $\GHscaleparam > 0$, and where $K_{\lambda}(\cdot)$ denotes the modified Bessel function of the second kind of order $\lambda$. Then, $M_{2|1}$ also follows a GIG distribution with parameters  
 \begin{align*}
 \omega_{2\mid 1} &= \sqrt{\{\omega\GHscaleparam + (\vec{Z}_1 - \vec{\mu}_1)'\vec{\Sigma}_{11}^{-1}(\vec{Z}_1 - \vec{\mu}_1)\}\{\omega\GHscaleparam^{-1} + \vec{\alpha}_1'\vec{\Sigma}_{11}^{-1} \vec{\alpha}_1\}},\\
 \GHscaleparam_{2\mid 1} &= \sqrt{\{\omega\GHscaleparam + (\vec{Z}_1 - \vec{\mu}_1)'\vec{\Sigma}_{11}^{-1}(\vec{Z}_1 - \vec{\mu}_1)\}/\{\omega\GHscaleparam^{-1} + \vec{\alpha}_1'\vec{\Sigma}_{11}^{-1} \vec{\alpha}_1\}},\\
 \lambda_{2\mid 1} &= \lambda - n_1/2,
\end{align*}
where $n_1$ is the dimension of $\vec{Z}_1$ \citep[][Cor.~2]{Jamalizadeh_Balakrishnan_2019_conditional_distribution_normal_mean_variance_mixture}. Throughout, we fix $\GHscaleparam = 1$ for identifiability reasons. We then write $\vec{Z} \sim {\textrm{GH}}(\vec{\mu}, \vec{\alpha}, \vec{\Sigma},  \omega,  \GHscaleparam, \lambda)$, which has density function, 
$$
f_{\rm{GH}}(\vec{z};\vec{\mu}, \vec{\alpha}, \vec{\Sigma}, \omega, \GHscaleparam, \lambda) 
\propto
\frac{
K_{\lambda - d/2}\!\left(\GHsqrtterm\right)e^{(\vec{z} - \vec{\mu})^\tp \vec{\Sigma}^{-1}\vec{\alpha}}
}{
\left[\GHsqrtterm\right]^{n/2-\lambda}
}, 
$$
with normalizing constant,
$
(\omega\GHscaleparam^{-1} + \vec{\alpha}^\tp \vec{\Sigma}^{-1} \vec{\alpha})^{n/2-\lambda}(2\pi)^{-n/2}|\vec{\Sigma}|^{-1/2}\GHscaleparam^{-\lambda}/K_\lambda(\omega) 
$
 (\citeauthor{McNeil_2015}, \citeyear{McNeil_2015}, Eqn.~6.29; \citeauthor{Zhang_Huser_2022_normal_mean_variance_mixtures}, \citeyear{Zhang_Huser_2022_normal_mean_variance_mixtures}). In addition to being closed under conditioning, the GH family of distributions is also closed under marginalization (\citeauthor{Wei_2019}, \citeyear{Wei_2019}, Prop.~2). That is, if $\vec{Z} \equiv (\vec{Z}_1', \vec{Z}_2')' \sim {\textrm{GH}}(\vec{\mu}, \vec{\alpha}, \vec{\Sigma}, \omega, \GHscaleparam, \lambda)$, then $\vec{Z}_1 \sim {\textrm{GH}}(\vec{\mu}_1, \vec{\alpha}_1, \vec{\Sigma}_{11}, \omega, \GHscaleparam, \lambda)$.  However, although the likelihood function is available, likelihood-based estimators for the GH distribution require repeated evaluation of the Bessel function, which can be computationally burdensome. Likelihood-free methods may therefore still be useful for improving computational efficiency. 

Following \citet{Zhang_Huser_2022_normal_mean_variance_mixtures}, we consider a spatial version of the GH distribution. In this example, we take the spatial domain to be $\mathcal{D} \equiv [0, 1] \times [0, 1]$, we simulate complete data on a regular square grid of size $n = 16^2 = 256$, and we consider inference from $150$ incompletely observed spatial fields. In~\eqref{eqn:NMVM}, we set $\vec{\alpha} = \alpha \vec{1}$ for $\alpha \in \mathbb{R}$; we take $\vec{\Sigma}$ to be a correlation matrix determined by \eqreffmain{eqn:Matern_covariance_function} with range parameter $\rho > 0$, shape parameter $\nu > 0$, and $\tau = 0$; and we assume that the mean is modeled separately so that $\vec{\mu}$ can be set to zero. The parameters to be estimated are thus $\vec{\theta} \equiv (\alpha, \omega, \lambda, \rho, \nu)^\tp$. We assume that the parameters are independent \textit{a priori} with marginal priors \mbox{$\alpha\sim \Unif{-0.3}{0.3}$}, \mbox{$\omega\sim \Unif{0.1}{1}$}, \mbox{$\lambda\sim \Unif{-1}{1}$},  \mbox{$\rho\sim \Unif{0.05}{0.35}$}, and \mbox{$\nu \sim \Unif{0.5}{2}$}. 

We consider inference using NBEs based on either the masking approach (Algorithm~\reffmain{alg:one-hot}) or the EM approach (Algorithm~\reffmain{alg:neuralEM}) for handling missing data, under the same general settings given in Section~\reffmain{sec:generalsetting}. We compare the competing NBEs to the MAP estimator that numerically maximizes the posterior density, and a MAP estimator based on ABC. 

To facilitate a comparison with the NBEs, for ABC inference we use the same $25000$ simulated parameter--data pairs that were used to train the NBEs \citep{Mestdagh_2019_prepaid_ABC}. As the initial set of user-defined summary statistics, we use the empirical variogram, madogram, and rodogram \citep[corresponding to the variogram of order 2, 1, and 0.5, respectively;][]{Matheron_1987}, obtained by binning pairwise distances into 10 intervals with centers equally spaced between 0.05 and 0.5. We also include two measures of extremal dependence, the pairwise tail correlation coefficient \citep{Joe_1997} and the residual-tail-dependence coefficient \citep{Ledford_Tawn_1996}, both computed for all site pairs, evaluated at the high-probability threshold $u=0.95$, and averaged using the same distance bins as the variograms. Low-dimensional summary statistics offer numerous advantages for ABC \citep{Blum_2013_dimension_reduction_in_ABC}. To reduce the dimensionality of our user-defined summary statistics, we employ the semi-automatic procedure of \citet{Fearnhead_Prangle_2012}, implemented via the \pkg{abctools} package \citep{Nunes_Prangle_2016}, which yields a $d$-dimensional vector of summary statistics, where recall that $d$ denotes the dimension of $\vec{\theta}$ (see Figure~\ref{fig:ABC_summaries:GH}). ABC sampling is then performed by comparing the Euclidean distance between the summary statistics of the observed data and simulated data, using the package \pkg{abc} \citep{abc_package}. We retain only the parameter vectors whose corresponding simulated summary statistics fall within the smallest 1\% of distances from the observed summary statistics, 
 yielding $25000 \times 0.01 = 250$ accepted parameter vectors. 
 To refine the ABC posterior, we apply regression-adjustment methods \citep{Blum_2018_regression_adjustment_approaches} using the default neural-network regression model \citep{Blum_Francois_2010_ABC} employed by \pkg{abc}. We then define the ``ABC MAP'' by computing the mode of each parameter via kernel density estimation. Although the mode of the joint posterior generally does not coincide with the modes of the marginal posteriors, we adopt this approximation for simplicity of implementation.

\begin{figure}[t!]
	\centering 
	\includegraphics[width = \linewidth]{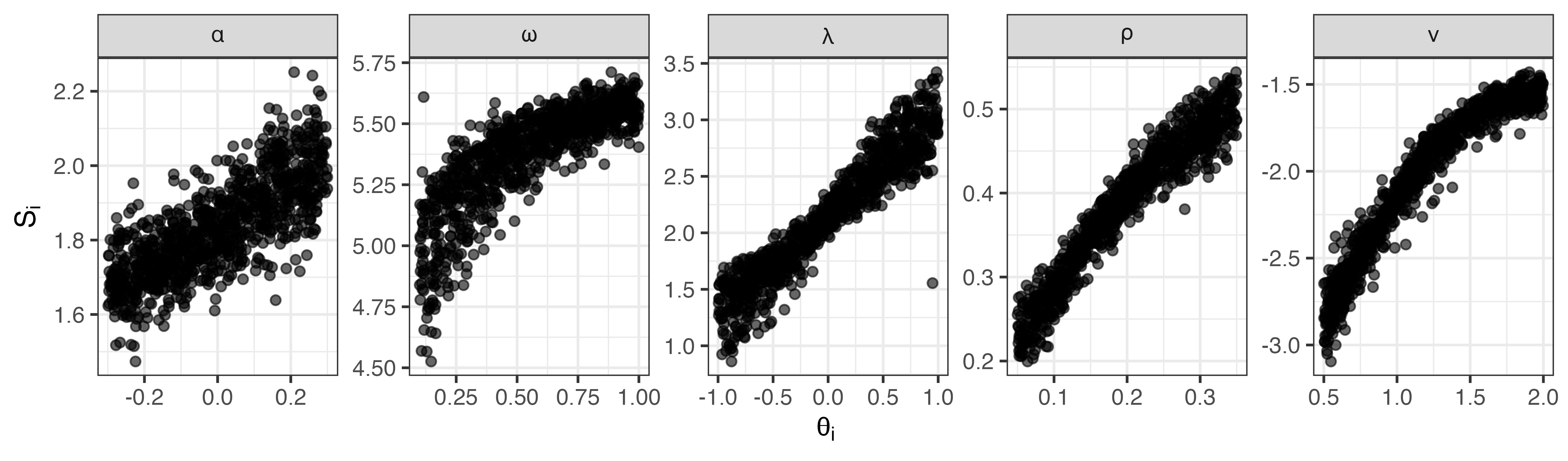}  
	\caption{Empirical plot of the semi-automatic summary statistics, $S_i$, $i = 1, \dots, 5$, obtained by the method of \citet{Fearnhead_Prangle_2012} against the corresponding true parameter values $\theta_i$,  $i = 1, \dots, 5$, shown separately for each of the five components of $\vec{\theta} \equiv (\alpha, \omega, \lambda, \rho, \nu)'$ in the GH distribution (Section~\ref{sec:additionalsimulationstudies}). Headings above each panel give the model parameter.}\label{fig:ABC_summaries:GH}
\end{figure}

 We compare the estimators using the time taken for their ``set-up'' and for estimation for a single data set post set-up (computational efficiency); and the empirical root-mean-squared errors (RMSEs) based on simulated data using 1000 parameter vectors sampled from the prior (statistical efficiency). 
 (Figure~\reffsupp{fig:risk_profile} shows the value of the objective function in Equation~\eqreff{eqn:optimization_task_encoding} or~\eqreff{eqn:optimization_task:adjustedprior_bayes} of the main text evaluated at the end of each epoch during training.)
 We test the estimators under two missingness mechanisms, MCAR and MICB, as defined in Section~\reffmain{sec:simulationstudies}, and we use RMSE$_{\text{MCAR}}$ and RMSE$_{\text{MICB}}$ to denote the RMSE of an estimator based on incomplete data simulated under the MCAR and MICB mechanisms, respectively. Table~\ref{tab:GH} summarizes the results, while Figure~\ref{fig:GH} shows simulated data and corresponding box plots of the empirical distributions of the estimates for one parameter setting. 

\begin{table}[t!]
\centering
\caption{
The ``set-up'' time, estimation time for a single test data set, and empirical RMSE under two missingness models for four estimators  of the parameters of the GH distribution 
(Section~\ref{sec:additionalsimulationstudies}). ``Set-up'' time refers to the total time required to sample $25000$ parameter vectors from the prior, simulate data from the model conditional on these parameters, and either train the neural networks (EM NBE and Masking NBE) or compute summary statistics (ABC MAP).
}\label{tab:GH}
\begin{tabular}{lccccc}
\hline
Estimator  & Set-up time (hrs) & Estimation time (s)    & RMSE$_{\text{MCAR}}$ & RMSE$_{\text{MICB}}$  \\ 
\hline
MAP              &  --             & 89.7               &  \textbf{0.062}      & \textbf{0.064} \\ 
EM NBE           &  1.67           & 0.812              &  \textbf{0.062}      & \textbf{0.064} \\  
Masking NBE      &  2.01           & \textbf{0.005}     &  0.067               &  0.191 \\
ABC MAP          &  \textbf{1.07}  & 2.03              &  0.097              & 0.107 \\ 
\hline
\end{tabular}
\end{table}

\begin{figure}[t!]
	\centering 
	\includegraphics[width = \linewidth]{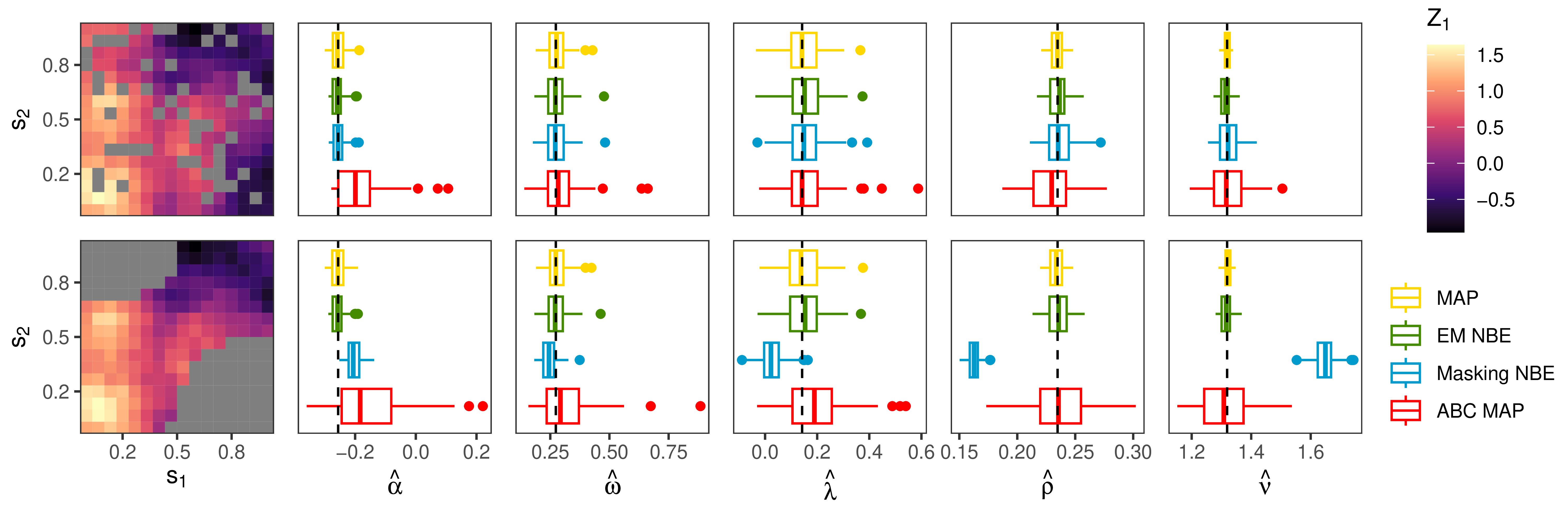}  
	\caption{Spatial data (first column) where the missingness is of type MCAR (first row) or MICB (second row) with missingness shown in gray, and corresponding empirical  distributions (second and third columns) for four estimators (MAP, EM NBE, Masking NBE, ABC MAP) of the parameters of the GH distribution  (Section~\ref{sec:additionalsimulationstudies}). True parameter values are shown as a dashed vertical line.
}\label{fig:GH}
\end{figure}

These results corroborate the findings in the main text. First, NBE approaches can be almost as statistically efficient as a gold-standard likelihood-based estimator, while being substantially more computationally efficient post-training. Second, the masking and EM approaches to neural Bayes estimation represent a trade-off between computational efficiency and robustness to the missingness mechanism. Third, the almost-as-good performance of the EM NBE relative to the MAP estimator indicates that Algorithm~\reff{alg:neuralEM} in the Estimation stage converges to a suitable point estimate across most, if not all, data sets; convergence for a single data set is illustrated in Figure~\ref{fig:convergence:GH}. 

Finally, like the EM NBE, the ABC MAP is agnostic to the missingness pattern, but its reliance on user-defined summary statistics reduces statistical efficiency. This reflects a broader limitation of ABC which is addressed by NBEs: Neural inference methods generally are well suited to settings in which sufficient statistics are unavailable, which is the case for the majority of models used in practice, including the models considered in 
 Sections~\reff{sec:simulationstudies} and~\reffmain{sec:application}. 
 
\section{Ensemble of NBEs}\label{sec:ensemble}

An ensemble of neural networks \citep{Hansen_Salamon_1990, Breiman_1996} consists of neural networks trained to solve the same task but with variations introduced during the training process. These variations could include different initial values for the neural-network parameters, different architectures, or different learning rates. The inference for an ensemble can be obtained by averaging the outputs of its individual networks, and it is often more accurate than that of any individual network.

 In our case, we obtain an ensemble of $J$ NBEs $\{\hat{\vec{\theta}}_{\vec{\gamma}_j^*}(\cdot) : j = 1, \dots, J\}$ by using $J$ different initial values for the neural-network parameters. We may posit the following working model to characterize the error of each NBE:
 \begin{equation}\label{eqn:single_NBE_estimate}
\hat{\vec{\theta}}_{\vec{\gamma}^*_j}(\vec{Z}) = \hat{\vec{\theta}}_{\text{Bayes}}(\vec{Z}) + \epsilon_j(\vec{Z}),\quad \vec{Z} \in \mathcal{Z}, \quad j = 1,\dots,J,
\end{equation} 
 where $\hat{\vec{\theta}}_{\text{Bayes}}(\vec{Z})$ is the Bayes estimate of $\vec{\theta}$ for the loss function being used, and $\epsilon_j(\vec{Z})$ is a mean-zero error term with variance $\sigma^2_\epsilon$. Then, the ensemble estimate is obtained as:  
\begin{equation}\label{eqn:ensemble_estimate}
\hat{\vec{\theta}}_{\text{ensemble}}(\vec{Z}) \equiv \frac{1}{J}\sum_{j=1}^J\hat{\vec{\theta}}_{\vec{\gamma}_j^*}(\vec{Z}), \quad \vec{Z} \in \mathcal{Z}.
\end{equation} 
Under our working model, for fixed $\vec{Z} \in \mathcal{Z}$,
 $$\E\{\hat{\vec{\theta}}_{\text{ensemble}}(\vec{Z})\} = \hat{\vec{\theta}}_{\text{Bayes}}(\vec{Z}), \quad \var\{\hat{\vec{\theta}}_{\text{ensemble}}(\vec{Z})\} = \sigma^2_\epsilon/J,$$ 
 where the expectation and variance are taken over the ensemble estimates for a fixed $\vec{Z}$. In practice, the errors may not have mean zero, nor be independent; however, provided that they are not perfectly correlated, the ensemble estimate \eqref{eqn:ensemble_estimate} will still have reduced variance compared with estimates from individual ensemble members. 

To illustrate the efficacy of an ensemble of NBEs, we consider a classical spatial Gaussian process model, where $\vec{Z} \equiv (Z_{1}, \dots, Z_{n})'$ are data obtained at locations $\{\vec{s}_{1}, \dots, \vec{s}_{n}\}$ in a spatial domain $\mathcal{D}$. Assume the data are Gaussian random variables with mean zero, variance one, and exponential spatial covariance function, 
$$
\textrm{cov}(Z_i, Z_j) = \textrm{exp}(-\|\vec{s}_i - \vec{s}_j\|/\theta), \quad i, j = 1, \dots, n,
$$
for unknown range parameter $\theta > 0$. Here, we take $\mathcal{D}$ to be the unit square, we simulate data on a grid with $n = 16^2 = 256$ observation locations, and we adopt the prior $\theta \sim \text{Unif}(0, 0.5)$. Note that the data are completely observed in this experiment, where our aim is to show the utility of an ensemble of NBEs. 
 
Since our data are gridded and complete, we construct the NBEs using convolutional neural networks \citep[CNNs; see, e.g.,][Ch.~9]{Dumoulin_2016, Goodfellow_2016_Deep_learning}. To demonstrate that the ensemble approach improves estimation across different neural-network architectures, we consider three architectures, each with a different number of trainable parameters~$\vec{\gamma}$. The first architecture (Architecture 1) was used by \citet{Zammit_2024_ARSIA}; it contains two convolutional layers and 150913 trainable parameters. The second architecture (Architecture 2) was proposed by \citet{Gerber_Nychka_2021_NN_param_estimation} and subsequently used by \citet{Sainsbury-Dale_2022_neural_Bayes_estimators} and \citet{Richards_2023_censoring}; it contains three convolutional layers and 638657 trainable parameters. The third architecture (Architecture 3), summarized in Table~\ref{tab:architecture:spatial}, is inspired by the well known ResNet architecture \citep{He_2016}. It contains a total of nine convolutional layers couched within so-called residual blocks \citep{He_2016}, and 390321 trainable parameters. The residual blocks mitigate the issue of vanishing gradients during training, thereby allowing for the construction of deeper networks that often outperform their shallower counterparts. We consider ensembles containing up to $10$ NBEs of the range parameter $\theta$, where each NBE is initialised with different, randomly generated values for the neural-network parameters $\vec{\gamma}$. 

\begin{table}[!t]
	\centering
	\caption{
	Summary of the CNN architecture used in Sections~\reff{sec:simulationstudies}~and~\reffmain{sec:application}, and in 
    Sections~\ref{sec:additionalsimulationstudies}~and~\ref{sec:ensemble}, 
    with $d$ the number of parameters in the given statistical model. The architecture can be used with grids of arbitrary size and shape; however, for simplicity, here we show the input and output dimensions of each layer given a square input grid of dimension $64\times64$. The table is divided into layers used for the summary network $\vec{\psi}(\cdot)$ and inference network $\vec{\phi}(\cdot)$ of the DeepSets representation given in Equation~(\reff{eqn:DeepSets}) of the main text. Each residual block consists of two sequential convolutional layers and batch normalization layers, along with a skip (shortcut) connection that directly connects the input of the block to its output \citep{He_2016}. The batch normalization layers compute the mean and variance for each input slice, normalize the input accordingly, and then apply a learnable affine transformation \citep{Ioffe_Szegedy_2015}. A padding of size 1 is used in each convolutional layer, and a stride of 2 is used in layers that reduce the input resolution (a stride of 1 is used otherwise) \citep[][Ch.~9]{Goodfellow_2016_Deep_learning}. For all but the final layer, we use rectified linear unit (ReLU) activation functions, $\text{ReLU}(x) \equiv \text{max}(0, x)$, while the final layer employs a softplus activation function, $\text{softplus}(x) \equiv \log(1 + e^x)$,  to ensure positive parameter estimates. When employing the masking approach of \citet{Wang_2022_neural_missing_data}, an extra input channel is needed to encode the missingness pattern, which doubles the number of parameters in the first layer. 
	}\label{tab:architecture:spatial}
	\begin{tabular}{ll*{4}{r}}
		\hline
		\hline
		Network & Layer type & Input dim. & Output dim.  & Kernel size & Parameters\\
		\hline
		$\vec{\psi}(\cdot)$ & Convolutional & [64, 64, 1] & [64, 64, 16] & $3\times 3$ & $\text{144}$ \\
		$\vec{\psi}(\cdot)$ & Batch normalization & [64, 64, 16] & [64, 64, 16] & - & 32 \\
		$\vec{\psi}(\cdot)$ & Residual block & [64, 64, 16] & [64, 64, 16] & $3\times 3$ & 4672 \\
		$\vec{\psi}(\cdot)$ & Residual block  & [64, 64, 16] & [32, 32, 32] & $3\times 3$ & 14528 \\
		$\vec{\psi}(\cdot)$ & Residual block & [32, 32, 32] & [16, 16, 64] & $3\times 3$ & 57728 \\
		$\vec{\psi}(\cdot)$ & Residual block & [16, 16, 64] & [8, 8, 128] & $3\times 3$ & 230144 \\
		$\vec{\psi}(\cdot)$ & Global mean pooling & [8, 8, 128] & [1, 1, 128] & - & 0 \\
		$\vec{\psi}(\cdot)$ & Flatten & [1, 1, 128] & [128] & - & 0 \\
		$\vec{\phi}(\cdot)$ & Dense & [128] & [128] & - & 16512 \\
		$\vec{\phi}(\cdot)$ & Dense & [128] & [512] & - & 66048 \\
		$\vec{\phi}(\cdot)$ & Dense & [512] & [$d$] & - & $513d$ \\
		\hline 
		\multicolumn{5}{l}{Total trainable parameters:} & $389808 + 513d$ \\
		\hline
	\end{tabular}
\end{table}

Figure~\ref{fig:ensemble} shows box plots of empirical RMSEs for 10 individual NBEs and the corresponding ensemble of $J=10$ NBEs, plotted for each architecture (left panel), and empirical RMSEs plotted as a function of the number of NBEs in the ensemble (right panel). The empirical RMSEs are based on a test set of 1000 parameter-data pairs sampled from the joint distribution of the parameters and the data, that is, 
\begin{equation}\label{eqn:RMSE_supp}
\text{RMSE}(\hat{\theta}(\cdot)) \equiv \sqrt{\frac{1}{1000} \sum_{k=1}^{1000} \{\hat{\theta}(\vec{Z}^{(k)}) - \theta^{(k)}\}^2},
\end{equation}
where $\theta^{(k)} \sim \text{Unif}(0, 0.5)$, $\vec{Z}^{(k)} \sim p(\vec{Z} \mid \theta^{(k)})$, and the estimator $\hat{\theta}(\cdot)$ corresponds either to a single NBE or an ensemble of NBEs as defined in \eqref{eqn:ensemble_estimate}. From Figure~\ref{fig:ensemble}, we see that the use of an ensemble of NBEs substantially reduces the RMSE when compared to a single NBE, irrespective of the architecture. 

\begin{figure}[!t]
	\centering
	\includegraphics[width = \linewidth]{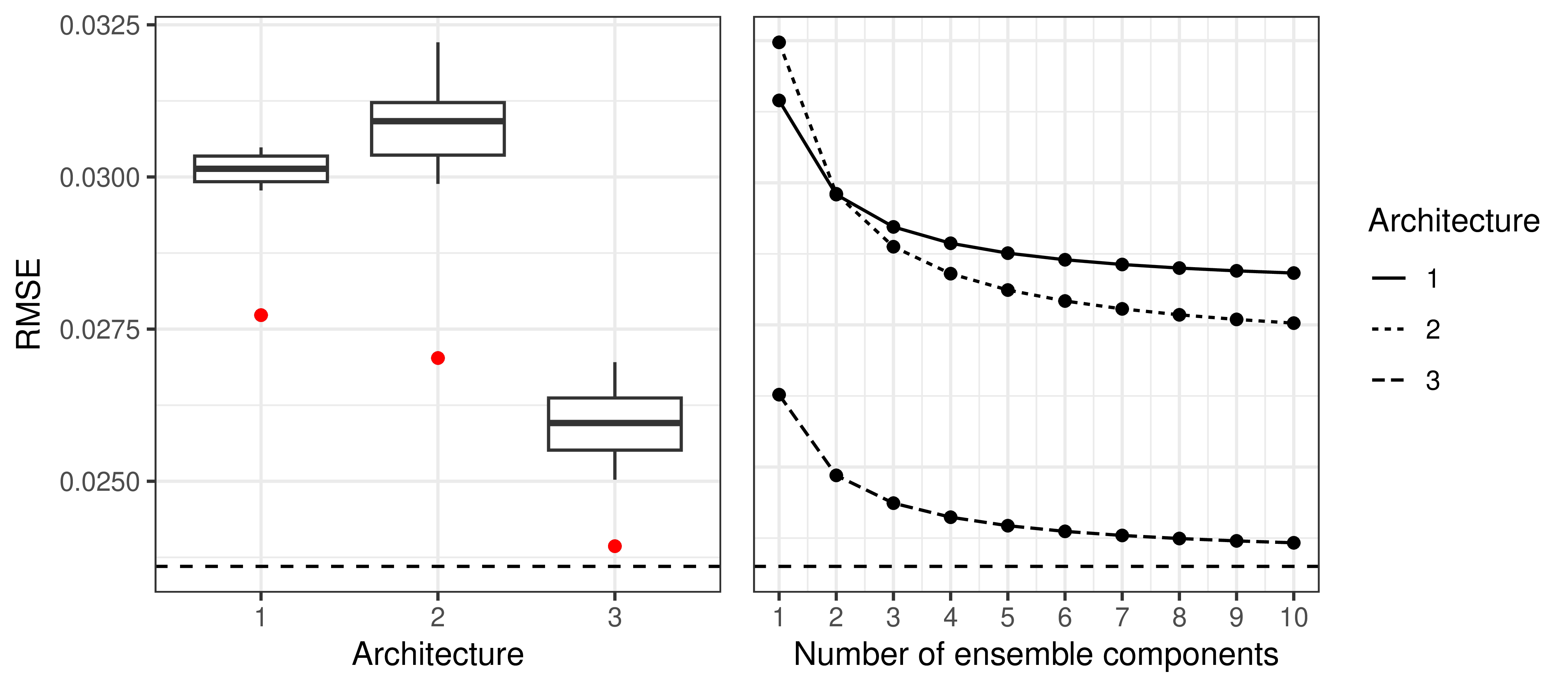}  
	\caption{
	(Left panel) Empirical RMSEs for the three architectures detailed in Section~\ref{sec:ensemble}, with box plots summarizing the RMSEs for 10 individual NBEs, and red points showing the RMSEs of the corresponding ensemble of $J=10$ NBEs. 
  (Right panel) Empirical RMSEs as a function of the number of NBEs in the ensemble, for the three different architectures used in the experiment. For each ensemble size, RMSEs are averaged over all possible combinations of the NBEs, so that performance is not tied to a specific ordering of the components. In both panels, the horizontal dashed line corresponds to the RMSE of the MAP estimator that numerically maximizes the unnormalized analytic posterior density. 
	}\label{fig:ensemble} 
\end{figure}

\clearpage
\section{MCMC sampling in hidden Potts models}\label{sec:MCMCHiddenPotts}

Here we describe two MCMC algorithms for sampling from the joint distribution of the hidden Potts model, which we consider in Sections~\reff{sec:Potts}~and~\reffmain{sec:application}, and which consists of a latent \citet{Potts_1952} model and a data model. In Section~\ref{sec:MCMC:Gibbsupdates}, we describe a simple Gibbs sampler, while in Section~\ref{sec:MCMC:MH} we describe a joint Metropolis--Hastings (MH) scheme that allows for point-mass distributions of the observations, given the hidden process, which we proposed for the analysis of the Arctic sea-ice data in Section~\reffmain{sec:application}. 

Throughout, we use the notation for the hidden Potts model introduced in the main text: Let $\vec{Y} \equiv (Y_1, \dots, Y_n)^\tp$ and $\vec{Z} \equiv (Z_1, \dots, Z_n)^\tp$ denote the hidden labels and observed data, respectively; let $\vec{Y}_{\backslash i}$ and $\vec{Z}_{\backslash i}$ denote all pixel labels and data, respectively, excluding the $i$th pixel; let $\mathcal{N}_i$ denote the indices of the ``neighbors'' of pixel $i$, $\beta$ denote the Potts model's parameter, $\vec{\lambda}$ denote the parameters associated with the conditional distributions of the data, and $\vec{\theta} \equiv (\beta, \vec{\lambda}^\tp)^\tp$. 

\subsection{Simple Gibbs updates}\label{sec:MCMC:Gibbsupdates}

In the standard Gibbs sampler, for each $i = 1, \dots, n,$ we alternate between sampling from the following full conditionals:
\begin{align}
 &p(Y_i \mid \vec{Y}_{\backslash i}, \vec{Z}, \vec{\theta}), \label{eqn:Yi}\\
 &p(Z_i \mid \vec{Y}, \vec{Z}_{\backslash i}, \vec{\theta}). \label{eqn:Zi}
\end{align}
By conditional independence, \eqref{eqn:Zi} simplifies to
\begin{equation}\label{eqn:gibbsZi:conditionalindep}
p(Z_i \mid \vec{Y}, \vec{Z}_{\backslash i}, \vec{\theta}) 
= p(Z_i \mid Y_i, \vec{\lambda}),
\end{equation}
which is straightforward to sample from by construction.  
Next, since $Y_i$ is discrete, we can sample \eqref{eqn:Yi} by enumerating all possible labels $y = 1, \dots, Q$ and normalizing their unnormalized probabilities. Specifically,
\begin{align*}
p(Y_i = y \mid \vec{Y}_{\backslash i}, \vec{Z}, \vec{\theta}) 
&\propto p(Y_i = y, \vec{Y}_{\backslash i}, \vec{Z} \mid \vec{\theta}) \\
&= p(\vec{Y} \mid \beta) \prod_{j=1}^n p(Z_j \mid Y_j, \vec{\lambda}) \,\bigg|_{Y_i = y} \\
&\propto \exp\!\big\{\beta T(\vec{Y})\big\} p(Z_i \mid Y_i, \vec{\lambda}) \,\bigg|_{Y_i = y} \\
&\propto \exp\!\left\{\beta S(y, i)\right\} p(Z_i \mid Y_i = y, \vec{\lambda}),
\end{align*}
where 
$T(\vec{Y}) \equiv \sum_{\{i', j'\} \in \mathcal{E}} \mathbb{I}(Y_{j'} = Y_{i'})$ 
is the Potts model's sufficient statistic, $\mathcal{E}$ denotes the set of undirected edges between neighbors, and 
$S(y, i) \equiv \sum_{j \in \mathcal{N}_i} \mathbb{I}(Y_j = y)$ 
is the number of neighbors of site $i$ with label $y$.  
Normalizing over all possible labels gives
\begin{equation}\label{eqn:gibbsYi}
p(Y_i = y \mid \vec{Y}_{\backslash i}, \vec{Z}, \vec{\theta})
= \frac{
\exp\!\left\{\beta S(y, i) \right\} p(Z_i \mid Y_i = y, \vec{\lambda})
}{
\sum\limits_{y'=1}^Q 
\exp\!\left\{\beta S(y', i) \right\} p(Z_i \mid Y_i = y', \vec{\lambda})
}.
\end{equation}
When the conditional distribution of the data, $p(Z_i \mid Y_i = y, \vec{\lambda})$, is a point mass for some label $y$, the corresponding pair $\{Y_i, Z_i\}$ becomes an absorbing state in the Markov chain, violating irreducibility (i.e., the chain cannot move between all states with positive probability).  
 Next, we discuss how this problem can be circumvented. 

\subsection{Joint MH updates}\label{sec:MCMC:MH}

To avoid absorbing states when sampling from a hidden Potts model using MCMC, each pair $\{Y_i, Z_i\}$, $i = 1, \dots, n$, is treated as a single unit and updated jointly via a MH update step, with carefully chosen proposals that ensure the acceptance probabilities do not depend on the conditional distributions of the data. Specifically, at each iteration and for each pair $\{Y_i, Z_i\}$, the target distribution is:
\begin{align*}
p(Y_i, Z_i \mid \vec{Y}_{\backslash i}, \vec{Z}_{\backslash i}, \vec{\theta})
&= p(Y_i \mid \vec{Y}_{\backslash i}, \vec{Z}_{\backslash i}, \vec{\theta}) \, p(Z_i \mid Y_i, \vec{Y}_{\backslash i}, \vec{Z}_{\backslash i}, \vec{\theta}) \\
&= p(Y_i \mid \vec{Y}_{\backslash i}, \beta) \, p(Z_i \mid Y_i, \vec{\lambda}) \\
&\propto \exp\!\left\{\beta S(Y_i, i)\right\} p(Z_i \mid Y_i, \vec{\lambda}).
\end{align*}
Given that the current value of the pair $\{Y_i, Z_i\}$ is $\{y, z\}$, we propose a new value $\{y', z'\}$ from a proposal distribution $q\big(y', z' \mid y, z\big)$, and we accept this proposal with probability
\begin{align*}
\alpha
&\equiv \min\!\left\{ 1,\,
\frac{p(Y_i = y', Z_i = z' \mid \vec{Y}_{\backslash i},\vec{Z}_{\backslash i},\vec{\theta})}
{p(Y_i = y, Z_i = z \mid \vec{Y}_{\backslash i},\vec{Z}_{\backslash i},\vec{\theta})}
\cdot
\frac{q(y, z \mid y', z')}
{q(y', z' \mid y, z)}
\right\}\\
&= \min\!\left\{ 1,\,
\frac{ \exp\!\left\{\beta S(y', i)\right\}
p(Z_i = z' \mid Y_i = y', \vec{\lambda})}
{\exp\!\left\{\beta S(y, i)\right\}
p(Z_i = z \mid Y_i = y, \vec{\lambda})}
\cdot
\frac{q(y, z \mid y', z')}
{q(y', z' \mid y, z)}
\right\}.
\end{align*}
Suppose that we take
$$
q\big(y', z' \mid y, z\big) = g(y' \mid y)\, r(z' \mid y'),
$$
where $g(y' \mid y)$ is a proposal distribution for the new label given the current label, and $r(z' \mid y')$ is a proposal distribution for the new observation given the proposed label. 
Then the MH acceptance probability becomes
\begin{align*}
\alpha
&= \min\!\left\{ 1,\,
\frac{ \exp\!\left\{\beta S(y', i)\right\}
p(Z_i = z' \mid Y_i = y', \vec{\lambda})}
{\exp\!\left\{\beta S(y, i)\right\}
p(Z_i = z \mid Y_i = y, \vec{\lambda})}
\cdot
\frac{g(y \mid y')\, r(z \mid y)}
{g(y' \mid y)\, r(z' \mid y')}
\right\}.
\end{align*}
If we choose $r(z' \mid y') = p(Z_i = z' \mid Y_i = y', \vec{\lambda})$ (i.e., propose the new observation directly from its conditional distribution given the proposed label), the observation-distribution terms 
cancel. 
If, in addition, $g(y' \mid y)$ is symmetric, 
 the label-proposal terms cancel as well, leaving
\begin{equation*}\label{eqn:acceptance_probability}
\alpha = \min\!\left\{ 1,\,\exp\!\left[\beta \{S(y', i) - S(y, i)\}\right]\right\}.
\end{equation*}
Thus, under these choices, the MH acceptance probability depends only on the parameter $\beta$ and the difference between the local summary statistics of the proposed and current labels, $S(y', i) - S(y, i)$. In particular, it does not depend on the conditional distribution of the data and can therefore be used to sample from a hidden Potts model that has conditional point-mass distributions in the data model. 

\section{Additional figures}\label{sec:additionalfigures}

\begin{figure}[!htb]
	\centering 
	\includegraphics[width = 0.95\linewidth]{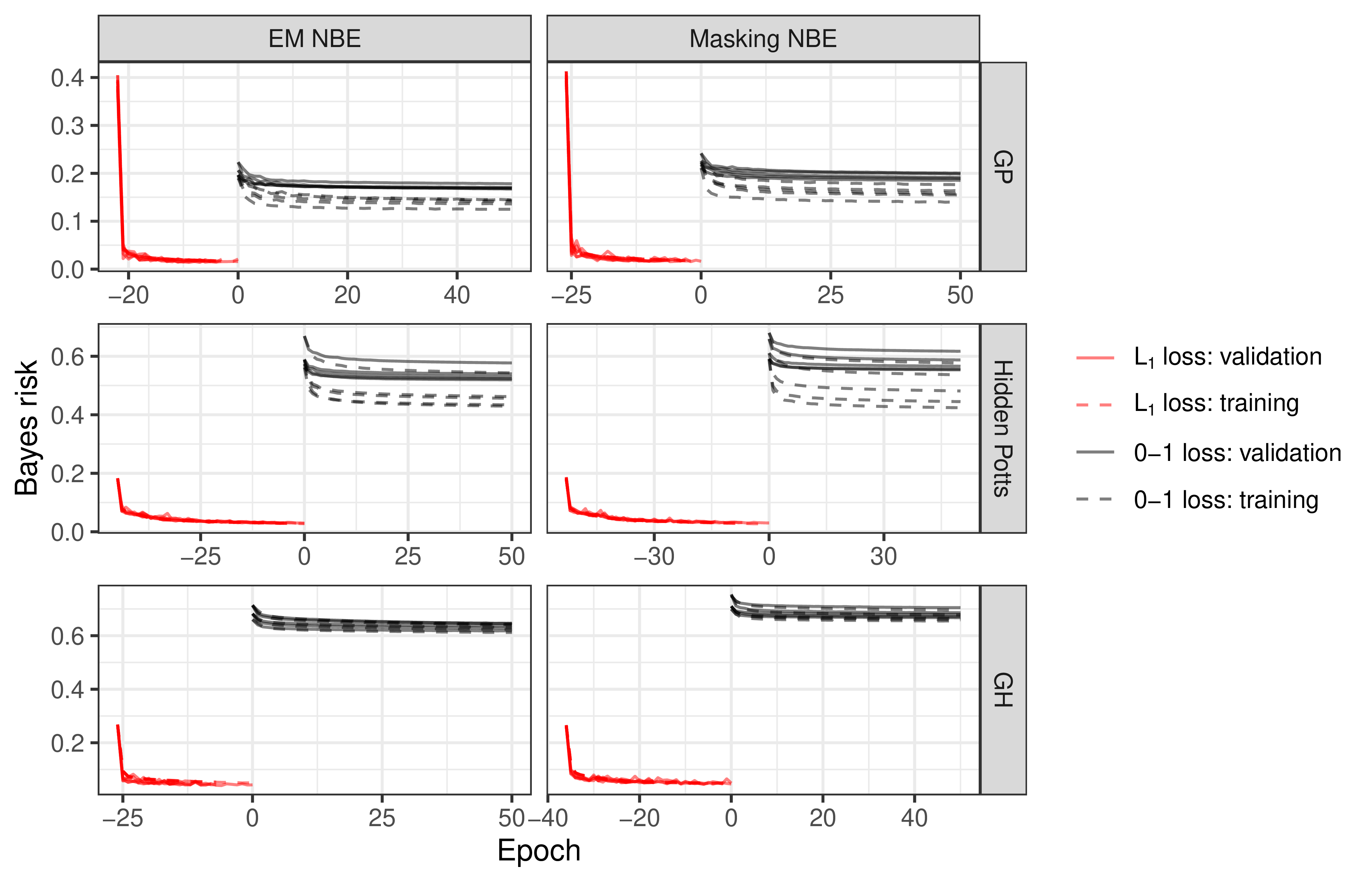}  
	\caption{Empirical risk function (objective function in Equation~\eqreff{eqn:optimization_task_encoding} or~\eqreff{eqn:optimization_task:adjustedprior_bayes} of the main text) evaluated on the training data (dashed lines) and withheld validation data (solid lines) versus epoch, for an ensemble of five NBEs trained for the EM (first column) and masking (second column) approaches, across all models considered in our simulation experiments: the Gaussian process model (Section~\reffmain{sec:GP}), the hidden Potts model (Section~\reffmain{sec:Potts}), and the spatial GH distribution (Section~\ref{sec:additionalsimulationstudies}). Red lines correspond to the Bayes risk under the $L_1$ (mean-absolute-error) loss function, used to pretrain the networks, while gray lines correspond to the Bayes risk under the continuous approximation of the 0--1 loss function in \eqreffmain{eqn:surrogateloss}, with tuning parameter $\kappa = 0.1$. Negative epoch values denote the pretraining stage, with epoch 0 corresponding to the transition to the training stage under the continuous approximation of the 0--1 loss function.}\label{fig:risk_profile}
\end{figure}

\begin{figure}[t!]
	\centering 
	\includegraphics[width = 0.9\linewidth]{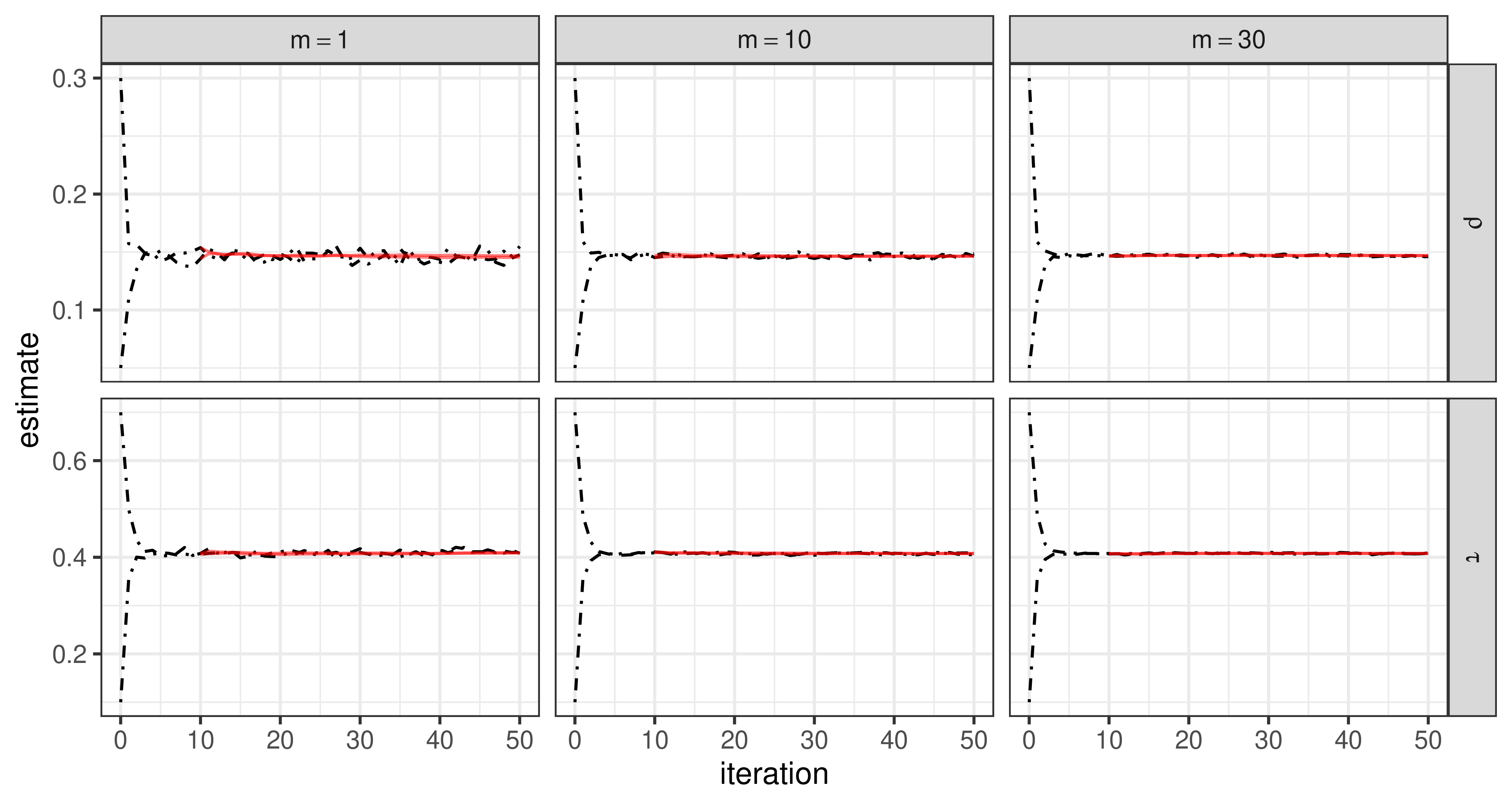}  
	\caption{\convergencecaption{Gaussian process model (Section~\reffmain{sec:GP})} The data are MCAR with 20\% missingness. The true parameter values are $\rho = 0.15$ and $\tau = 0.4$.}\label{fig:convergence:GP}
\end{figure}

\begin{figure}[t!]
	\centering 
	\includegraphics[width = 0.9\linewidth]{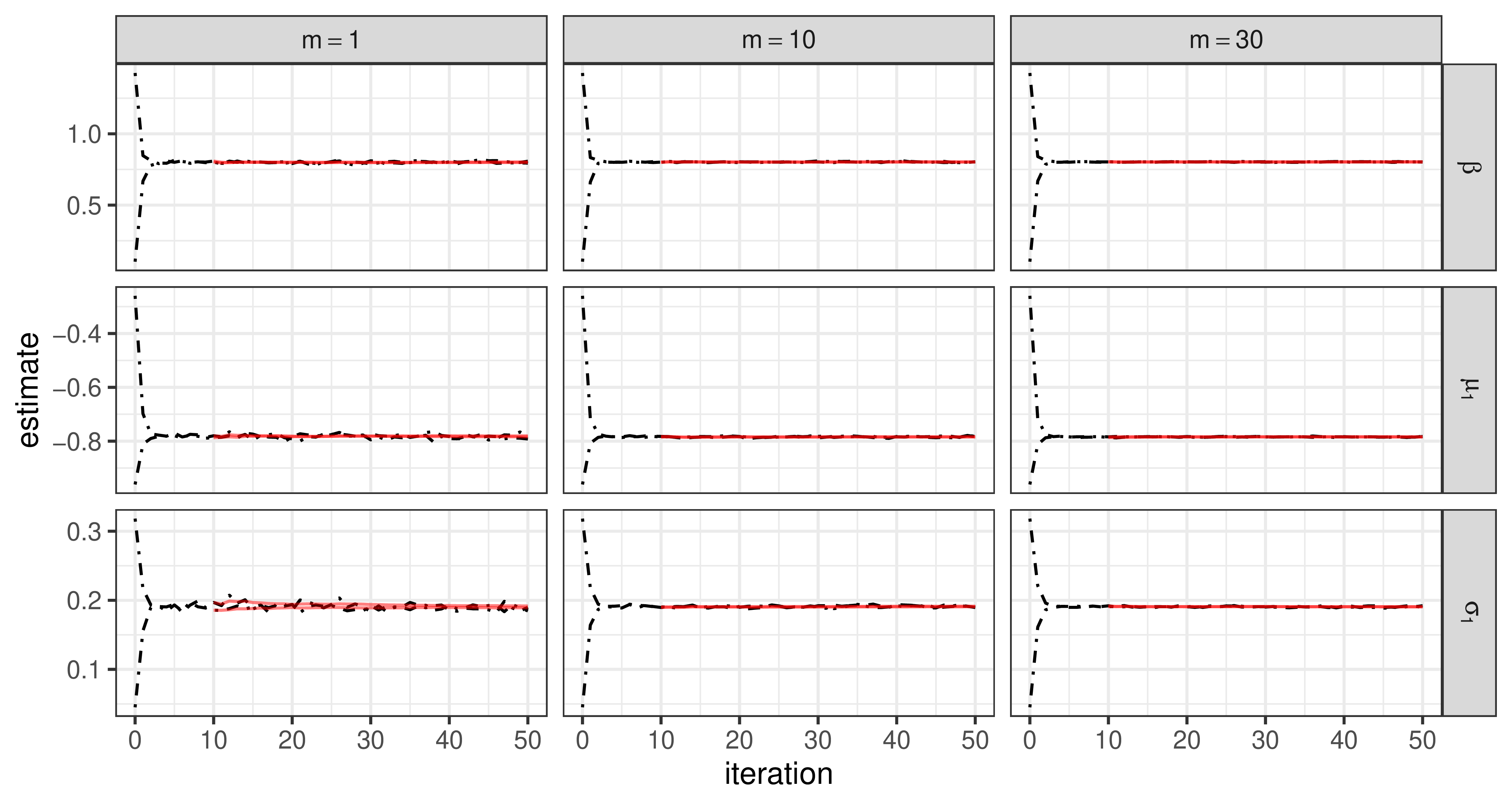}  
	\caption{\convergencecaption{hidden Potts model (Section~\reffmain{sec:Potts})} For simplicity, just three of the seven parameters are shown, where the true parameter values are $\beta = 0.79$, $\mu_1 = -0.79$, and $\sigma_1 = 0.19$.}\label{fig:convergence:Potts}
\end{figure}

\begin{figure}[t!]
	\centering 
	\includegraphics[width = \linewidth]{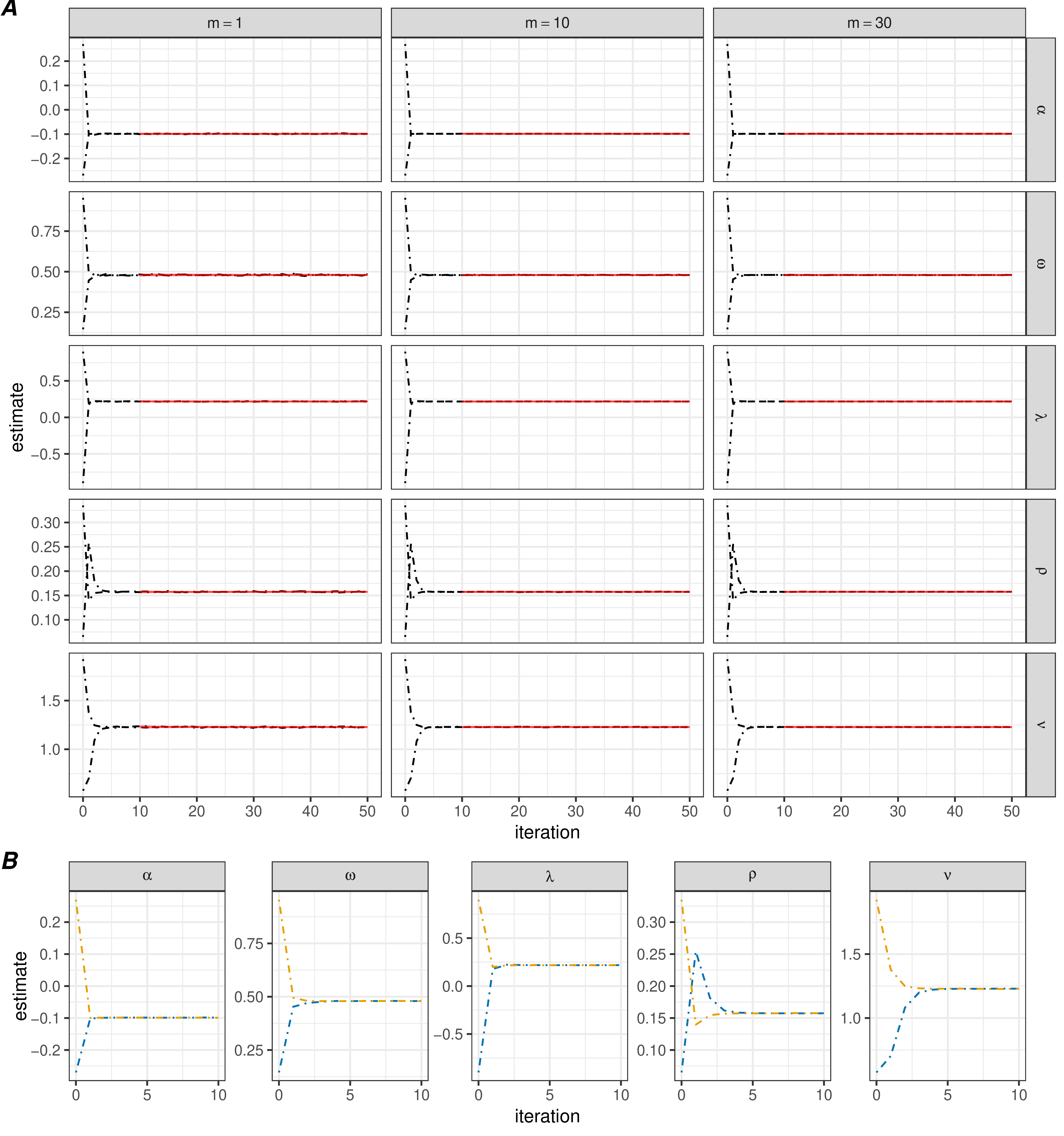}  
	\caption{(A) \convergencecaption{GH distribution (Section~\ref{sec:additionalsimulationstudies})} The data are MCAR with 20\% missingness. (B) EM NBE sequences computed with $m=30$ and colored by their initial parameter values. 
  The true parameter values are $\alpha = -0.122$, $\omega = 0.373$, $\lambda = 0.195$, $\rho = 0.160$, and $\nu = 1.24$.}\label{fig:convergence:GH}
\end{figure}

\begin{figure}[t!]
\centering
\includegraphics[width = \linewidth]{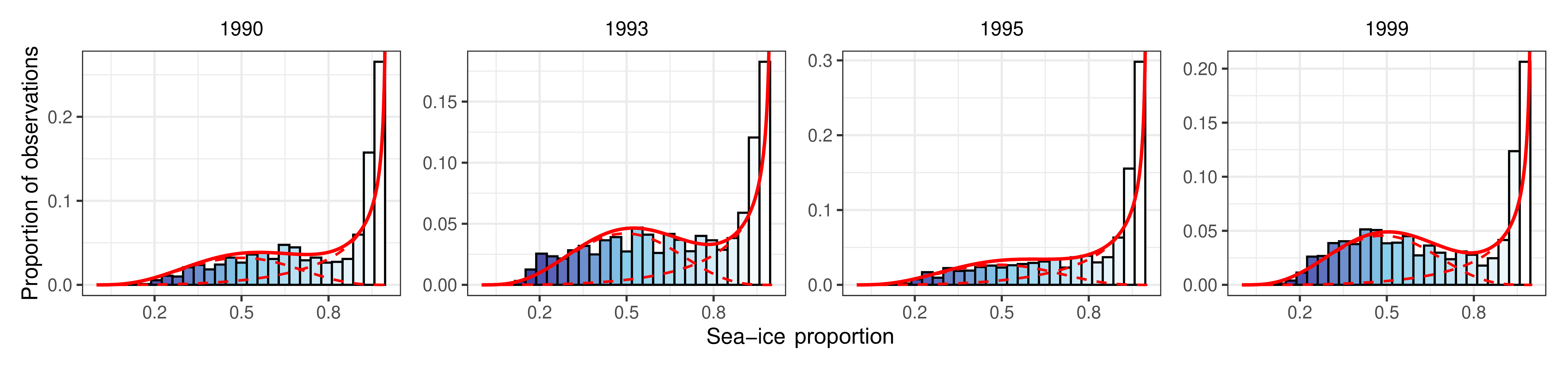}  
\caption{Histograms of observed Arctic sea‐ice proportions, restricted to the open interval \mbox{(0, 1)}, together with fitted marginal Beta mixture densities for selected years. Red solid lines show the fitted marginal mixture densities, and red dashed lines show the individual Beta components. 
}\label{fig:sea_ice:histograms} 
\end{figure}

\begin{figure}[t!]
\centering
\includegraphics[width = \linewidth]{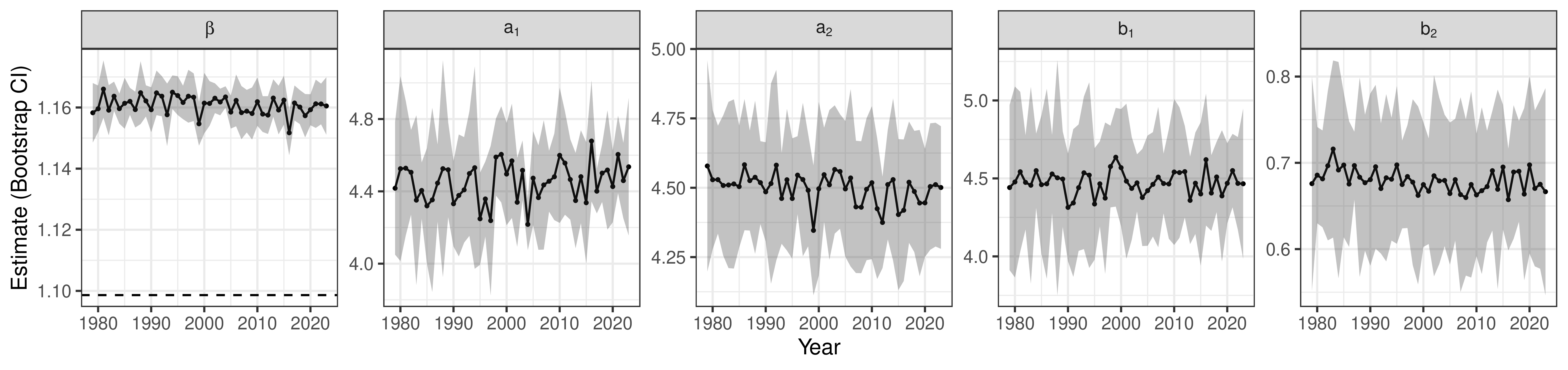}  
\caption{Estimated parameters of the hidden Potts model (described in Section~\reffmain{sec:application}) for Arctic sea‐ice proportion over time. Lines show point estimates with shaded regions indicating bootstrap confidence intervals. Headings above each panel give the model parameter. The dashed horizontal line in the first panel marks the critical value of the spatial-dependence parameter $\beta$, separating disordered and ordered spatial regimes.}\label{fig:sea_ice:all_parameters} 
\end{figure}

\begin{figure}[t!]
	\centering
	\includegraphics[width = \linewidth]{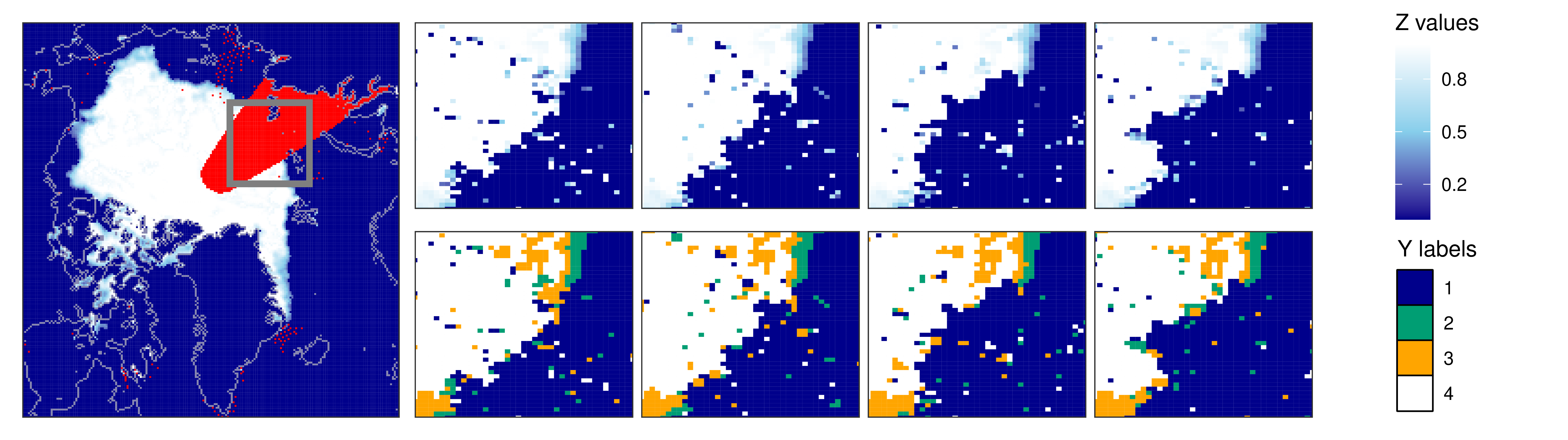}
	\caption{(Left) Arctic sea-ice data from September 1, 1995, with missing pixels colored red and faint gray lines denoting coastlines, with Greenland appearing at the bottom. (Remaining panels) conditional simulations for all grid cells within the gray box of the left panel. The conditional simulations are obtained from the fitted hidden Potts model described in Section~\reffmain{sec:application}, where the hidden Potts model has $Q=4$ labels. Conditional on the label $Y_i$ (bottom four panels of conditional simulations), the observation $Z_i$ (top four panels of conditional simulations) follows one of two Beta distributions for proportions in the open interval $(0, 1)$ ($Y_i = 2, 3$), or a point mass at $0$ or $1$ corresponding to no sea ice ($Y_i = 1$) and full sea ice ($Y_i = 4$), respectively. 
	}\label{fig:sea_ice:sims} 
\end{figure} 

\begin{figure}[t!]
	\centering
	\includegraphics[width = \linewidth]{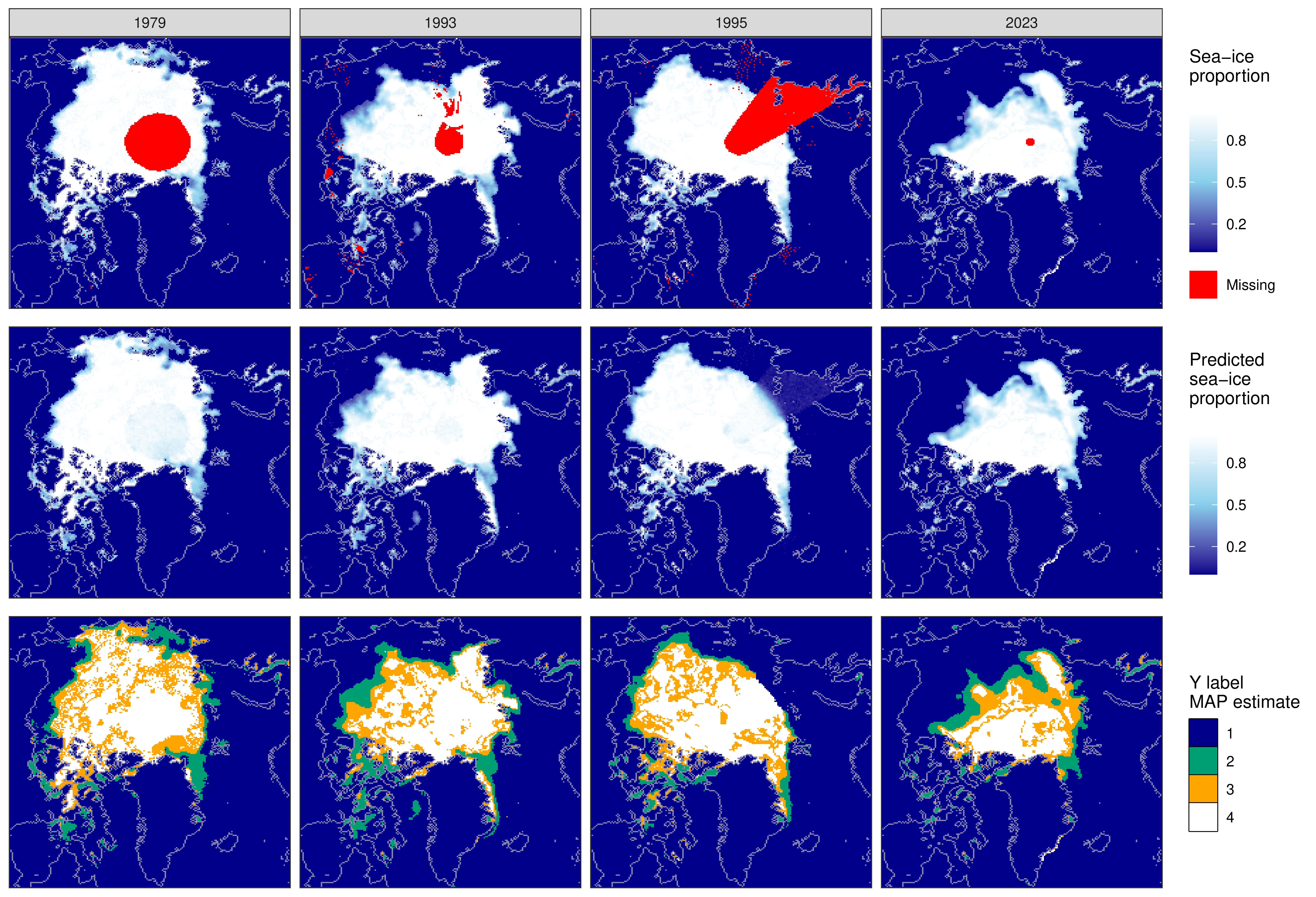}
	\caption{
		Arctic sea-ice data (first row) from September 1 for the years 1979, 1993, 1995, and 2023, together with corresponding predictions of sea-ice proportion and hidden states (second and third rows) from the fitted hidden Potts model described in Section~\reffmain{sec:application}. Faint gray lines denote coastlines, with Greenland appearing at the bottom.
	}\label{fig:sea_ice:predictions_and_Y} 
\end{figure} 

\begin{figure}[htbp]
    \centering
    
    \begin{subfigure}[b]{\linewidth}
        \centering
        \includegraphics[width=0.95\linewidth]{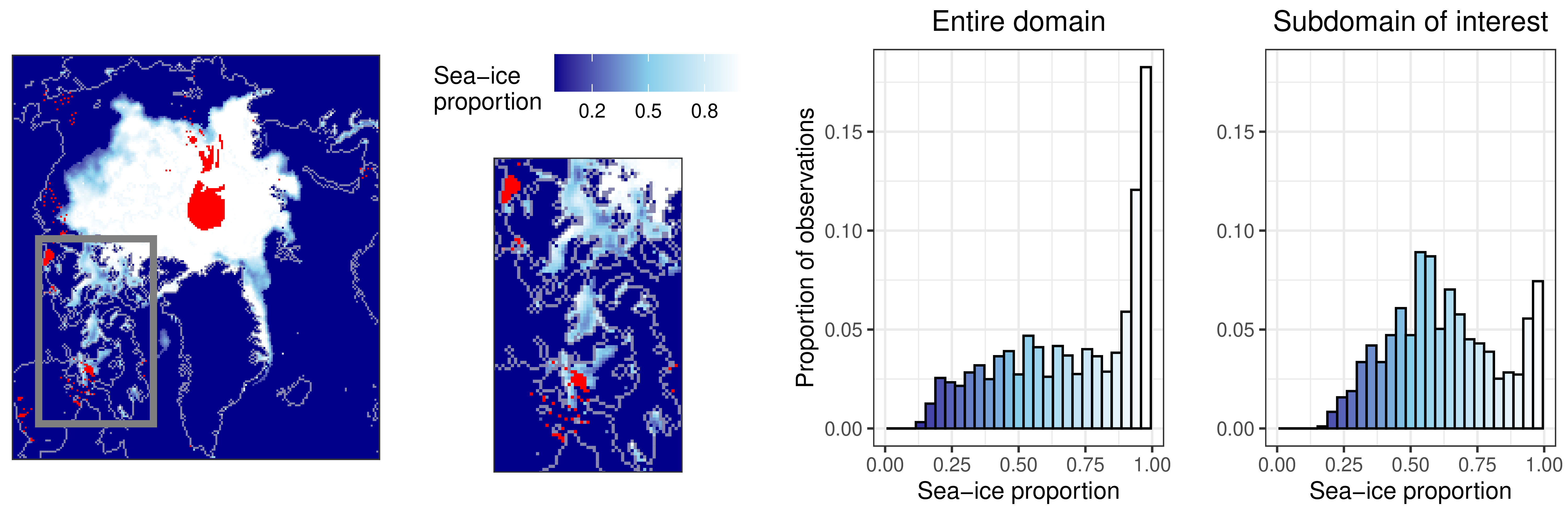}
        \caption{Arctic sea‐ice proportions from September 1, 1993. The left panel shows the full Arctic domain, with the gray box indicating the subdomain of interest shown in the middle panel, corresponding to the Canadian Arctic Archipelago west of Greenland. Faint gray lines denote coastlines, with Greenland appearing at the bottom. The right panels display histograms of sea‐ice proportions restricted to the open interval \mbox{(0, 1)}, computed over the entire domain and over the subdomain of interest, respectively.}
        \label{fig:plot1}
    \end{subfigure}
    
	\vspace{0.5cm} 
    \begin{subfigure}[b]{\linewidth}
        \centering
        \includegraphics[width=\linewidth]{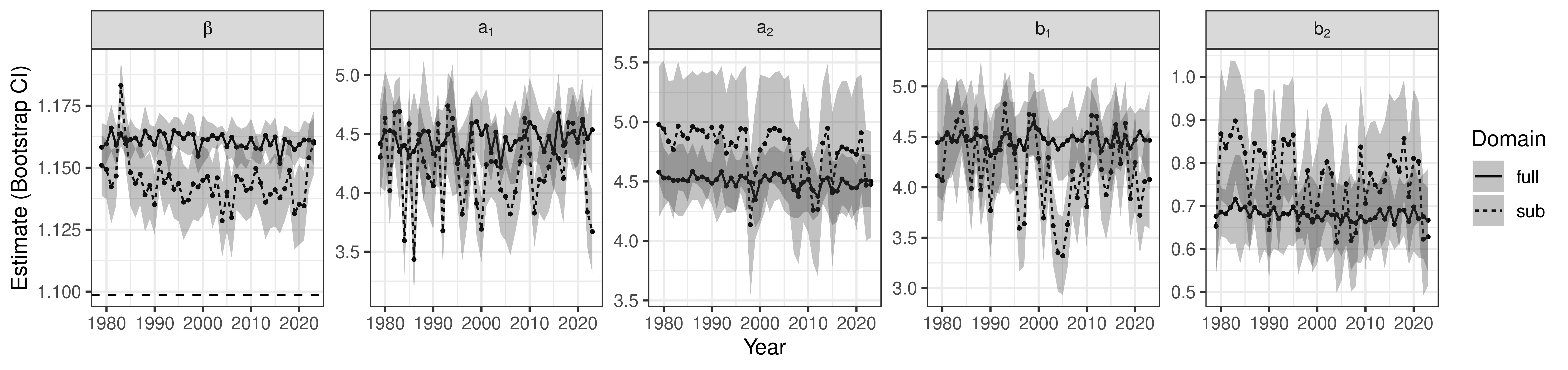}
        \caption{Estimated parameters of the hidden Potts model for Arctic sea‐ice proportion over time, and over both the full domain (solid lines) and the subdomain of interest (dashed lines). Headings above each panel give the model parameter.}
        \label{fig:plot2}
    \end{subfigure}
    
	\vspace{0.5cm} 
    \begin{subfigure}[b]{\linewidth}
        \centering
        \includegraphics[width=\linewidth]{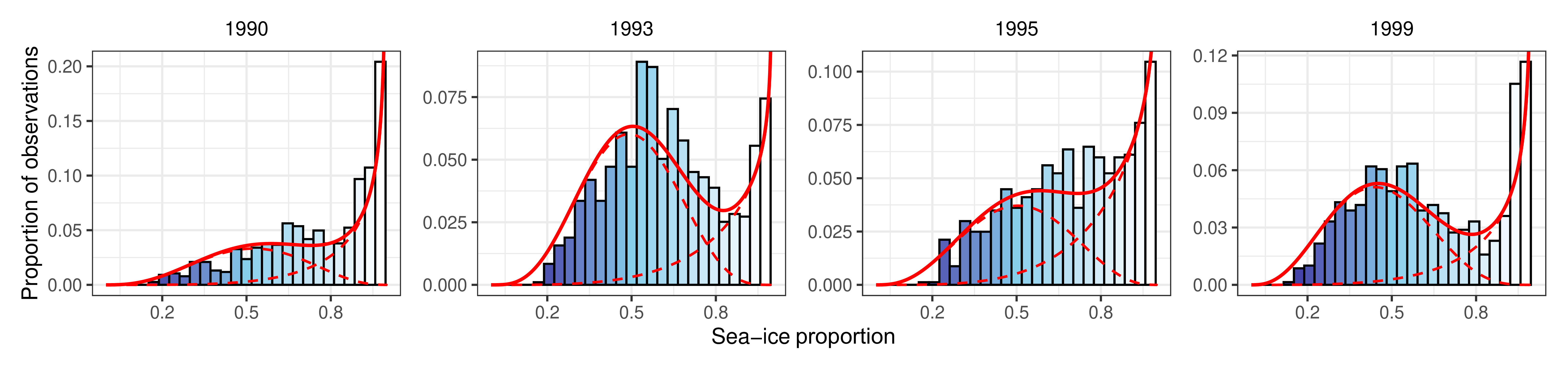}
        \caption{
		Histograms of observed Arctic sea‐ice proportions in the subdomain of interest, restricted to the open interval \mbox{(0, 1)}, together with fitted Beta mixture densities for selected years (see Figure~\ref{fig:sea_ice:histograms} for the corresponding plots over the entire domain). Red solid lines show the overall fitted mixture densities, and red dashed lines show the individual Beta components. 
		}
        \label{fig:plot3}
    \end{subfigure}
    
    \caption{Evidence of spatial nonstationarity in the Arctic sea-ice proportions analyzed in Section~\reffmain{sec:application}.}
    \label{fig:sea_ice_nonstationarity}
\end{figure}

\clearpage
{\small \putbib[bibliography]}
\end{bibunit}

}{}

\end{document}